\documentclass[11pt]{article}

\usepackage{amsmath,amsthm,amsfonts,amssymb,srcltx,empheq}
\usepackage[dvipdfmx]{graphicx}
\usepackage[all]{xy}
\usepackage{color}
\usepackage{longtable}
\usepackage{hhline}
\usepackage{url}

\usepackage{caption}
\def\IZ{\mathbb {Z}}

\def\IR{\mathbb {R}}
\def\IC{\mathbb {C}}

\def\rank{\operatorname{rank}}

\def\bu{\boldsymbol{u}}
\def\bU{\boldsymbol{U}}

\def\bd{\boldsymbol{d}}
\def\bz{\boldsymbol{z}}

\def\bgamma{\boldsymbol{\gamma}}

\def\bsig{\boldsymbol{\sigma}}

\def\bxi{\boldsymbol{\xi}}

\def\bQ{\boldsymbol{Q}}
\def\bR{\boldsymbol{R}}

\def\bk{\boldsymbol{k}}

\def\hcZ{\widehat{\mathcal{Z}}}
\def\cK{\mathcal{K}}

\def\mfR{\mathfrak{R}}
\def\mfg{\mathfrak{g}}
\def\mfh{\mathfrak{h}}
\def\sfr{\mathsf{r}}

\def\inc#1{%
\setbox1\hbox{\includegraphics[scale=.200]{#1}}%
\vcenter{\box1}%
}
\def\incc#1{%
\setbox1\hbox{\includegraphics[scale=.240]{#1}}%
\vcenter{\box1}%
}

\newcommand{\qpoch}[3]{\left(#1;#2\right)_{#3}}
\newcommand{\qq}[2]{\left(#1\right)_{#2}}
\newcommand{\knot}[2]{\mathbf{#1}_{\mathbf{#2}}}

\theoremstyle{plain}
  \newtheorem{prob}{Problem}[section]

  \newtheorem{prop}[prob]{Proposition}
   \newtheorem{cor}[prob]{Corollary}

\theoremstyle{remark}
  \newtheorem{remark}[prob]{\bf Remark}

\renewcommand{\thefootnote}{\fnsymbol{footnote}}

\makeatletter
\def\Left#1#2\Right{\begingroup%
   \def\ts@r{\nulldelimiterspace=0pt \mathsurround=0pt}%
   \let\@hat=#1%
   \def\sht@im{#2}%
   \def\@t{{\mathchoice{\def\@fen{\displaystyle}\k@fel}%
          {\def\@fen{\textstyle}\k@fel}%
          {\def\@fen{\scriptstyle}\k@fel}%
          {\def\@fen{\scriptscriptstyle}\k@fel}}}%
   \def\g@rin{\ts@r\left\@hat\vphantom{\sht@im}\right.}%
   \def\k@fel{\setbox0=\hbox{$\@fen\g@rin$}\hbox{%
      $\@fen \kern.3875\wd0 \copy0 \kern-.3875\wd0%
      \llap{\copy0}\kern.3875\wd0$}}%
      \def\pt@h{\mathopen\@t}\pt@h\sht@im%
      \Right}%
\def\Right#1{\let\@hat=#1%
   \def\st@m{\mathclose\@t}%
   \st@m\endgroup}
\makeatother

\makeatletter
 \renewcommand{\theequation}{%
       \thesection.\arabic{equation}}
 \@addtoreset{equation}{section}
\makeatother

\makeatletter
\def\eqnarray{%
 \stepcounter{equation}%
 \let\@currentlabel=\theequation
 \global\@eqnswtrue
 \global\@eqcnt\z@
 \tabskip\@centering
 \let\\=\@eqncr
 $$\halign to \displaywidth\bgroup\@eqnsel\hskip\@centering
 $\displaystyle\tabskip\z@{##}$&\global\@eqcnt\@ne
 \hfil$\displaystyle{{}##{}}$\hfil
 &\global\@eqcnt\tw@$\displaystyle\tabskip\z@{##}$\hfil
 \tabskip\@centering&\llap{##}\tabskip\z@\cr}
\makeatother

\begin{document}
\begin{titlepage}

\begin{flushright}
YITP-21-116
\end{flushright}

\begin{center}
\vspace*{1cm}
{\Large \bf
The colored Jones polynomials as vortex partition functions}
\vskip 1.5cm
{\large
Masahide Manabe${}^{a}$\footnote[2]{masahidemanabe@gmail.com}, 
Seiji Terashima${}^{b}$\footnote[3]{terasima@yukawa.kyoto-u.ac.jp}, and 
Yuji Terashima${}^{c}$\footnote[4]{yujiterashima@tohoku.ac.jp}}
\vskip 1.0em
{\it 
${}^a$%
School of Mathematics and Statistics, University of Melbourne
\footnote[1]{Affiliation until December 2020}, \\
Royal Parade, Parkville, VIC 3010, Australia \\

${}^b$%
Yukawa Institute for Theoretical Physics, Kyoto University, \\
Kyoto 606-8502, Japan \\

${}^c$%
Graduate School of Science, Tohoku University, \\
Aramaki-aza-Aoba 6-3, Aoba-ku, Sendai, 980-8578, Japan \\}
\end{center}
\vskip2.5cm

\begin{abstract}
We construct 3D $\mathcal{N}=2$ abelian gauge theories on $\mathbb{S}^2 \times \mathbb{S}^1$ labeled by knot diagrams whose K-theoretic vortex partition functions, 
each of which is a building block of twisted indices, 
give the colored Jones polynomials of knots in $\mathbb{S}^3$. 
The colored Jones polynomials are obtained as the Wilson loop expectation values along knots in $SU(2)$ Chern-Simons gauge theories on $\mathbb{S}^3$, 
and then our construction provides an explicit correspondence between 3D $\mathcal{N}=2$ abelian gauge theories and 3D $SU(2)$ Chern-Simons gauge theories. 
We verify, in particular, the applicability of our constructions 
to a class of tangle diagrams of 2-bridge knots with certain specific twists.
\end{abstract}
\end{titlepage}


\renewcommand{\thefootnote}{\arabic{footnote}} \setcounter{footnote}{0}

\section{Introduction}

In supersymmetric quantum field theories, recent localization techniques enable us to 
obtain various exact results which are made accessible to researchers 
for understanding or proposing conjectural dualities and 
for providing various mathematical conjectures, etc. See e.g. \cite{Pestun:2016zxk} 
for a review and references therein. 
In this paper, we focus on the A-twisted partition function of 3D $\mathcal{N}=2$ gauge theory on $\mathbb{S}^2 \times_{q} \mathbb{S}^1$,  
with the $\Omega$-deformation parameter $\hbar=-\log q$, 
known as the twisted index and 
obtained exactly by Benini and Zaffaroni in \cite{Benini:2015noa}, 
where the twisted partition function is factorized into 
the K-theoretic vortex partition functions 
which are considered as 
the partition functions on $\mathbb{D}^2 \times_{q} \mathbb{S}^1$ 
with appropriate boundary conditions
\cite{Pasquetti:2011fj, Dimofte:2011py, Beem:2012mb, Cecotti:2013mba, Fujitsuka:2013fga, Benini:2013yva, Benini:2015noa, Nieri:2015yia, Hwang:2017kmk, Crew:2020jyf}. 


The main object of this paper is to construct 
3D $\mathcal{N}=2$ abelian gauge theories $T[\cK]$ for knots $\cK$, 
and show that the colored Jones polynomials of $\cK$ in $\mathbb{S}^3$, 
which can be understood as the Wilson loop expectation values along $\cK$ 
in $SU(2)$ Chern-Simons gauge theories on $\mathbb{S}^3$ \cite{Witten:1988hf}, 
are obtained as the K-theoretic vortex partition functions of $T[\cK]$.
Therefore, our construction gives a new correspondence between 3D $\mathcal{N}=2$ abelian gauge theories and 3D $SU(2)$ Chern-Simons gauge theories:
$$
\textrm{vortex partition function of }T[\cK]=\textrm{Wilson loop along $\cK$ in $SU(2)$ Chern-Simons theory. }
$$
This reminds of the 3D-3D correspondence \cite{Terashima:2011qi, Terashima:2011xe, Dimofte:2011jd, Dimofte:2011ju, Cecotti:2011iy, Dimofte:2011py, Fuji:2012pi, Chung:2014qpa} (see also \cite{Dimofte:2010tz}) 
which is a 3D-3D analogue of the Alday-Gaiotto-Tachikawa
(4D-2D) correspondence \cite{Gaiotto:2009we, Alday:2009aq}, 
and says that 
the compactification of the 6D (2,0) theory of type $A_1$ twisted along a 3-manifold $M_3$ implies that 
a 3D $\mathcal{N}=2$ abelian gauge theory $T[M_3]$ labeled by $M_3$ 
is related to 
a 3D $SL(2,{\IC})$ Chern-Simons gauge theory on $M_3$.%
\footnote{Although the relation between the gauge theories $T[\cK]$ in this paper and the gauge theories $T[M_3]$ in the context of the 3D-3D correspondence is not clear, we use the same notation $T[\cK]$ to denote 
our gauge theories.} 
Remark that our correspondence treats Wilson loops in $SU(2)$ Chern-Simons theories which are
more manageable than $SL(2,{\IC})$ Chern-Simons theories.


More specifically, we propose how the colored Jones polynomials of knots, 
associated with the quantum group $U_q(\mathfrak{sl}_2)$, 
are constructed as K-theoretic vortex partition functions, 
obtained from a factorization of the twisted indices
of abelian gauge theories on $\mathbb{S}^2 \times_{q} \mathbb{S}^1$,
where the deformation parameter $q$ of $U_q(\mathfrak{sl}_2)$ 
is identified with the $\Omega$-deformation parameter.
Our strategy is to construct the building blocks of 
the colored Jones polynomial, given by the $R$-matrix etc. assigned to 
a tangle diagram of a knot \cite{Turaev:1988eb} 
(see \cite{Murakami:2010} for a survey), 
as building blocks of a K-theoretic vortex partition function. 
Then, for any knot diagram, we systematically associate 
a matter content and Chern-Simons couplings 
in 3D $\mathcal{N}=2$ abelian gauge theory 
whose K-theoretic vortex partition function, in a certain specific limit,  
gives the colored Jones polynomial of the knot. 
In this paper, we refer to the gauge theories $T[\cK]$ 
labeled by knot diagrams as \textit{knot-gauge theories}.
Here, for a knot, one can consider infinitely many tangles, 
related by the Reidemeister moves I, II and III, 
which provide the same colored Jones polynomial of the knot. 
As a result, for a knot $\cK$, we have infinitely many 3D $\mathcal{N}=2$ gauge theories $T[\cK]$ that are hopefully related to one another by some 3D dualities.

The localization formula in \cite{Benini:2015noa} 
of the A-twisted partition function is written as a middle-dimensional 
contour integral, in the space parametrized by the complex scalars 
(the real scalars in the vector multiplet and 
the holonomies of the gauge fields along $\mathbb{S}^1$), which is organized as 
the Jeffrey-Kirwan (JK) residue \cite{JeKi:1993} 
(see also \cite{BrVer:1999, SzVe:2003, Benini:2013xpa}). 
The JK residue depends on the choice of a vector (stability parameters), and 
the most technical part of our gauge theory constructions is how to choose 
the vector. 
In this paper, we show that there exists a choice of the vectors which gives
identifications between K-theoretic vortex partition functions and the
colored Jones polynomials for a class of tangle diagrams of 2-bridge knots with
certain specific twists.
We expect that, by a refinement of the choice, the identifications are also established for any knot diagram.%
\footnote{
The full twisted index with the boundary contribution in \eqref{3d_tw_pf} 
should not depend on the stability vector. 
In this paper, we ignore the boundary contribution by assuming that it is irrelevant to the vortex partition functions obtained for our choice of stability parameters. 
If the stability parameters are changed, the boundary contribution 
would need to be taken into account for obtaining the
colored Jones polynomials. 
This is a subtle point, and 
it may be interesting to investigate other choices of the 
stability vectors in the knot-gauge theories.}

It is expected that the K-theoretic vortex partition function of $T[\cK]$
is interpreted as a generating function of Euler characteristics for
moduli spaces of vortices \cite{Hwang:2017kmk, Bullimore:2018yyb, Bullimore:2018jlp, Crew:2020jyf}. 
Under our correspondence, it implies that we have a new
geometric interpretation of
the colored Jones polynomial. This will be an interesting mathematical
and physical problem.


This paper is organized as follows. In Section \ref{sec:k_v_pf}, the A-twisted partition function (twisted index) of 
3D $\mathcal{N}=2$ gauge theory on $\mathbb{S}^2 \times_{q} \mathbb{S}^1$ 
in \cite{Benini:2015noa} is recalled, 
and then the factorization into 
the K-theoretic vortex partition functions is provided. 
In Section \ref{sec:knot_building}, we first construct 
elementary constituents of knot-gauge theories which correspond to the building blocks of the colored Jones polynomials of knots as $R$-matrix. 
We then construct the knot-gauge theories $T[\cK]$, 
and discuss the JK residue procedure in detail. 
In particular, we show Proposition \ref{prop:bridge} for 
a class of 2-bridge knots, and exemplify the trefoil knot 
as well as the unknot. 
In Section \ref{sec:ded_gauge}, we construct yet another but simpler 
abelian gauge theories $T^{\text{red}}[\cK]$ labeled by knot diagrams referred to as \textit{reduced knot-gauge theories}, and exemplify the trefoil knot, 
the figure-eight knot and the 3-twist knot. 
In Appendix \ref{app:pochhammer}, properties of the $q$-Pochhammer symbol are summarized.

\section{K-theoretic vortex partition function}
\label{sec:k_v_pf}

In this section, we first recall the A-twisted partition function of 
3D $\mathcal{N}=2$ gauge theory on $\mathbb{S}^2 \times_{q} \mathbb{S}^1$ 
in \cite{Benini:2015noa},
with the $\Omega$-deformation parameter $\hbar=-\log q$, and 
then provide the building blocks of K-theoretic vortex partition function 
by factorizing the twisted partition function.%
\footnote{The $\mathcal{N}=2$ superconformal index on 
$\mathbb{S}^2 \times_{q} \mathbb{S}^1$ without the topological twist 
\cite{Kim:2009wb, Imamura:2011su, Kapustin:2011jm} also 
enjoys the similar factorization \cite{Beem:2012mb}.}

\subsection{Twisted partition function on $\mathbb{S}^2 \times_{q} \mathbb{S}^1$}

Consider a topologically twisted 3D $\mathcal{N}=2$ gauge theory 
on $\mathbb{S}^2 \times_{q} \mathbb{S}^1$ 
which consists of vector multiplet $V$ with $\rank(\mfg)$ Lie algebra $\mfg$ 
of a gauge group $G$ and 
$N$ chiral multiplets $\Phi_{\mfR_i}^{\sfr_i}$, 
$i=1,\ldots,N$, with representation $\mfR_i$ of $\mfg$, 
$U(1)_R$ charge $\sfr_i \in {\IZ}$, and 
(complexified) mass $\gamma_i=\bU_f^{\rho_{f,i}}$, 
where $\bU_f$ are the (complexified) holonomies 
(the real masses and the holonomies of the background gauge fields along $\mathbb{S}^1$) associated with a global symmetry $G_f$ and 
$\rho_{f,i}$ are the flavor weights.%
\footnote{Although, for 3D-3D correspondence discussed in this paper, 
it is enough to consider the cases with abelian symmetries, 
we will not assume it in this section.}
In \cite{Benini:2015noa}, by the supersymmetric localization \cite{Pestun:2016zxk}, 
the A-twisted partition function of the 3D gauge theory is obtained and 
written in terms of the following building blocks 
associated with the vector multiplet $V$ 
and each chiral multiplet $\Phi_{\mfR}^{\sfr}$ of 
mass $\gamma$ as
\begin{align}
&
\hcZ_{\bd}^{V}(\bU;q)=
\left(-q^{-\frac12}\right)^{\sum_{\alpha \in \Delta_+} \alpha(\bd)}\,
\prod_{\alpha \in \Delta_+} 
\left(1-\bU^{-\alpha} q^{\frac12 \alpha(\bd)}\right)
\left(1-\bU^{\alpha} q^{\frac12 \alpha(\bd)}\right),
\label{k_bb_1}
\\
&
\hcZ_{\bd}^{\Phi_{\mfR}^{\sfr}}(\bU;\gamma,q)=
\prod_{\rho \in \mfR}
\frac{\left(\gamma \bU^{\rho} \right)^{\frac{\rho(\bd)
+1-\sfr}{2}}}
{\qpoch{\gamma \bU^{\rho}q^{\frac{\sfr-\rho(\bd)}{2}}}{q}{\rho(\bd)+1-\sfr}},
\label{k_bb_2}
\end{align}
where $\Delta_+$ is the set of positive roots of $\mfg$, 
$\alpha(*)$ and $\rho(*)$ are the canonical pairings, 
and $\bU=(U_1,\ldots,U_{\rank(\mfg)})$, 
$U_a= \mathrm{e}^{-u_a}$, $\bU^{\alpha}= \mathrm{e}^{-\alpha(\bu)}$, 
$\bU^{\rho}= \mathrm{e}^{-\rho(\bu)}$. 
Here $\bu=(u_1,\ldots,u_{\rank(\mfg)})$, 
$u_a \in \mfh \otimes_{\IR}{\IC}$, are the complex scalars 
(the real scalars in the vector multiplet $V$ and 
the holonomies of the gauge fields along $\mathbb{S}^1$), 
where $\mfh$ is the Cartan subalgebra of $\mfg$.
These building blocks are indexed by 
the magnetic fluxes $\bd=(d_1,\ldots,d_{\rank(\mfg)})$, $d_a\in {\IZ}$, 
associated with $U(1)^{\rank(\mfg)} (\subset G)$ gauge fields on $\mathbb{S}^2$. 
The $q$-Pochhammer symbol $\qpoch{x}{q}{d}$ is defined in \eqref{qPoch}. 
When the gauge group contains central $U(1)^{\mathsf{c}} \subset G$ factors, 
the gauge theories admit deformations with the 3D complexified 
Fayet-Iliopoulos (FI) parameters $\tau^a$, $a=1,\ldots,\mathsf{c}$. 
Then the A-twisted partition function on $\mathbb{S}^2 \times_{q} \mathbb{S}^1$ 
without Chern-Simons factors 
(see Remark \ref{rem:cs_factor} for Chern-Simons factors) 
is given by \cite{Benini:2015noa}
\begin{align}
\begin{split}
&
Z_{\mathbb{S}^2 \times_{q} \mathbb{S}^1}(\bz,\bgamma,q) = \frac{1}{|\mathcal{W}|}
\sum_{\bd \in {\IZ}^{\rank(\mfg)}} 
\oint_{\Gamma} d^{\, \rank(\mfg)} u \, 
\hcZ_{\bd}^{\textrm{total}}(\bU;\bz,\bgamma,q)
+\textrm{boundary contribution},
\\
&
\hcZ_{\bd}^{\textrm{total}}(\bU;\bz,\bgamma,q):=
\bz^{\bd}\, \hcZ_{\bd}^{V}(\bU;q)\,
\prod_{i=1}^N \hcZ_{\bd}^{\Phi_{\mfR_i}^{\sfr_i}}(\bU;\gamma_i,q).
\label{3d_tw_pf}
\end{split}
\end{align}
Here $|\mathcal{W}|$ is the order of the Weyl group of $G$, and 
$\bz^{\bd}=\mathrm{e}^{2\pi \mathrm{i}\, \tau(\bd)}$, where 
the pairing $\tau(\bd)=\sum_a \tau^a d_a$ is defined by the embedding 
${\boldsymbol\tau} \hookrightarrow \mfh^* \otimes_{\IR}{\IC}$. 
The middle-dimensional contour integral along $\Gamma$ is defined by 
the JK residue \cite{JeKi:1993} 
(see also \cite{BrVer:1999, SzVe:2003, Benini:2013xpa}) which 
picks up relevant poles in $\hcZ_{\bd}^{\Phi_{\mfR_i}^{\sfr_i}}$ 
depending on the choice of a vector (stability parameters). 
In this paper, for simplicity we identify the stability parameters 
with the FI parameters $\xi_a=\textrm{Im}(\tau^a)$ 
(see \cite{Benini:2013xpa} for the difference between them). 
Although the partition function has a boundary contribution at 
$u_a = \pm \infty$, we assume that the boundary contribution is irrelevant to 
the vortex partition functions in this paper 
and only focus on the bulk (first) factor 
(see \cite{Bullimore:2020nhv} for an interpretation of the boundary contribution 
as \textit{topological saddles} of an effective supersymmetric quantum mechanics).

Note that, the background magnetic fluxes $\bd_f$ 
associated with the global symmetry $G_f$ are also introduced by the shifts 
in \eqref{k_bb_2} as
\begin{align}
\rho(\bd) \to \rho(\bd) + d_{f,i}, \qquad
d_{f,i}=\rho_{f,i}(\bd_f).
\label{bg_flux}
\end{align}

\begin{remark}\label{rem:cs_factor}
Assume that the gauge symmetry $G$ to be abelian 
(or consider abelian factors in $G$). 
The gauge/flavor-gauge/flavor/R Chern-Simons factors 
associated with $G$ and $G_f$ are introduced by
\cite{Benini:2015noa, Benini:2016hjo, Closset:2016arn} 
(see also \cite{Ueda:2019qhg})
\begin{align}
\begin{split}
&
\hcZ_{\bd}^{\text{g-g}}(\bU;\bk)=
\prod_{a,b} U_a^{k_{ab} d_b},
\quad
\hcZ_{\bd, \bd_f}^{\text{g-f}}(\bU;\bgamma,\bk^{\text{g-f}})=
\prod_{i, a}
\left(\gamma_i^{d_a} U_a^{d_{f,i}}\right)^{k_{a i}^{\text{g-f}}},
\\
&
\hcZ_{\bd_f}^{\text{f-f}}(\bgamma,\bk^{\text{f-f}})=
\prod_{i, j}
\gamma_i^{k_{ij}^{\text{f-f}} d_{f,j}},
\quad
\hcZ^{\text{g-R}}(\bU;\bk^{\text{g-R}})=
\prod_{a} U_a^{k_a^{\text{g-R}}},
\quad
\hcZ^{\text{f-R}}(\bgamma,\bk^{\text{f-R}})=
\prod_{i} \gamma_{i}^{k_{i}^{\text{f-R}}}.
\label{cs_factor}
\end{split}
\end{align}
Here $\bk$, $\bk^{\text{g-f}}$, $\bk^{\text{f-f}}$, 
$\bk^{\text{g-R}}$ and $\bk^{\text{f-R}}$ 
are, respectively, the set of relevant Chern-Simons couplings 
$k_{ab}=k_{ba}$, $k_{a i}^{\text{g-f}}$, $k_{ij}^{\text{f-f}}=k_{ji}^{\text{f-f}}$, 
$k_a^{\text{g-R}}$ and $k_{i}^{\text{f-R}}$.
\end{remark}

\subsection{K-theoretic vortex partition function by factorization}

Let us recall how the A-twisted partition function is factorized into 
K-theoretic vortex partition functions 
\cite{Benini:2015noa, Nieri:2015yia, Hwang:2017kmk, Crew:2020jyf}. 
For a choice of FI parameters, 
the contour $\Gamma$ in \eqref{3d_tw_pf} encloses the poles in 
$\hcZ_{\bd}^{\Phi_{\mfR}^{\sfr}}$, with the inclusion of 
the background magnetic fluxes by \eqref{bg_flux}, as 
\begin{align}
\gamma \bU^{\rho}=q^{-\frac12 \rho(\bd)+\frac12 \sfr + p},
\quad
p=-\frac12 \rho_f(\bd_f),
-\frac12 \rho_f(\bd_f)+1,\ldots,
\rho(\bd)+\frac12 \rho_f(\bd_f)-\sfr,
\label{poles_factor}
\end{align}
with $\rho(\bd)+\rho_f(\bd_f)-\sfr \ge 0$. By a change of variables
\begin{align}
U_a=q^{\frac{d'_a-d''_a}{2}}\sigma_a,\qquad
\bd=\bd' + \bd'',\qquad
\gamma=q^{\frac{\rho_f(\bd_f')-\rho_f(\bd_f'')}{2}}\widetilde{\gamma},\qquad
\bd_f=\bd_f'+\bd_f'',
\label{repara_vortex}
\end{align}
with
$\rho(\bd')=p+\sfr/2-\rho_f(\bd_f')/2+\rho_f(\bd_f'')/2$,
the poles \eqref{poles_factor} yield
\begin{align}
\widetilde{\gamma} \bsig^{\rho}=1,
\label{poles_factor_ref}
\end{align}
and the sum over the magnetic fluxes and the integration 
in \eqref{3d_tw_pf} for $\Phi_{\mfR}^{\sfr}$ are written as
\begin{align}
\sum_{\rho(\bd)+\rho_f(\bd_f) \ge \sfr} 
\sum_{p=-\frac12 \rho_f(\bd_f)}^{\rho(\bd)+\frac12\rho_f(\bd_f)-\sfr} 
\oint_{\gamma \bU^{\rho}=q^{-\frac12 \rho(\bd)+\frac12 \sfr + p}}
&=
\sum_{\rho(\bd)+\rho_f(\bd_f) \ge \sfr} \sum_{\rho(\bd')=\frac{\sfr}{2}-\rho_f(\bd_f')}^{\rho(\bd)+\rho_f(\bd_f'')-\frac{\sfr}{2}}
\oint_{\gamma \bsig^{\rho}=1}
\nonumber
\\
&=
\sum_{\rho(\bd')+\rho_f(\bd_f'), \rho(\bd'')+\rho_f(\bd_f'') \ge \frac{\sfr}{2}}
\oint_{\gamma \bsig^{\rho}=1}.
\label{pole_sum}
\end{align}
Here,  we can take $\bd_f'$ as an arbitrary integer satisfying $0 \leq
\bd_f' \leq \bd_f$.
Because the twisted index and the discussion below do not depend on this choice,
we will take, for simplicity, $\bd_f'=\bd_f''$ below, and then
$\widetilde{\gamma}=\gamma$.
Under the reparametrization \eqref{repara_vortex}, 
by using $\qpoch{x}{q}{d}=\qpoch{q^{d-1}x}{q^{-1}}{d}$ in \eqref{qPoch_prop1} 
and $\qpoch{x}{q}{d_1+d_2}
=\qpoch{x}{q}{d_1}\qpoch{q^{d_1}x}{q}{d_2}$ in \eqref{qPoch_prop2}, 
the building blocks \eqref{k_bb_1} and \eqref{k_bb_2} are factorized, respectively, as
\begin{align}
\hcZ_{\bd}^{V}(\bU;q)=
I_{\textrm{1-loop}}^{V}(\bsig)\,
I_{\bd'}^{V}(\bsig;q)\,
I_{\bd''}^{V}(\bsig;q^{-1}),
\end{align}
where
\begin{align}
\begin{split}
&
I_{\textrm{1-loop}}^{V}(\bsig)=
\prod_{\alpha \in \Delta_+} 
(-1)\left(\bsig^{\frac{\alpha}{2}}-\bsig^{-\frac{\alpha}{2}}
\right)^2,\\
&
I_{\bd}^{V}(\bsig;q)=
\prod_{\alpha \in \Delta_+} 
\left(-q^{-\frac12}\right)^{\alpha(\bd)}
\frac{1-\bsig^{\alpha}q^{\alpha(\bd)}}
{1-\bsig^{\alpha}}
=
\prod_{\alpha \in \Delta_+} 
\left(-q^{-\frac12}\right)^{\alpha(\bd)}
\frac{\qpoch{q \bsig^{\alpha}}{q}{\alpha(\bd)}}
{\qpoch{\bsig^{\alpha}}{q}{\alpha(\bd)}},
\label{bk_vector}
\end{split}
\end{align}
and
\begin{align}
\hcZ_{\bd}^{\Phi_{\mfR}^{\sfr}}(\bU;\gamma,q)=
I_{\textrm{1-loop}}^{\Phi_{\mfR}^{\sfr}}(\bsig;\gamma,q)\,
I_{\bd'}^{\Phi_{\mfR}^{\sfr}}(\bsig;\gamma,q)\,
I_{\bd''}^{\Phi_{\mfR}^{\sfr}}(\bsig;\gamma,q^{-1}),
\label{matter_factor}
\end{align}
where
\begin{align}
\begin{split}
&
I_{\textrm{1-loop}}^{\Phi_{\mfR}^{\sfr}}(\bsig;\gamma,q)=
\prod_{\rho \in \mfR} 
\frac{\left(\gamma \bsig^{\rho} \right)^{\frac{1-\sfr}{2}}}
{\qpoch{\gamma \bsig^{\rho} q^{\frac{\sfr}{2}}}{q}{1-\sfr}},
\\
&
I_{\bd}^{\Phi_{\mfR}^{\sfr}}(\bsig;\gamma,q)=
\prod_{\rho \in \mfR}
\frac{q^{\frac14 \rho(\bd) \left(\rho(\bd)+1-\sfr\right)}
\left(\gamma \bsig^{\rho} \right)^{\frac12 \rho(\bd)}}
{\qpoch{\gamma \bsig^{\rho} q^{1-\frac{\sfr}{2}}}{q}{\rho(\bd)}},
\label{bk_hyper}
\end{split}
\end{align}
and the background magnetic fluxes are now introduced 
by the similar shift to \eqref{bg_flux} 
(i.e. $\rho(\bd) \to \rho(\bd) + 2\rho_f(\bd_f)$ for $\hcZ_{\bd}^{\Phi_{\mfR}^{\sfr}}$ and 
$\rho(\bd) \to \rho(\bd) + \rho_f(\bd_f)$ for $I_{\bd}^{\Phi_{\mfR}^{\sfr}}$).
The twisted partition function \eqref{3d_tw_pf} 
is then factorized into
\begin{align}
I_{\textrm{1-loop}}(\bsig;\bgamma,q)=
I_{\textrm{1-loop}}^{V}(\bsig)\,
\prod_{i=1}^N I_{\textrm{1-loop}}^{\Phi_{\mfR_i}^{\sfr_i}}(\bsig;\gamma_i,q),
\end{align}
and the K-theoretic vortex partition function
\begin{align}
I_{\textrm{vortex}}(\bsig;\bz,\bgamma,q)=
\sum_{\bd} I_{\bd}(\bsig;\bz,\bgamma,q)=
\sum_{\bd} \bz^{\bd}\, I_{\bd}^{V}(\bsig;q)\,
\prod_{i=1}^N I_{\bd}^{\Phi_{\mfR_i}^{\sfr_i}}(\bsig;\gamma_i,q),
\label{vortex_I_def_K}
\end{align}
as \cite{Benini:2015noa, Nieri:2015yia, Hwang:2017kmk, Crew:2020jyf} 
(see also \cite{Fujitsuka:2013fga, Benini:2013yva})
\footnote{The factorization into the ``1-loop factor'' and the ``vortex factor'' has ambiguities, and so it is desirable to prescribe how to fix them. 
A decomposition of $\mathbb{S}^2 \times_{q} \mathbb{S}^1$ into two 
$\mathbb{D}^2 \times_{q} \mathbb{S}^1$ is known to lead to the factorization \cite{Beem:2012mb}, and 
the partition functions on $\mathbb{D}^2 \times_{q} \mathbb{S}^1$ with appropriate boundary conditions are expected to unambiguously provide 
the K-theoretic vortex partition functions 
\cite{Yoshida:2014ssa, Crew:2020psc, Bullimore:2020jdq}.
}
\begin{align}
Z_{\mathbb{S}^2 \times_{q} \mathbb{S}^1}(\bz,\bgamma,q) \sim 
\sum_{\bsig^*}
I_{\textrm{1-loop}}(\bsig^*;\bgamma,q)\,
I_{\textrm{vortex}}(\bsig^*;\bz,\bgamma,q)\,
I_{\textrm{vortex}}(\bsig^*;\bz,\bgamma,q^{-1}),
\label{z_factorization}
\end{align}
up to an overall normalization and the boundary contribution, where 
the shift \eqref{bg_flux} for 
each $I_{\bd}^{\Phi_{\mfR_i}^{\sfr_i}}$ introduces 
background magnetic fluxes. 
Here the domain of $\bd$ is determined as \eqref{pole_sum} 
for a choice of the FI parameters. 
The domain of $\bsig^*=\bsig(\bgamma)$ is determined as well.
Note that when $\rho(\bd)>0$, the poles \eqref{poles_factor_ref} 
for an $\sfr=0$ chiral multiplet $\Phi_{\mfR}^{0}$ are 
in the ``1-loop factor'' $I_{\textrm{1-loop}}^{\Phi_{\mfR}^{0}}$ for $\sfr=0$, 
whereas the poles \eqref{poles_factor_ref} 
for an $\sfr=2$ chiral multiplet $\Phi_{\mfR}^{2}$ are in 
the ``vortex factor'' $I_{\bd}^{\Phi_{\mfR}^{2}}$ for $\sfr=2$ 
(one zeros from $I_{\textrm{1-loop}}^{\Phi_{\mfR}^{2}}$ and 
two poles from $I_{\bd'}^{\Phi_{\mfR}^{2}}$ and $I_{\bd''}^{\Phi_{\mfR}^{2}}$ 
in \eqref{matter_factor}).

\begin{remark}
By \eqref{qPoch_prop1}, 
the factors in \eqref{bk_hyper} are rewritten as
\begin{align}
\begin{split}
&
I_{\textrm{1-loop}}^{\Phi_{\mfR}^{\sfr}}(\bsig;\gamma,q)=
\prod_{\rho \in \mfR} 
(-1)^{1-\sfr}
\left(\gamma \bsig^{\rho} \right)^{-\frac{1-\sfr}{2}}
\qpoch{\gamma^{-1} \bsig^{-\rho} q^{1-\frac{\sfr}{2}}}{q}{\sfr-1},
\\
&
I_{\bd}^{\Phi_{\mfR}^{\sfr}}(\bsig;\gamma,q)=
\prod_{\rho \in \mfR}
(-1)^{\rho(\bd)}
q^{-\frac14 \rho(\bd) \left(\rho(\bd)+1-\sfr\right)}
\left(\gamma \bsig^{\rho} \right)^{-\frac12 \rho(\bd)}
\qpoch{\gamma^{-1} \bsig^{-\rho} q^{\frac{\sfr}{2}}}{q}{-\rho(\bd)},
\label{bk_hyper_inv}
\end{split}
\end{align}
where $\bsig^{-1}=(\sigma_1^{-1},\ldots,\sigma_{\rank(\mfg)}^{-1})$. 
\end{remark}


\begin{remark}\label{rem:cs_vortex}
Similarly, we factorize the Chern-Simons factors in \eqref{cs_factor} as
\begin{align}
\begin{split}
&
\hcZ_{\bd}^{\text{g-g}}(\bU;\bk)=
I_{\bd'}^{\text{g-g}}(\bsig;q, \bk)\,
I_{\bd''}^{\text{g-g}}(\bsig;q^{-1}, \bk),
\\
&
\hcZ_{\bd, 2\bd_f}^{\text{g-f}}(\bU;\bgamma, \bk^{\text{g-f}})=
I_{\bd', \bd_f}^{\text{g-f}}(\bsig;\bgamma, q, \bk^{\text{g-f}})\,
I_{\bd'', \bd_f}^{\text{g-f}}(\bsig;\bgamma, q^{-1}, \bk^{\text{g-f}}),
\\
&
\hcZ_{2\bd_f}^{\text{f-f}}(\bgamma,\bk^{\text{f-f}})=
I_{\bd_f}^{\text{f-f}}(q, \bk^{\text{f-f}})\,
I_{\bd_f}^{\text{f-f}}(q^{-1}, \bk^{\text{f-f}}),
\\
&
\hcZ^{\text{g-R}}(\bU;\bk^{\text{g-R}})=
I_{\textrm{1-loop}}^{\text{g-R}}(\bsig;\bk^{\text{g-R}})\,
I_{\bd'}^{\text{g-R}}(q, \bk^{\text{g-R}})\,
I_{\bd''}^{\text{g-R}}(q^{-1}, \bk^{\text{g-R}}),
\\
&
\hcZ^{\text{f-R}}(\bgamma, \bk^{\text{f-R}})=
I_{\textrm{1-loop}}^{\text{f-R}}(\bgamma, \bk^{\text{f-R}})\,
I_{\bd_f}^{\text{f-R}}(q, \bk^{\text{f-R}})\,
I_{\bd_f}^{\text{f-R}}(q^{-1}, \bk^{\text{f-R}}),
\label{cs_factor_factor}
\end{split}
\end{align}
where
\begin{align}
I_{\textrm{1-loop}}^{\text{g-R}}(\bsig;\bk^{\text{g-R}})=
\prod_{a} \sigma_a^{k_a^{\text{g-R}}},
\quad
I_{\textrm{1-loop}}^{\text{f-R}}(\bgamma, \bk^{\text{f-R}})=
\prod_{i} \gamma_i^{k_i^{\text{f-R}}},
\label{cs_factor_1loop}
\end{align}
are considered to be normalization (1-loop) factors 
which are irrelevant to $d_a$ and $d_{f,i}$, and
\begin{align}
\begin{split}
&
I_{\bd}^{\text{g-g}}(\bsig;q, \bk)=
\prod_{a,b}
\left(\sigma_a^{d_b} q^{\frac12 d_a d_b}\right)^{k_{ab}},
\quad
I_{\bd, \bd_f}^{\text{g-f}}(\bsig;\bgamma, q, \bk^{\text{g-f}})=
\prod_{i, a}
\left(\gamma_i^{d_a} \sigma_a^{d_{f,i}} 
q^{d_{f,i} d_a}\right)^{k_{a i}^{\text{g-f}}},
\\
&
I_{\bd_f}^{\text{f-f}}(q, \bk^{\text{f-f}})=
\prod_{i,j} 
\left(\gamma_i^{d_{f,j}} q^{\frac12 d_{f,i} d_{f,j}}\right)^{k_{ij}^{\text{f-f}}},
\\
&
I_{\bd}^{\text{g-R}}(q, \bk^{\text{g-R}})=
\prod_{a} q^{\frac12 k_a^{\text{g-R}} d_a},
\quad
I_{\bd_f}^{\text{f-R}}(q, \bk^{\text{f-R}})=
\prod_{i} q^{\frac12 k_i^{\text{f-R}} d_{f,i}}.
\label{cs_factor_vortex}
\end{split}
\end{align}
Here the background magnetic fluxes are taken to be $2\bd_f$ 
for the factorizations of the mass parameters $\gamma_i$ as 
$\gamma_i = q^{(d_{f,i}-d_{f,i})/2} \gamma_i$ 
following \eqref{repara_vortex} with $\bd_f'=\bd_f''$.
\end{remark}

\section{Knot-gauge theory}
\label{sec:knot_building}

We construct 3D $\mathcal{N}=2$ abelian gauge theories 
labeled by knot diagrams, referred to as knot-gauge theories, 
whose K-theoretic vortex partition functions give 
the colored Jones polynomials of knots. 
For that purpose, we first recall how the colored Jones polynomials are 
obtained for tangle diagrams of knots from elementary building blocks as $R$-matrix, 
and, in Section \ref{subsec:gauge_knot_build}, construct 
corresponding constituents of knot-gauge theories for the elementary building blocks. 
We then discuss the JK residue procedure in the knot-gauge theories, and 
in particular show that, for a class of knot diagrams in Proposition \ref{prop:bridge}, the K-theoretic vortex partition functions actually give the colored Jones polynomials.

\subsection{A brief summary of colored Jones polynomials from tangles}

Let us recall the building blocks $R$, $R^{-1}$, $\mu$ and $\mu^{-1}$, 
associated with the quantum group $U_q(\mathfrak{sl}_2)$, 
which give the $n$-colored Jones polynomials of knots colored by symmetric representations $S^n$ \cite{Turaev:1988eb} (see also \cite{Murakami:2010}).  
The $(n+1)^2\times (n+1)^2$ $R$-matrix 
$R=R(q)$ for $n$-colored Jones polynomial 
of a knot $\cK$, 
assigned to each positive crossing in a tangle diagram of $\cK$, 
is given by \cite{KirbyMelvin}
\begin{align}
R^{d_{12}d_{31}}_{d_{42}d_{34}}
\quad &= \quad 
\inc{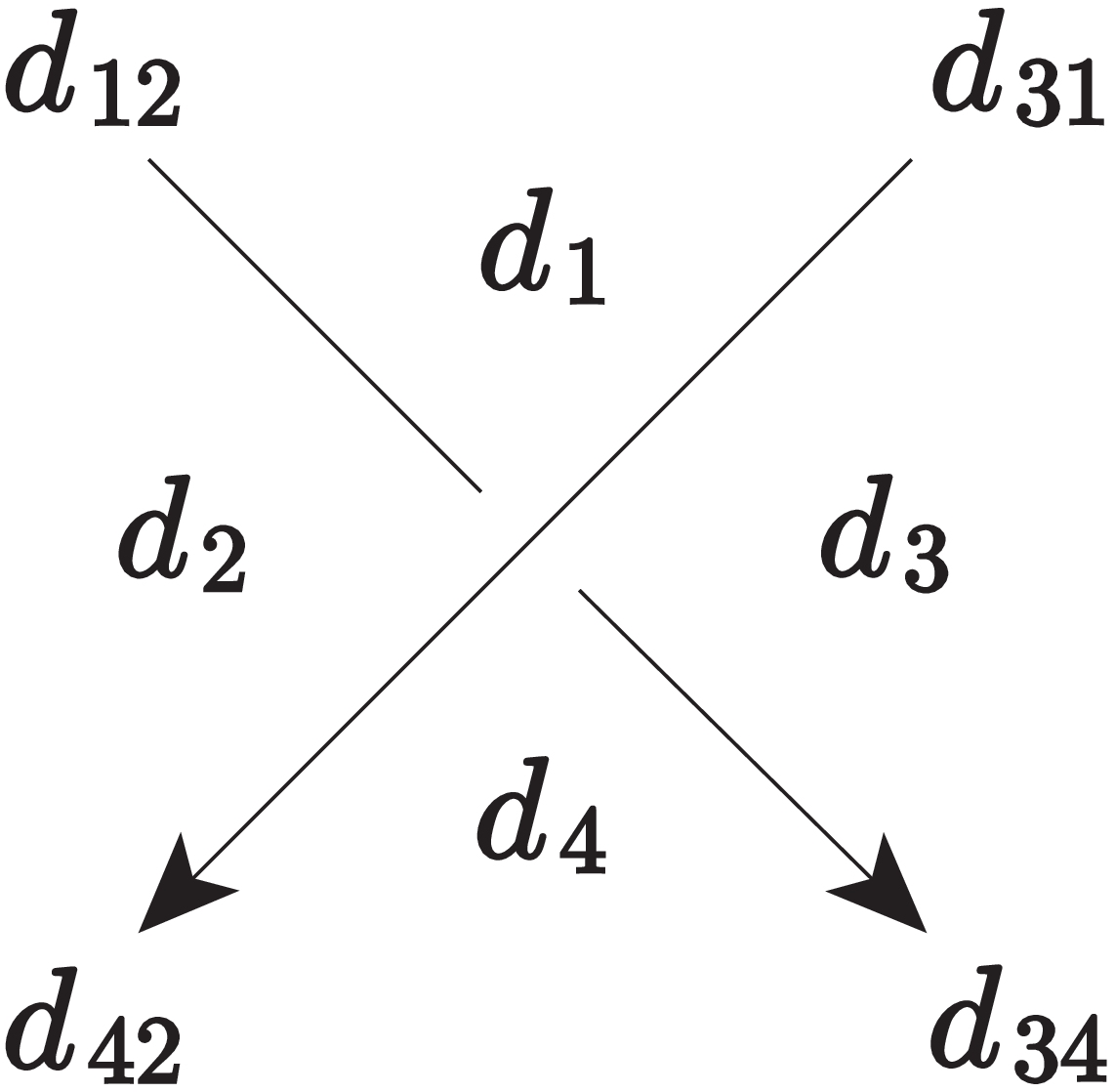}
\nonumber
\\
&=\quad
(-1)^{d_{31}-d_{42}}
q^{\frac12 (d_{31}-d_{42})(d_{31}-d_{42}-1) + d_{42}d_{34}
-\frac12 n(d_{31}+d_{34}) + \frac14 n^2}
\nonumber
\\
&\hspace{15em}\times
\frac{\qq{q}{n-d_{42}} \qq{q}{d_{34}}}
{\qq{q}{d_{12}} \qq{q}{n-d_{31}} \qq{q}{d_{31}-d_{42}}},
\label{R_matrix}
\end{align}
where $\qq{q}{d}=\qpoch{q}{q}{d}$. 
Here, following \cite{ChoMurakami}, the variables 
$d_i$, $i=1,2,3,4$, are assigned to regions around the crossing, 
$d_{ij}=d_i-d_j$ are assigned to arcs, and 
\begin{align}
d_{12},\ \ 
d_{31},\ \ 
d_{42},\ \ 
d_{34},\ \ 
d_{31}-d_{42} \in \{0,1,\ldots,n\}.
\label{R_charge_c}
\end{align}
For the variables $d_i$ that do not satisfy the conditions \eqref{R_charge_c}, 
$R^{d_{12}d_{31}}_{d_{42}d_{34}}=0$ is defined. 
Note that the arrows in \eqref{R_matrix} are promised to 
be in the downward directions (see \eqref{twist_R} for local deformations 
of tangle).
The inverse $R$-matrix $R^{-1}=R(q)^{-1}$, 
assigned to each negative crossing, is similarly given by
\begin{align}
\left(R^{-1}\right)^{d_{12}d_{31}}_{d_{42}d_{34}}
\quad &= \quad 
\inc{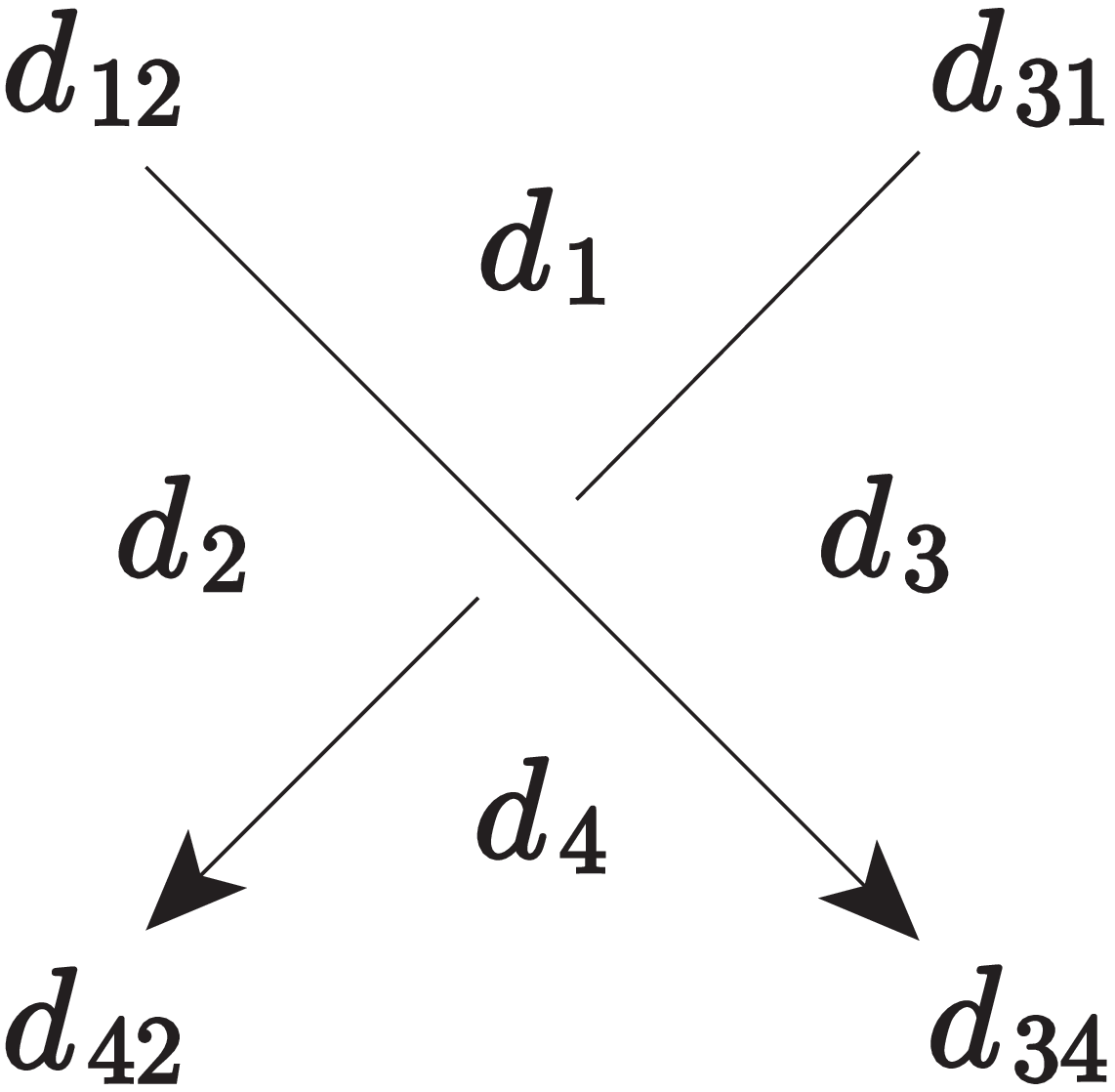} 
\quad
=\quad
R(q^{-1})^{d_{31}d_{12}}_{d_{34}d_{42}}
\nonumber
\\
&=\quad
q^{ - d_{12}d_{31} + \frac12 n(d_{31}+d_{34}) - \frac14 n^2}\,
\frac{\qq{q}{d_{42}} \qq{q}{n-d_{34}}}
{\qq{q}{n-d_{12}} \qq{q}{d_{31}} \qq{q}{d_{12}-d_{34}}},
\label{R_matrix_m}
\end{align}
where
\begin{align}
d_{12},\ \ 
d_{31},\ \ 
d_{42},\ \ 
d_{34},\ \ 
d_{12}-d_{34} \in \{0,1,\ldots,n\}.
\label{invR_charge_c}
\end{align} 
To each local minimum and maximum in the tangle diagram
\begin{align}
\begin{split}
&
\mu_{d_{12}}\quad
=\quad \inc{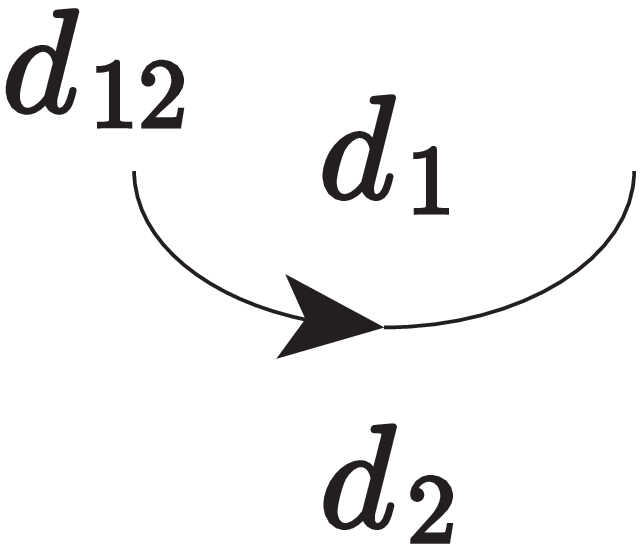} \quad
=\quad
q^{d_{12}-\frac12n},
\\
&
\mu^{-1}_{d_{12}}\quad
=\quad \inc{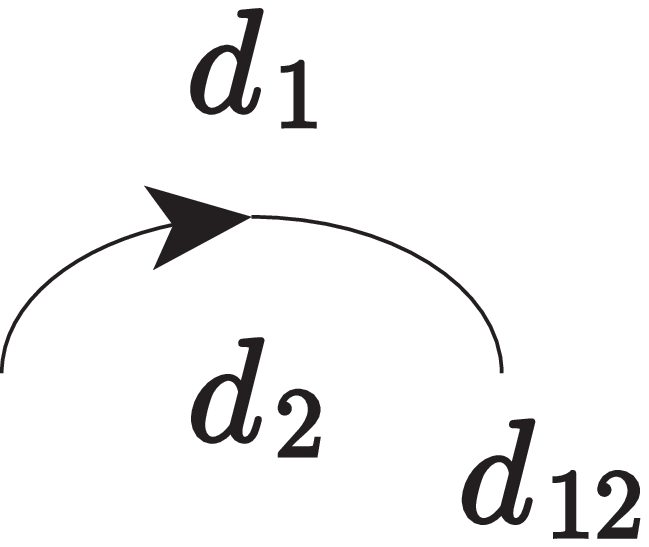} \quad
=\quad
q^{-d_{12}+\frac12n},
\label{loc_min_max}
\end{split}
\end{align}
are assigned, respectively, where $0 \le d_{12} \le n$. 
We also set
\begin{align}
\inc{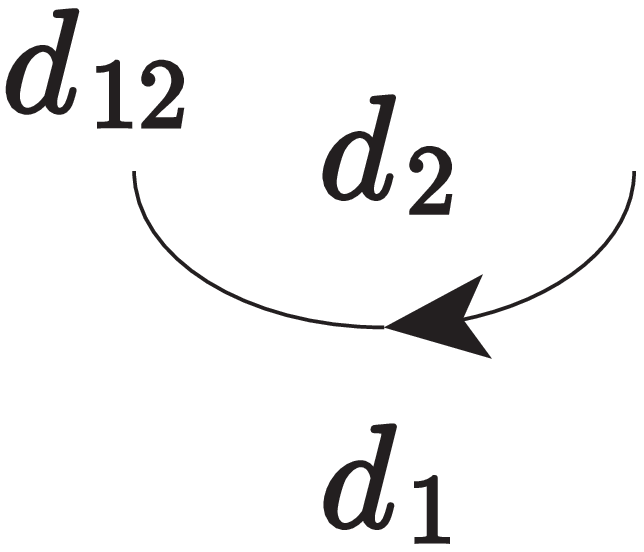} \quad
=\quad \inc{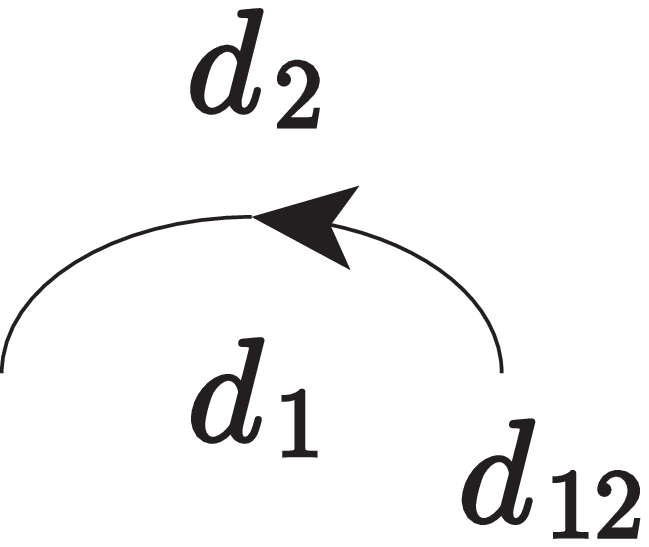} \quad
=\quad
1.
\end{align}

Using the above building blocks, 
the $n$-colored Jones polynomial of a knot $\cK$ is given by
\begin{align}
\overline{J}_n^{\cK}(q)=
\sum_{\bd} 
\left(\prod_{i \in P(\bd)} q^{-\frac14 n(n+2)}R_i \right)
\left(\prod_{i \in N(\bd)} q^{\frac14 n(n+2)}R_i^{-1} \right)
\left(\prod_{i \in \textrm{min}(\bd)} \mu_i \right)
\left(\prod_{i \in \textrm{max}(\bd)} \mu_i^{-1} \right).
\label{c_jones_knot}
\end{align}
Here, for a $(0,0)$-tangle (the closure of a $(1,1)$-tangle) diagram of $\cK$, 
$P(\bd)$, $N(\bd)$, $\textrm{min}(\bd)$ and $\textrm{max}(\bd)$ 
are the set of positive crossings, negative crossings, local minima and 
local maxima with the variables $\bd$, respectively, and 
$R_i$, $R_i^{-1}$, $\mu_i$ and $\mu_i^{-1}$ denote 
the $i$-th $R$-matrix, the $i$-th inverse $R$-matrix, 
the quantity assigned to the $i$-th local minimum and 
the quantity assigned to the $i$-th local maximum, respectively. 
The domain of $\bd$ is determined for the given $(0,0)$-tangle diagram of $\cK$ by \eqref{R_charge_c} and \eqref{invR_charge_c}. 
The normalized colored Jones polynomial is also introduced by
\begin{align}
J_n^{\cK}(q)=
\frac{\overline{J}_n^{\cK}(q)}{\overline{J}_n^{\mathbf{0}}(q)},
\end{align}
where 
\begin{align}
\overline{J}_n^{\mathbf{0}}(q)
=\sum_{d=0}^{n} \ \inc{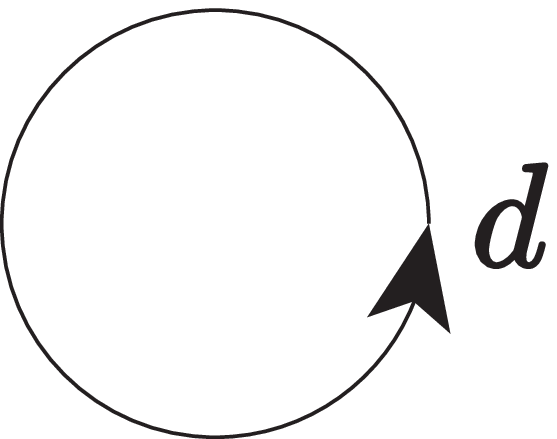}
=\sum_{d=0}^{n} \mu_{d}=
\frac{q^{\frac12 (n+1)}-q^{-\frac12 (n+1)}}
{q^{\frac12}-q^{-\frac12}}=
q^{-\frac12 n} \frac{\qpoch{q^2}{q}{n}}{\qpoch{q}{q}{n}},
\label{unknot_jones}
\end{align}
is the (unnormalized) colored Jones polynomial of unknot.
The normalized colored Jones polynomial 
is shown to be given for the $(1,1)$-tangle diagram of $\cK$ 
\cite[Lemma 3.9]{KirbyMelvin} (see also \cite[Section 2.5]{Murakami:2010}) 
with an incoming (outgoing) constant $f \in \{0,1,\ldots,n\}$ 
assigned to the external arcs as in Figure \ref{fig:tangle_31}, 
which specifies a basis of the vector space attached to the tangle diagram,
\begin{align}
J_n^{\cK}(q)=
\sum_{\bd} 
\left(\prod_{i \in P(\bd,f)} q^{-\frac14 n(n+2)}R_i \right)
\left(\prod_{i \in N(\bd,f)} q^{\frac14 n(n+2)}R_i^{-1} \right)
\left(\prod_{i \in \textrm{min}(\bd,f)} \mu_i \right)
\left(\prod_{i \in \textrm{max}(\bd,f)} \mu_i^{-1} \right).
\label{c_norm_jones_knot}
\end{align}
Remark that the normalized colored Jones polynomial does not depend on the constant $f$.

\begin{figure}[t]
\centering
\includegraphics[width=90mm]{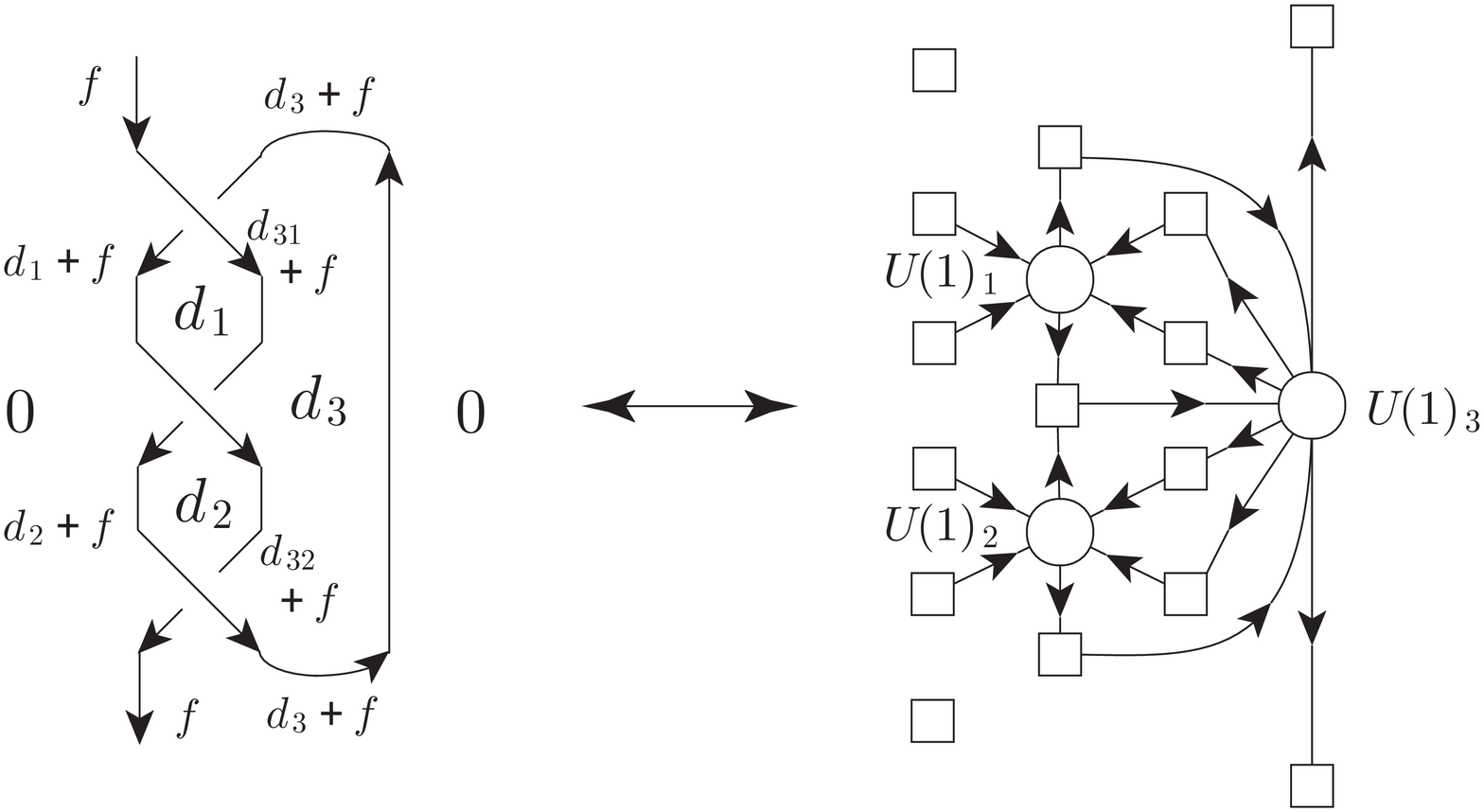}
\caption{The left figure is a $(1,1)$-tangle diagram of 
the trefoil knot $\knot{3}{1}$, and 
the right figure, 
which is obtained by the associations as in Figure \ref{fig:r_quiver} 
for all the crossings, 
represents the associated quiver-like diagram of 
$U(1)^3$ gauge theory $T[\knot{3}{1}]$ with 15 chiral fields 
(in Table \ref{31_matter}) and extra Chern-Simons couplings. 
The variables $d_1, d_2, d_3$ are assigned to the bounded regions, and 
an incoming (outgoing) constant $f \in \{0,1,\ldots,n\}$ 
is identified with the background magnetic flux for 
the global symmetry $U(1)_{ext}$.}
\label{fig:tangle_31}
\end{figure}

\subsection{Building blocks of knot-gauge theories}
\label{subsec:gauge_knot_build}

We recalled that the normalized colored Jones polynomials are obtained 
from the building blocks $R$, $R^{-1}$, $\mu$ and $\mu^{-1}$ by \eqref{c_norm_jones_knot}. Based on the formulation, 
we propose a construction of 3D $\mathcal{N}=2$ abelian gauge theories 
labeled by $(1,1)$-tangle diagrams.

Consider a (1,1)-tangle diagram of a knot $\cK$ 
with the number of loops (regions) $N_v$ and 
an incoming (outgoing) constant $f \in \{0,1,\ldots,n\}$ 
(see Figure \ref{fig:tangle_31} for an example of $N_v=3$).
We assign variables $d_I$, $I=0, 1, \ldots, N_v$, to the bounded and unbounded regions, and take $d_0=0$ without loss of generality by a shift of variables. 
By associating a $U(1)_I$ gauge symmetry to the region with the non-zero variable $d_I$, 
we construct a $U(1)^{N_v}=U(1)_1 \times \cdots \times U(1)_{N_v}$ gauge theory $T[\cK]$, where the variables $d_I$ are identified with the magnetic fluxes.
Here the number of crossings is also $N_v$, and 
a matter content for the $R$- or inverse $R$-matrix assigned at each crossing 
is constructed in the following, 
where each matter content provides a building block of 
the 3D $\mathcal{N}=2$ $U(1)^{N_v}$ gauge theory $T[\cK]$, that we call 
the knot-gauge theory. 
We will see that 
the knot-gauge theory has a global symmetry $U(1)_F^{2N_v}$, and 
the color $n$ and the constant $f$ are, respectively, identified with 
the background magnetic fluxes for 
global symmetries $U(1)_c$ and $U(1)_{ext}$ in $U(1)_F^{2N_v}$:
\begin{align}
U(1)_c, \ U(1)_{ext} \subset U(1)_F^{2N_v}.
\label{sub_ext}
\end{align}
In the following, we construct the building blocks of $T[\cK]$. 
Once they are constructed, it is straightforward to construct 
the gauge theory $T[\cK]$ from them.

\subsubsection{$R$-matrix}\label{subsubsec:R_gauge}

\begin{figure}[t]
\centering
\includegraphics[width=120mm]{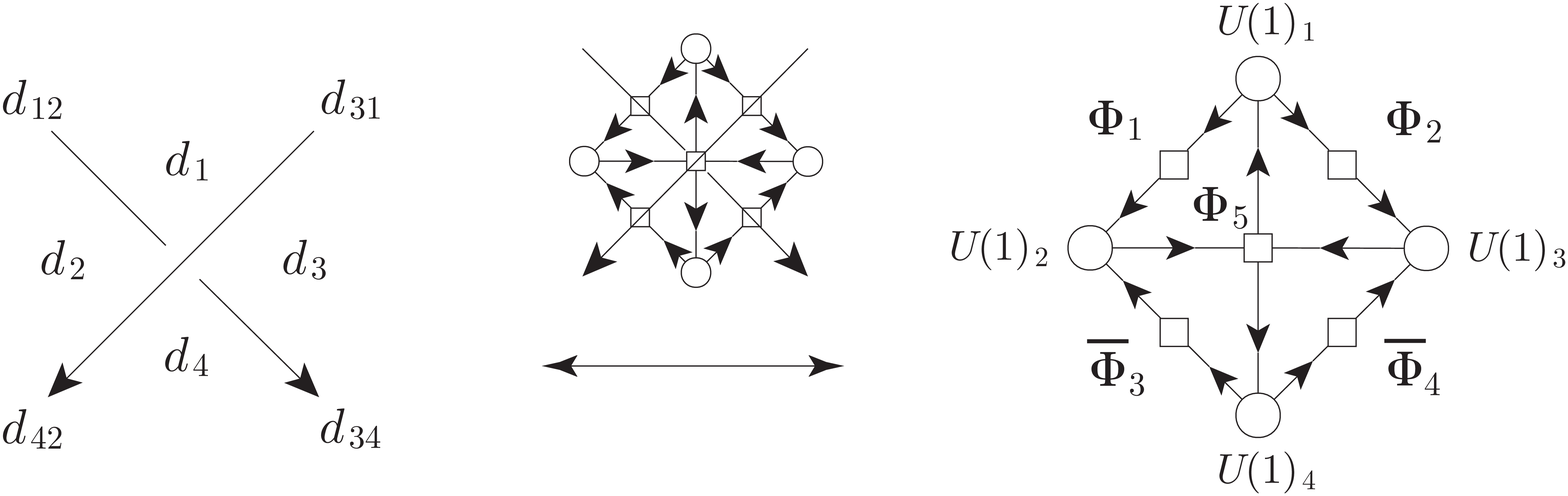}
\caption{For the $R$-matrix, 
the quiver-like diagram on the right is associated, where 
the circles represent the $U(1)_i$ gauge nodes and 
the squares denote five chiral fields given in Table \ref{R_mat_tb}.
The same applies to the inverse $R$-matrix in Section \ref{subsubsec:invR_gauge}.}
\label{fig:r_quiver}
\end{figure}

\begin{table}[t]
\begin{center}
\begin{tabular}{|c|c|c|c|c||c|c|c||c|}
\hline
Field & $U(1)_{1}$ & $U(1)_{2}$ 
& $U(1)_{3}$ & $U(1)_{4}$ & $U(1)_{c}$ & $U(1)_{ext}$ & mass & $U(1)_R$ 
\\ 
\hline
$\Phi_{1}$ & $1$ & $-1$ & $0$ & $0$ & $0$ & $1$ & $\gamma_1$ & $0$ \\
$\Phi_{2}$ & $1$ & $0$ & $-1$ & $0$ & $1$ & $-1$ & $\gamma_2$ & $0$ \\
$\overline{\Phi}_{3}$ & $0$ & $-1$ & $0$ & $1$ & $-1$ & $1$ & $\gamma_2^{-1}$ & $2$ \\
$\overline{\Phi}_{4}$ & $0$ & $0$ & $-1$ & $1$ & $0$ & $-1$ & $\gamma_1^{-1}$  & $2$ \\
$\Phi_{5}$ & $-1$ & $1$ & $1$ & $-1$ & $0$ & $0$ & $1$ & $0$ \\
\hline
\end{tabular}
\caption{
Matter content for the $R$-matrix $R^{d_{12}d_{31}}_{d_{42}d_{34}}$ 
in \eqref{R_matrix} which is denoted by $T[R^{d_{12}d_{31}}_{d_{42}d_{34}}(\bgamma)]$ 
including extra Chern-Simons couplings \eqref{R_cs_couple}, 
where $U(1)_R$ denotes the $U(1)_R$ charge.}
\label{R_mat_tb}
\end{center}
\end{table}

We focus on a crossing, with an assigned $R$-matrix, for a (1,1)-tangle diagram 
with $N_v$ crossings. 
Let $U(1)_i$, $i=1,2,3,4$, be gauge symmetries with 
associated complex scalars $\sigma_i$ and magnetic fluxes $d_i$, 
and associate the quiver-like diagram of Figure \ref{fig:r_quiver} with 
the five chiral fields (multiplets) in Table \ref{R_mat_tb}. 
Here, note that the gauge symmetries act, in general, on not only
the chiral fields at the crossing but also the chiral fields at other crossings 
which share the same regions in the tangle diagram. 
The chiral fields interact through a superpotential 
$W=\Phi_{1} \overline{\Phi}_{4} \Phi_{5} + 
\Phi_{2} \overline{\Phi}_{3} \Phi_{5}$ 
(see also Remark \ref{rem:octahedron}), and 
two mass parameters $\bgamma=(\gamma_1, \gamma_2)$, 
associated with a global symmetry $U(1)_F^2$, are introduced. 
Overall, $2N_v$ mass parameters, associated with a global symmetry $U(1)_F^{2N_v}$, are introduced for the tangle diagram, and 
the global symmetries $U(1)_c$ and $U(1)_{ext}$, 
which associates non-negative background magnetic fluxes $2n$ and $2f$, 
are two of $U(1)_F^{2N_v}$ as \eqref{sub_ext}.
Note that the $3 N_v$ $U(1)$ symmetries associated with the $5 N_v$
chiral multiplets with the superpotential are  generated by
$U(1)^{N_v}$ gauge symmetries and $U(1)_F^{2 N_v}$ flavor symmetries.
We now construct the $R$-matrix \eqref{R_matrix} as 
a building block of K-theoretic vortex partition functions. 


From \eqref{bk_hyper} and \eqref{bk_hyper_inv}, 
the five chiral fields lead to a 1-loop building block
\begin{align}
I_{\textrm{1-loop}}^{R}
(\bsig;\bgamma)=
\left(\frac{\sigma_1\sigma_4}{\sigma_2\sigma_3}\right)^{\frac12}
\frac{
\left(1-\gamma_2\frac{\sigma_2}{\sigma_4}\right)
\left(1-\gamma_1\frac{\sigma_3}{\sigma_4}\right)}
{
\left(1-\gamma_1\frac{\sigma_1}{\sigma_2}\right)
\left(1-\gamma_2\frac{\sigma_1}{\sigma_3}\right)
\left(1-\frac{\sigma_2\sigma_3}{\sigma_1\sigma_4}\right)},
\label{R_1loop}
\end{align}
and a building block of K-theoretic vortex partition functions
\begin{align}
I_{\bd,n,f}^R(\bsig;\bgamma,q)&=
(-1)^{d_{23}+n}\,
q^{\frac14 (d_{12}-d_{34})(3d_1-d_2-d_3-d_4+2n+1)}
\left(\gamma_1\gamma_2\right)^{\frac12 (d_{12}-d_{34})}
\nonumber
\\
&\quad\times
\left(\frac{\sigma_1^3 \sigma_4}
{\sigma_2^2 \sigma_3^2}\right)^{\frac12 d_1}
\left(\frac{\sigma_2 \sigma_3}
{\sigma_1^2}\right)^{\frac12 (d_2+d_3)}
\left(\frac{\sigma_1}
{\sigma_4}\right)^{\frac12 d_4}
\left(\frac{\sigma_1 \sigma_4}
{\sigma_2 \sigma_3}\right)^{\frac12 n}
\nonumber
\\
&\quad\times
\frac{\qpoch{q \gamma_2\frac{\sigma_2}{\sigma_4}}{q}{n-d_{42}-f}
\qpoch{q \gamma_1 \frac{\sigma_3}{\sigma_4}}{q}{d_{34}+f}}
{\qpoch{q \gamma_1 \frac{\sigma_1}{\sigma_2}}{q}{d_{12}+f}
\qpoch{q \gamma_2\frac{\sigma_1}{\sigma_3}}{q}{n-d_{31}-f}
\qpoch{q \frac{\sigma_2\sigma_3}{\sigma_1\sigma_4}}{q}{d_{31}-d_{42}}},
\label{vortex_pf_R}
\end{align}
where $\bsig=(\sigma_1, \sigma_2, \sigma_3, \sigma_4)$, 
$d_{ij}=d_i-d_j$, and the background magnetic fluxes $n$ and $f$ for 
the global symmetries $U(1)_{c}$ and $U(1)_{ext}$ are introduced by the shift \eqref{bg_flux}.
In addition to the matter content in Table \ref{R_mat_tb} we also 
introduce Chern-Simons couplings by \eqref{cs_factor_1loop} and \eqref{cs_factor_vortex} as
\begin{align}
I_{\textrm{1-loop}}^{\textrm{CS}}(\bsig)=
\left(\frac{\sigma_1\sigma_4}{\sigma_2\sigma_3}\right)^{\frac12},
\end{align}
and
\begin{align}
I_{\bd,n,f}^{\textrm{CS}}(\bsig;q)&=
q^{-\frac14 (d_1^2-d_2^2-d_3^2+d_4^2)+\frac12 (d_1d_4-d_2d_3+ n d_{23})
+\frac14 (d_{12}-d_{34}) - \frac12 n + f^2 - d_{23}f - nf}
\nonumber
\\
&\quad\times
\left(\frac{\sigma_4}{\sigma_1}\right)^{\frac12 d_{14}}
\left(\frac{\sigma_2}{\sigma_3}\right)^{\frac12 d_{23}+\frac12 n-f},
\label{R_cs_factor_off}
\end{align}
for non-zero Chern-Simons couplings
\begin{align}
\begin{split}
k_{11}=k_{44}=k_{23}=k_{32}=-\frac12,\quad
k_{22}=k_{33}=k_{14}=k_{41}=\frac12,
\\
k_{2c}^{\text{g-f}}=\frac12,\quad
k_{3c}^{\text{g-f}}=-\frac12,\quad
k_{2f}^{\text{g-f}}=-1,\quad
k_{3f}^{\text{g-f}}=1,
\\
k_1^{\text{g-R}}=k_4^{\text{g-R}}=\frac12,\quad
k_2^{\text{g-R}}=k_3^{\text{g-R}}=-\frac12,\quad
k_{c}^{\text{f-R}}=-1,\quad
k_{cf}^{\text{f-f}}=k_{fc}^{\text{f-f}}=-1,\quad
k_{ff}^{\text{f-f}}=2,
\label{R_cs_couple}
\end{split}
\end{align}
where the subscripts $i$, $c$ and $f$ denote the indices for 
the gauge symmetries $U(1)_i$, the global symmetries $U(1)_c$ and $U(1)_{ext}$, 
respectively, and the mass parameters for $U(1)_c$ and $U(1)_{ext}$ 
are abbreviated. 
By combining the block \eqref{vortex_pf_R} with the Chern-Simons factor \eqref{R_cs_factor_off}, we define a building block 
for the $R$-matrix:
\begin{align}
I_{n,f}^{T[R^{d_{12}d_{31}}_{d_{42}d_{34}}(\bgamma)]}(\bsig;q):=
(-1)^n\,
I_{\bd,n,f}^{\textrm{CS}}(\bsig;q)\, I_{\bd,n,f}^R(\bsig;\bgamma,q),
\label{vortex_pf_TR_cs}
\end{align}
where the label $T[R^{d_{12}d_{31}}_{d_{42}d_{34}}(\bgamma)]$ is introduced.%
\footnote{\label{footnote:p_anomaly}
The building block 
$$
I_{\textrm{1-loop}}^{\textrm{CS}}(\bsig)\,
I_{\textrm{1-loop}}^{R}(\bsig;\bgamma)
\big(I_{\bd',n,f}^{\textrm{CS}}(\bsig;q)\,
I_{\bd',n,f}^R(\bsig;\bgamma,q)\big)
\big(I_{\bd'',n,f}^{\textrm{CS}}(\bsig;q^{-1})\,
I_{\bd'',n,f}^R(\bsig;\bgamma,q^{-1})\big)
$$
of twisted partition functions, 
or Table \ref{R_mat_tb} with the Chern-Simons couplings \eqref{R_cs_couple}, 
shows the absence of the parity anomalies except for 
the $U(1)_1$-$U(1)_c$ and $U(1)_4$-$U(1)_c$ parity anomalies. 
These parity anomalies just exist for each crossing, and 
one can show that they are canceled out for each loop in (1,1)-tangle diagrams 
and actually absent. 
Here, in general, the $U(1)_a$-$U(1)_b$ parity anomalies are absent when 
$k_{ab}+N/2$ is integer-valued, where 
$k_{ab}$ is the $U(1)_a$-$U(1)_b$ Chern-Simons coupling and 
$N$ is the number of chirals charged under both $U(1)_a$ and $U(1)_b$.
} 
If the complex scalars are specialized as
\begin{align}
\sigma_1, \sigma_2, \sigma_3, \sigma_4 \to 1,
\label{R_sp_1}
\end{align}
then under massless limit 
\begin{align}
\gamma_1, \gamma_2 \to 1,
\label{R_massless}
\end{align}
the building block \eqref{vortex_pf_TR_cs} yields 
the $R$-matrix \eqref{R_matrix} with the constant shift $f$:
\begin{align}
(-1)^{d_{14}}\,
I_{n,f}^{T[R^{d_{12}d_{31}}_{d_{42}d_{34}}(\bgamma)]}(\bsig;q) 
\quad \to \quad
q^{-\frac14 n(n+2)}\,
R^{d_{12}+f\, d_{31}+f}_{d_{42}+f\, d_{34}+f}
\,.
\label{R_limit}
\end{align}
Here the prefactor $(-1)^{d_{14}}$ on the left side is 
introduced by specializing the complexified FI parameters, 
and the prefactor $q^{-n(n+2)/4}$ on the right side corresponds to 
a normalization factor appeared in \eqref{c_norm_jones_knot}. 
The conditions \eqref{R_charge_c} for the magnetic fluxes and 
the specializations \eqref{R_sp_1} of the complex scalars 
should be obtained by means of the JK residue in \eqref{3d_tw_pf} 
for a choice of the FI parameters. 
In Section \ref{sec:gluing} we will discuss the justification of these.


The non-trivial degenerate $R$-matrices
\begin{align}
&
R^{\ 0\ d_{32}}_{d_{42}d_{34}}
=
(-1)^{d_{34}}
q^{\frac12 d_{34}(d_{32}+d_{42}-1)
-\frac12 n(d_{32}+d_{34}) + \frac14 n^2}\,
\frac{\qq{q}{n-d_{42}}}{\qq{q}{n-d_{32}}},
\label{deg_R_matrix_1}
\\
&
R^{d_{12}d_{31}}_{\ 0\ d_{32}}=
(-1)^{d_{31}}
q^{\frac12 d_{31}(d_{31}-1) - \frac12 n(d_{31}+d_{32}) + \frac14 n^2}\,
\frac{\qq{q}{n} \qq{q}{d_{32}}}
{\qq{q}{d_{12}} \qq{q}{n-d_{31}} \qq{q}{d_{31}}},
\label{deg_R_matrix_2}
\\
&
R^{d_{12}\ 0\ }_{\ 0\ d_{12}}=R^{\ 0\ d_{12}}_{d_{12} \ 0\ }=
q^{-\frac12 n d_{12} + \frac14 n^2},
\label{deg_R_matrix_3}
\end{align}
can be also constructed by identifying some $U(1)$ gauge symmetries, 
by identifying $U(1)_1$ with $U(1)_2$ for \eqref{deg_R_matrix_1} 
and by identifying $U(1)_4$ with $U(1)_3$ for \eqref{deg_R_matrix_2}. 
Here, by identifying $U(1)_1$ with $U(1)_2$, 
the $U(1)_1 \times U(1)_2$ charges $(a_1, a_2)$ of a chiral field 
are changed to the $U(1)_2$ charge $a_1+a_2$. 
For the degenerate $R$-matrix \eqref{deg_R_matrix_3}, 
which is used in Section \ref{sec:ded_gauge}, 
we will construct it just as a Chern-Simons factor as \eqref{vortex_pf_R_deg_a}.

\begin{figure}[t]
\centering
\includegraphics[width=90mm]{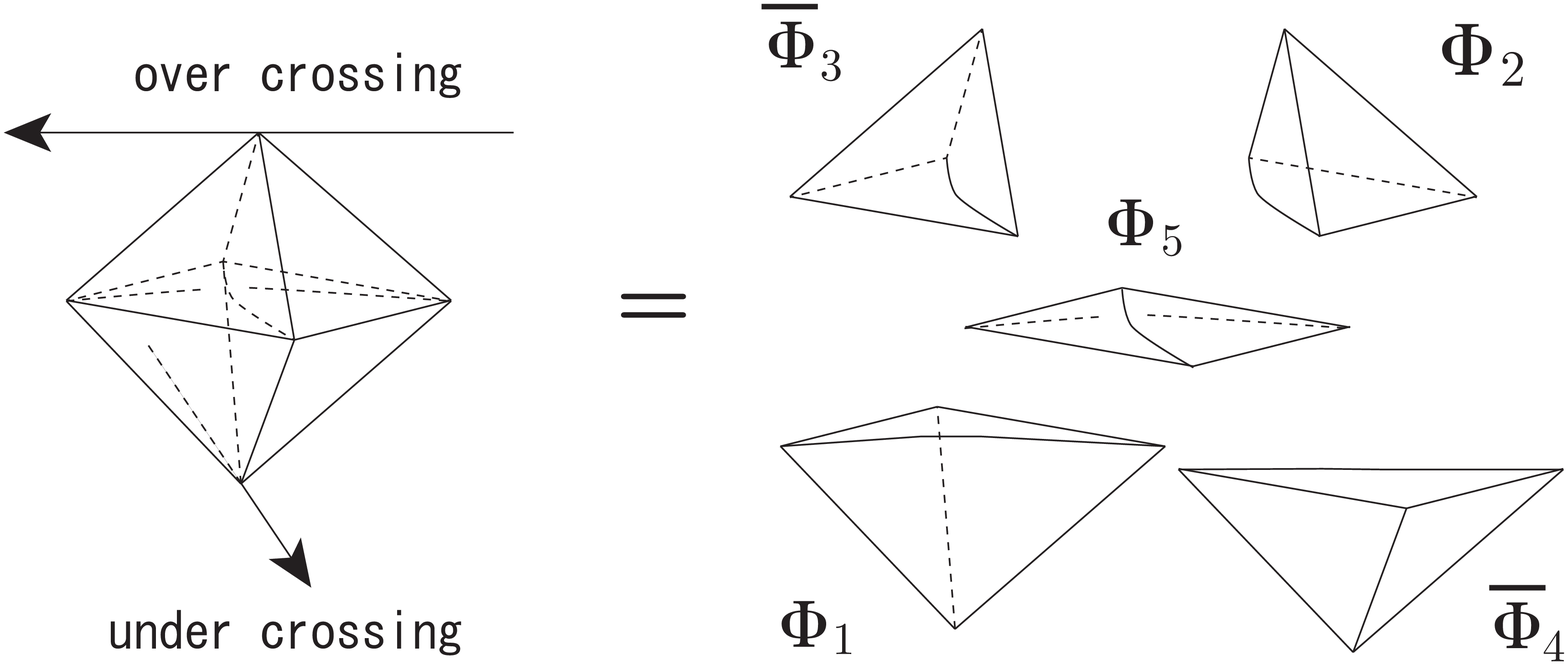}
\caption{Geometric interpretation of the $R$-matrix by 
D. Thurston \cite{Thurston:99}. 
An octahedron is associated for the $R$-matrix 
and decomposed into five tetrahedra which correspond to 
the five chiral fields in Table \ref{R_mat_tb}.}
\label{fig:octahedron}
\end{figure}

\begin{remark}\label{rem:octahedron}
Following D. Thurston, to each crossing, with an assigned $R$-matrix, 
an octahedron, which can be decomposed into five tetrahedra, is attached 
(see Figure \ref{fig:octahedron}) \cite{Thurston:99}. 
This provides a geometric interpretation of the $R$-matrix and 
is utilized to prove the volume conjecture 
\cite{Kashaev:1996kc, MurakamiMurakami} for some specific hyperbolic knots 
in \cite{Yokota:2011, Murakami:2001}. 
In our gauge theory construction, 
each decomposed tetrahedron corresponds to 
a chiral field in Table \ref{R_mat_tb} or \ref{R_mat_tb_m} 
(see \cite{Dimofte:2011ju} for a similar gauge theory construction of 
the octahedron by Dimofte, Gaiotto and Gukov).
\end{remark}

\subsubsection{Inverse $R$-matrix}\label{subsubsec:invR_gauge}

\begin{table}[t]
\begin{center}
\begin{tabular}{|c|c|c|c|c||c|c|c||c|}
\hline
Field & $U(1)_{1}$ & $U(1)_{2}$ 
& $U(1)_{3}$ & $U(1)_{4}$ & $U(1)_{c}$ & $U(1)_{ext}$ & mass & $U(1)_R$ 
\\ 
\hline
$\Phi_{1}'$ & $-1$ & $1$ & $0$ & $0$ & $1$ & $-1$ & $\gamma_1$ & $0$ \\
$\Phi_{2}'$ & $-1$ & $0$ & $1$ & $0$ & $0$ & $1$ & $\gamma_2$ & $0$ \\
$\overline{\Phi}_{3}'$ & $0$ & $1$ & $0$ & $-1$ & $0$ & $-1$ & $\gamma_2^{-1}$ & $2$ \\
$\overline{\Phi}_{4}'$ & $0$ & $0$ & $1$ & $-1$ & $-1$ & $1$ & $\gamma_1^{-1}$ & $2$ \\
$\Phi_{5}'$ & $1$ & $-1$ & $-1$ & $1$ & $0$ & $0$ & $1$ & $0$ \\
\hline
\end{tabular}
\caption{
Matter content for the inverse $R$-matrix 
$\left(R^{-1}\right)^{d_{12}d_{31}}_{d_{42}d_{34}}$ in \eqref{R_matrix_m} 
which is denoted by $T[\overline{R}^{d_{12}d_{31}}_{d_{42}d_{34}}(\bgamma)]$ 
including extra Chern-Simons couplings \eqref{invR_cs_couple}.}
\label{R_mat_tb_m}
\end{center}
\end{table}

For the inverse $R$-matrix \eqref{R_matrix_m}, we consider 
the mater content in Table \ref{R_mat_tb_m} with mass parameters 
$\bgamma=(\gamma_1, \gamma_2)$ associated with a global symmetry $U(1)_F^2$. 
The chiral fields interact through a superpotential
$W=\Phi_{1}' \overline{\Phi}_{4}' \Phi_{5}' +
\Phi_{2}' \overline{\Phi}_{3}' \Phi_{5}'$. 
The associated building blocks of the 1-loop factors and the 
K-theoretic vortex partition functions are, respectively, given by
\begin{align}
I_{\textrm{1-loop}}^{\overline{R}}
(\bsig;\bgamma)=
\left(\frac{\sigma_2\sigma_3}{\sigma_1\sigma_4}\right)^{\frac12}
\frac{
\left(1-\gamma_2\frac{\sigma_4}{\sigma_2}\right)
\left(1-\gamma_1\frac{\sigma_4}{\sigma_3}\right)}
{
\left(1-\gamma_1\frac{\sigma_2}{\sigma_1}\right)
\left(1-\gamma_2\frac{\sigma_3}{\sigma_1}\right)
\left(1-\frac{\sigma_1\sigma_4}{\sigma_2\sigma_3}\right)},
\label{invR_1loop}
\end{align}
and
\begin{align}
I_{\bd,n,f}^{\overline{R}}(\bsig;\bgamma, q)&=
(-1)^{d_{23}+n}\,
q^{\frac14 (d_{12}-d_{34})(3d_1-d_2-d_3-d_4-2n-1)}
\left(\gamma_1\gamma_2\right)^{-\frac12 (d_{12}-d_{34})}
\nonumber
\\
&\quad\times
\left(\frac{\sigma_1^3 \sigma_4}
{\sigma_2^2 \sigma_3^2}\right)^{\frac12 d_1}
\left(\frac{\sigma_2 \sigma_3}
{\sigma_1^2}\right)^{\frac12 (d_2+d_3)}
\left(\frac{\sigma_1}
{\sigma_4}\right)^{\frac12 d_4}
\left(\frac{\sigma_2 \sigma_3}
{\sigma_1 \sigma_4}\right)^{\frac12 n}
\nonumber
\\
&\quad\times
\frac{\qpoch{q \gamma_2\frac{\sigma_4}{\sigma_2}}{q}{d_{42}+f}
\qpoch{q \gamma_1\frac{\sigma_4}{\sigma_3}}{q}{n-d_{34}-f}}
{\qpoch{q \gamma_1\frac{\sigma_2}{\sigma_1}}{q}{n-d_{12}-f}
\qpoch{q \gamma_2\frac{\sigma_3}{\sigma_1}}{q}{d_{31}+f}
\qpoch{q \frac{\sigma_1\sigma_4}{\sigma_2\sigma_3}}{q}{d_{12}-d_{34}}}.
\label{vortex_pf_R_m}
\end{align}
By combining this block with Chern-Simons factors
\begin{align}
I_{\textrm{1-loop}}^{\overline{\textrm{CS}}}(\bsig)=
\left(\frac{\sigma_1\sigma_4}{\sigma_2\sigma_3}\right)^{\frac12},
\end{align}
and
\begin{align}
I_{\bd,n,f}^{\overline{\textrm{CS}}}(\bsig;q)&=
q^{\frac14 (d_1^2-d_2^2-d_3^2+d_4^2)-\frac12 (d_1d_4-d_2d_3+ n d_{23})
+\frac14 (d_{12}-d_{34}) + \frac12 n -f^2 + d_{23}f + nf}
\nonumber
\\
&\quad\times
\left(\frac{\sigma_1}{\sigma_4}\right)^{\frac12 d_{14}}
\left(\frac{\sigma_3}{\sigma_2}\right)^{\frac12 d_{23} + \frac12 n-f},
\label{invR_cs_factor_off}
\end{align}
for non-zero Chern-Simons couplings
\begin{align}
\begin{split}
k_{11}=k_{44}=k_{23}=k_{32}=\frac12,\quad
k_{22}=k_{33}=k_{14}=k_{41}=-\frac12,
\\
k_{2c}^{\text{g-f}}=-\frac12,\quad
k_{3c}^{\text{g-f}}=\frac12,\quad
k_{2f}^{\text{g-f}}=1,\quad
k_{3f}^{\text{g-f}}=-1,
\\
k_1^{\text{g-R}}=k_4^{\text{g-R}}=\frac12,\quad
k_2^{\text{g-R}}=k_3^{\text{g-R}}=-\frac12,\quad
k_{c}^{\text{f-R}}=1,\quad
k_{cf}^{\text{f-f}}=k_{fc}^{\text{f-f}}=1,\quad
k_{ff}^{\text{f-f}}=-2,
\label{invR_cs_couple}
\end{split}
\end{align}
a building block for the inverse $R$-matrix is introduced by
\begin{align}
I_{n,f}^{T[\overline{R}^{d_{12}d_{31}}_{d_{42}d_{34}}(\bgamma)]}(\bsig;q):=
(-1)^n\,
I_{\bd,n,f}^{\overline{\textrm{CS}}}(\bsig;q)\, 
I_{\bd,n,f}^{\overline{R}}(\bsig;\bgamma,q),
\label{vortex_pf_invTR_cs}
\end{align}
and labeled by $T[\overline{R}^{d_{12}d_{31}}_{d_{42}d_{34}}(\bgamma)]$. 
By the specializations \eqref{R_sp_1}, 
the building block \eqref{vortex_pf_invTR_cs}, 
in the massless limit as \eqref{R_massless}, 
yields the inverse $R$-matrix \eqref{R_matrix_m} 
with a normalization factor in \eqref{c_norm_jones_knot}:
\begin{align}
(-1)^{d_{23}}\,
I_{n,f}^{T[\overline{R}^{d_{12}d_{31}}_{d_{42}d_{34}}(\bgamma)]}(\bsig;q) 
\quad \to \quad
q^{\frac14 n(n+2)}
\left(R^{-1}\right)^{d_{12}+f\, d_{31}+f}_{d_{42}+f\, d_{34}+f}
\,,
\label{invR_limit}
\end{align}
where the prefactor $(-1)^{d_{23}}$ is introduced by the complexified FI parameters.


From above, the non-trivial degenerate inverse $R$-matrices
\begin{align}
&
\left(R^{-1}\right)^{d_{32}\ 0}_{d_{42}d_{34}}=
q^{\frac12 n d_{34} - \frac14 n^2}\,
\frac{\qq{q}{n-d_{34}}}{\qq{q}{n-d_{32}}},
\label{deg_R_matrix_m1}
\\
&
\left(R^{-1}\right)^{d_{12}d_{31}}_{d_{32}\ 0}=
q^{ - d_{12}d_{31} + \frac12 n d_{31} - \frac14 n^2}\,
\frac{\qq{q}{d_{32}} \qq{q}{n}}
{\qq{q}{n-d_{12}} \qq{q}{d_{12}} \qq{q}{d_{31}}},
\label{deg_R_matrix_m2}
\\
&
\left(R^{-1}\right)^{d_{12} \ 0\ }_{\ 0\ d_{12}}
=\left(R^{-1}\right)^{\ 0\ d_{12}}_{d_{12} \ 0\ }=
q^{\frac12 n d_{12} - \frac14 n^2},
\label{deg_R_matrix_m3}
\end{align}
are also constructed by identifying some $U(1)$ gauge symmetries, 
by identifying $U(1)_1$ with $U(1)_3$ for \eqref{deg_R_matrix_m1} 
and by identifying $U(1)_4$ with $U(1)_3$ for \eqref{deg_R_matrix_m2}. 
The degenerate inverse $R$-matrix \eqref{deg_R_matrix_m3} is 
constructed in \eqref{vortex_pf_R_deg_b} just as a Chern-Simons factor.

\subsubsection{Local minimum and maximum}

The quantities $\mu_{d_{12}}$ and $\mu^{-1}_{d_{12}}$ for 
local minimum and maximum in \eqref{loc_min_max} 
are constructed, by introducing gauge/flavor-R 
Chern-Simons couplings 
$k_1^{\text{g-R}}=\pm 2$, $k_2^{\text{g-R}}=\mp 2$, 
$k_c^{\text{f-R}}=\mp 1$ and $k_f^{\text{f-R}}=\pm 2$ 
from \eqref{cs_factor_vortex} as
\begin{align}
I_{n,f}^{T[\mu_{d_{12}}]}(q)=
q^{d_{12}+f-\frac12n}=
\mu_{d_{12}+f},
\qquad
I_{n,f}^{T[\overline{\mu}_{d_{12}}]}(q)=
q^{-d_{12}-f+\frac12n}= 
\mu^{-1}_{d_{12}+f},
\label{vortex_pf_min_max}
\end{align}
where $d_1$ and $d_2$ are the associated magnetic fluxes 
for a $U(1)_1\times U(1)_2$ gauge symmetry, and 
they are labeled by $T[\mu_{d_{12}}]$ and $T[\overline{\mu}_{d_{12}}]$.

\subsection{K-theoretic vortex partitions in the knot-gauge theories}\label{sec:gluing}

\subsubsection{Summary}\label{subsec:summary_jk}

Because this section contains some technical details for the JK residue procedure and flux conditions, we first summarize what will be discussed.

From the building blocks in Section \ref{subsec:gauge_knot_build}, we can construct 
a $U(1)^{N_v}$ knot-gauge theory $T[\cK]$ 
labeled by a (1,1)-tangle diagram, with $N_v$ crossings, of knot $\cK$, 
and obtain the K-theoretic vortex partition%
\footnote{
By the Reidemeister moves I, II and III for a tangle diagram of knot $\cK$, 
infinitely many knot-gauge theories for $\cK$ are constructed, 
and they are expected to be related to one another by some 3D dualities.}
\begin{align}
I^{T[\cK]}_{\textrm{vortex}}(\bsig;\bz,\bgamma,q)=
\sum_{\bd} \left(\prod_{I=1}^{N_v} z_I^{d_I}\right)
\prod_{i} I_{n,f}^{T_i}(\bsig;q),
\label{k_vpf_gen}
\end{align}
where $\bsig=(\sigma_1, \ldots, \sigma_{N_v})$, 
$\bz=(z_1, \ldots, z_{N_v})$ are the exponentiated FI parameters 
associated with the $U(1)^{N_v}$ gauge symmetry, and
$\bgamma=(\gamma_1, \ldots, \gamma_{2N_v})$ are mass parameters. 
Here $i$ runs over all the building blocks labeled by $T_i$ in 
\eqref{vortex_pf_TR_cs}, \eqref{vortex_pf_invTR_cs} and \eqref{vortex_pf_min_max}. 
By construction, 
if the conditions \eqref{R_charge_c} and \eqref{invR_charge_c} 
for the magnetic fluxes $\bd$ are satisfied, 
the K-theoretic vortex partition yields the normalized colored Jones polynomial $J_n^{\cK}(q)$ of $\cK$ under 
$\sigma_I \to \sigma_I^*=1$, $\gamma_i \to \gamma_i^*=1$, and 
$z_I \to z_I^*=+1$ or $-1$ depending on the sign factors in \eqref{R_limit} and \eqref{invR_limit}:
\begin{align}
I^{T[\cK]}_{\textrm{vortex}}(\bsig^*;\bz^*,\bgamma^*,q)=J_n^{\cK}(q).
\label{vortex_jones}
\end{align}

Therefore, for the relation \eqref{vortex_jones}, we need to choose the stability parameters so that the JK residue in \eqref{3d_tw_pf} is taken at the locus $\sigma_I = \sigma_I^*=1$ and the domain of the magnetic fluxes $\bd$ is restricted by the conditions \eqref{R_charge_c} and \eqref{invR_charge_c}. 
In this paper, we consider a cone in \eqref{xi_cone} as a rough choice, 
where not only any choices of stability parameters inside the cone give 
the locus $\sigma_I = \sigma_I^*=1$, but also, in general, 
some choices inside the cone may also give other loci 
(see Proposition \ref{prop:other_pole} and \eqref{other_sq_31} for 
an example). 
Therefore, for establishing the relation \eqref{vortex_jones} 
we have to show that 
such other loci do not contribute to the twisted partition functions 
in the massless limit $\gamma_i \to \gamma_i^*=1$. 
In Section \ref{subsec:charge_c}, 
we discuss conditions for the magnetic fluxes $\bd$ 
which give non-zero contributions to the twisted partition functions 
in the massless limit. We first show Proposition \ref{prop:flux_positive} which implies the conditions \eqref{R_charge_c} and \eqref{invR_charge_c} for 
the JK residue at the locus $\sigma_I = \sigma_I^*=1$, and then discuss 
the contributions coming from the other loci. 
For our rough choice of the stability parameters, 
in Proposition \ref{prop:bridge} we find a class of knot diagrams 
such that the other loci do not contribute to the twisted partition functions in 
the massless limit and the relation \eqref{vortex_jones} is established. 
We expect that, by carefully choosing the stability parameters, 
the relation \eqref{vortex_jones} is, in general, 
established for any knot diagram.

\begin{remark}
In the context of knots-quivers correspondence \cite{Kucharski:2017poe, Kucharski:2017ogk}, 
the K-theoretic vortex partitions of abelian Chern-Simons-matter theories $T[Q_K]$ associated with quivers $Q_K$, which provide 
generating functions of $S^n$-colored HOMFLY-PT polynomials of knots $K$, 
and a class of 3D $\mathcal{N}=2$ dualities associated with quivers, 
are discussed in \cite{Dimofte:2010tz, Ekholm:2018eee, Ekholm:2019lmb, Jankowski:2021flt, Ekholm:2021gyu} 
(see also \cite{Gorsky:2015toa} for a different proposal of 
the relation between the K-theoretic vortex partitions and the HOMFLY-PT polynomials of torus knots).
The 3D $\mathcal{N}=2$ gauge theories $T[Q_K]$ seem to be quite different from 
the knot-gauge theories $T[\cK]$ in this paper, and 
it would be interesting to clarify the relation between them.
\end{remark}

\subsubsection{JK residue procedure}\label{subsec:jk_res}

Consider a (1,1)-tangle diagram with $N_v$ crossings. 
As we constructed in the previous section, 
the matter content at the $I$-th crossing is composed of five chiral fields 
$\Phi_1^{(I)}$, $\Phi_2^{(I)}$, $\overline{\Phi}_3^{(I)}$, 
$\overline{\Phi}_4^{(I)}$ and $\Phi_5^{(I)}$ with 
the superpotential 
$W^{(I)}=\Phi_1^{(I)}\overline{\Phi}_4^{(I)}\Phi_5^{(I)}+
\Phi_2^{(I)}\overline{\Phi}_3^{(I)}\Phi_5^{(I)}$, where 
we assume that $\Phi_1^{(I)}$, $\Phi_2^{(I)}$, 
$\overline{\Phi}_3^{(I)}$ and $\overline{\Phi}_4^{(I)}$ have 
generic masses $\gamma_1^{(I)}$, $\gamma_2^{(I)}$, 
$(\gamma_2^{(I)})^{-1}$ and $(\gamma_1^{(I)})^{-1}$, respectively, 
whereas $\Phi_5^{(I)}$ is massless. 
Let $Q_i^{(I)}$ be the $U(1)^{N_v}$ gauge charge vectors 
of $\Phi_i^{(I)}$ and $\overline{\Phi}_i^{(I)}$, 
where 
$Q_5^{(I)}$, $I=1,\ldots,N_v$, form a basis in ${\IR}^{N_v}$, 
and have relations $-Q_5^{(I)}=Q_1^{(I)}+Q_4^{(I)}=Q_2^{(I)}+Q_3^{(I)}$. 
Here we take the gauge charge vectors to be zero for 
the incoming and outgoing chiral fields of the $(1,1)$-tangle diagram. 
For the JK residue 
\cite[Theorem 2.6]{SzVe:2003} (see also \cite{Benini:2013xpa}), 
we choose the stability parameters (identified, in this paper, 
with the FI parameters) 
$\bxi=(\xi_1,\ldots,\xi_{N_v})$  
inside $\mathrm{Cone} (\bQ_5)$ which is the cone spanned by 
$\bQ_5=\big\{Q_5^{(I)}\big\}_{I=1,\ldots,N_v}$, i.e., 
\begin{align}
\bxi=\sum_{I=1}^{N_v}c_{I}\, Q_5^{(I)} \in \mathrm{Cone} (\bQ_5),
\qquad
c_{I}>0.
\label{xi_cone}
\end{align}
We now need to consider the sets of the charge vectors for the poles, 
of the integrand in the twisted partition function, 
whose cones contain the vector \eqref{xi_cone} inside.

From $I_{\textrm{1-loop}}^{R}$ in \eqref{R_1loop} or 
$I_{\textrm{1-loop}}^{\overline{R}}$ in \eqref{invR_1loop}, 
once the residue at the (referred to as \textit{massless}) pole 
``$\sigma_1\sigma_4=\sigma_2\sigma_3$'' 
 relevant to $Q_5^{(I)}$ is taken, 
the (referred to as \textit{massive}) poles relevant to 
$Q_i^{(I)}$, $i=1,2,3,4$, are moved away for generic masses. 
The charge vectors $\bQ_5$ form a basis of ${\IR}^{N_v}$, 
and the residues relevant to them 
boil down to the specializations 
$\sigma_1 = \sigma_2 = \ldots = \sigma_{N_v}=1$ 
(i.e. \eqref{R_sp_1}) of the complex scalars. 
Furthermore, as a corollary of Proposition \ref{prop:flux_positive} 
in Section \ref{subsec:charge_c}, 
when the residues at the poles relevant to $\bQ_5$ are taken, 
the flux conditions \eqref{R_charge_c} and \eqref{invR_charge_c} 
are also satisfied in the massless limit \eqref{R_massless}, 
where note that the poles relevant to $\bQ_5$ imply the non-negativity of 
the magnetic fluxes for $\Phi_5^{(I)}$ at the crossings. 
As a result, if the other contributions in the JK residue are absent, 
the K-theoretic vortex partition function yields, in the massless limit, 
the normalized colored Jones polynomial as \eqref{vortex_jones}.

Therefore, the remaining problem is, for the choice of 
the stability parameters \eqref{xi_cone}, 
whether other poles 
contribute to the twisted partition function in the massless limit.

Let $\Phi_i^{(I)}$ (resp. $\overline{\Phi}_i^{(I)}$) be 
an incoming (resp. outgoing) chiral field 
with $U(1)_R$ charge $\sfr=0$ (resp. $\sfr=2$) assigned to the $I$-th crossing, 
where $i=1$ or $2$ (resp. $3$ or $4$). 
The chiral fields associated with an arc between the over 
$I$-th crossing and the over $J$-th crossing 
(or the under $I$-th crossing and the under $J$-th crossing) as
\begin{align}
\inc{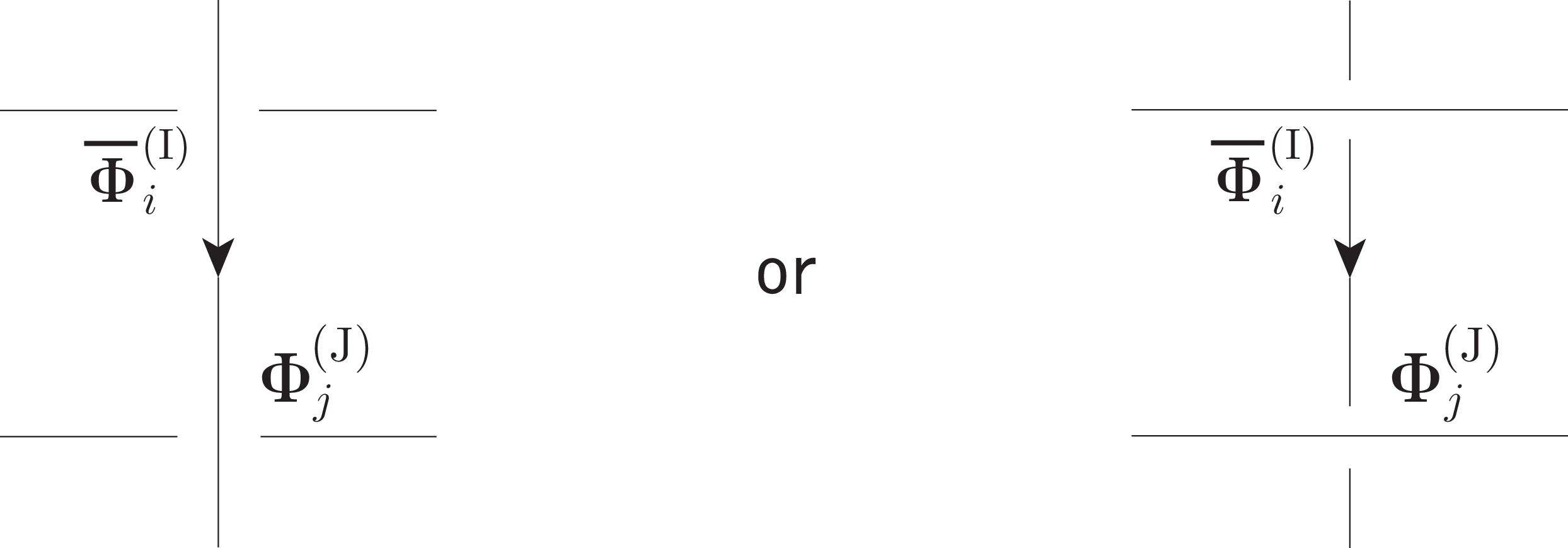}
\qquad,
\label{pair_annihilate}
\end{align}
have opposite $U(1)$ gauge charges, i.e. 
$Q^{(I)}_{i}=-Q^{(J)}_{j}$, where $i=3$ or $4$ and $j=1$ or $2$. 
On the other hand, the chiral fields associated with an arc between 
the over $I$-th crossing and the under $J$-th crossing 
(or the under $I$-th crossing and the over $J$-th crossing) as 
\begin{align}
\inc{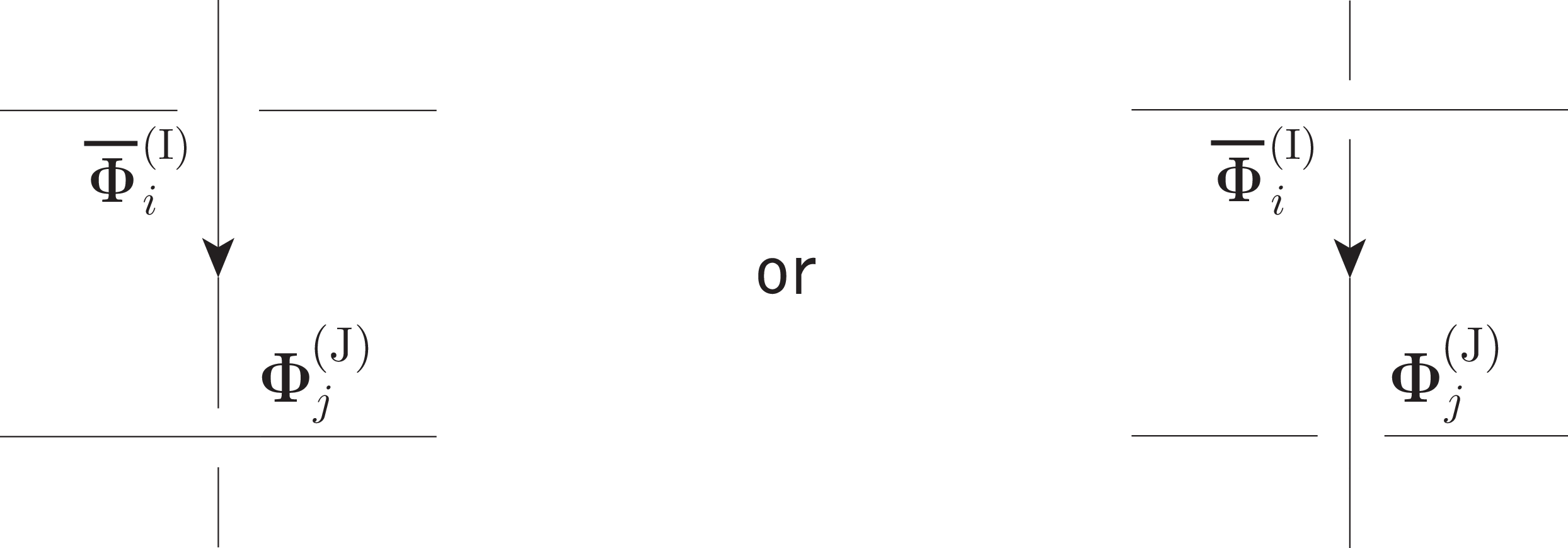}
\qquad,
\label{pair_chiral}
\end{align}
have same $U(1)$ gauge charges, i.e. 
$Q^{(I)}_{i}=Q^{(J)}_{j}$, where $i=3$ or $4$ and $j=1$ or $2$. 
Therefore, by the relations $-Q_5^{(I)}=Q_1^{(I)}+Q_4^{(I)}=Q_2^{(I)}+Q_3^{(I)}$, 
the charge vector $Q_5^{(I)}$ is expressed as 
$Q_5^{(I)}=-Q_i^{(I)} + Q_j^{(J)}=Q_k^{(I)} - Q_{\ell}^{(J)}$ 
for \eqref{pair_annihilate} or 
$Q_5^{(I)}=-Q_i^{(I)} - Q_j^{(J)}=-Q_k^{(I)} - Q_{\ell}^{(J)}$ 
for \eqref{pair_chiral}, 
where $i,j=1$ or $2$ and $k,\ell=3$ or $4$.

For the over (resp. under) $I$-th crossing, 
let $P_+^{(I)}$ (resp. $P_-^{(I)}$), 
$\overline{P}_+^{(I)}$ (resp. $\overline{P}_-^{(I)}$), 
and $Q_+^{(I)}$ (resp. $Q_-^{(I)}$) be 
the charge vectors of $\Phi_1^{(I)}$ or $\Phi_2^{(I)}$,
$\Phi_3^{(I)}$ or $\Phi_4^{(I)}$, and 
$\Phi_5^{(I)}$, respectively, where 
$P_{\pm}^{(I)}=Q_i^{(I)}$ for $i=1$ or $2$, 
$\overline{P}_{\pm}^{(I)}=Q_i^{(I)}$ for $i=3$ or $4$, 
and $Q_{\pm}^{(I)}=Q_5^{(I)}$.
We now have 
$Q_{\pm}^{(I)}=-P_{\pm}^{(I)}+P_{\pm}^{(J)}
=\overline{P}_{\pm}^{(I)}-\overline{P}_{\pm}^{(J)}$
for \eqref{pair_annihilate} and 
$Q_{\pm}^{(I)}=-P_{\pm}^{(I)}-P_{\mp}^{(J)}
=-\overline{P}_{\pm}^{(I)}-\overline{P}_{\mp}^{(J)}$ 
for \eqref{pair_chiral}. 
The charge vectors $P_{\pm}^{(I)}$ are then expressed as 
\begin{align}
P_{+}^{(I)}=-\sum_{J \ge I} Q_{+}^{(J)}+\sum_{K \ge I} Q_{-}^{(K)},
\qquad
P_{-}^{(I)}=-\sum_{J \ge I} Q_{-}^{(J)}+\sum_{K \ge I} Q_{+}^{(K)},
\label{charge_p}
\end{align} 
in terms of the basis $\bQ_5=\big\{Q_5^{(I)}\big\}_{I=1,\ldots,N_v}$, 
where the sum means that starting from 
the over (resp. under) $I$-th crossing, 
the over (resp. under) $J$-th crossings and 
the under (resp. over) $K$-th crossings pass through along 
the (1,1)-tangle diagram, 
and end at the last crossing with the bounded outgoing arc. 
Similarly, the charge vectors $\overline{P}_{\pm}^{(I)}$ are expressed as 
\begin{align}
\overline{P}_{+}^{(I)}=-\mathop{\overline{\sum}}_{J \le I} Q_{+}^{(J)}+
\mathop{\overline{\sum}}_{K \le I} Q_{-}^{(K)},
\qquad
\overline{P}_{-}^{(I)}=-\mathop{\overline{\sum}}_{J \le I} Q_{-}^{(J)}+
\mathop{\overline{\sum}}_{K \le I} Q_{+}^{(K)},
\label{charge_p_bar}
\end{align} 
where the sum means that starting from 
the over (resp. under) $I$-th crossing, 
the over (resp. under) $J$-th crossings and 
the under (resp. over) $K$-th crossings pass through, 
along the (1,1)-tangle diagram, backward, 
and end at the first crossing with the bounded incoming arc.

We now describe the (1,1)-tangle diagram by an ordered sequence of 
the charge vectors $Q_{\pm}^{(I)}$ by aligning them from 
the first crossing with the bounded incoming arc to 
the last crossing with the bounded outgoing arc along the tangle diagram. 
For convenience, we refer to the sequence as \textit{original sequence}. 
For example, the $(1,1)$-tangle diagram of the trefoil knot $\knot{3}{1}$ in 
Figure \ref{fig:tangle_31} is described by a sequence
\begin{align}
Q_+^{(1)},\ Q_-^{(2)},\ Q_+^{(3)},\ 
Q_-^{(1)},\ Q_+^{(2)},\ Q_-^{(3)}.
\label{sq_31}
\end{align}
The following proposition is then proved.

\begin{prop}\label{prop:sequence_adjacent}
For the stability parameters \eqref{xi_cone}, 
if, in the original sequence, 
there exists a cyclic sequence 
$\{Q_+^{(I_1)}, Q_-^{(I_2)}\}$, $\{Q_+^{(I_2)}, Q_-^{(I_3)}\}$, $\ldots$, 
$\{Q_+^{(I_{M-1})}, Q_-^{(I_{M})}\}$, $\{Q_+^{(I_{M})}, Q_-^{(I_1)}\}$ 
consisted of adjacent pairs $\{Q_+^{(I)}, Q_-^{(J)}\}$, 
the residue at a massless pole relevant to 
one of the charge vectors $Q_5^{(I_k)}$, $k=1,\ldots,M$, should be taken.%
\footnote{The order of the adjacent charge vectors $Q_+^{(I)}$ and $Q_-^{(J)}$ is not assumed.}
\end{prop}
\begin{proof}
Consider the cyclic sequence in the assertion. 
By \eqref{charge_p} and \eqref{charge_p_bar}, 
all the charge vectors 
$Q_5^{(I)}$, $P_{\pm}^{(I)}$, $\overline{P}_{\pm}^{(I)}$ 
(i.e. $Q_i^{(I)}$) other than $Q_5^{(I_k)}$ satisfy 
$\sum_{k=1}^{M} Q_5^{(I_k)}=-1$ or $0$ in terms of the basis $\bQ_5$.
This means that any cones consist of the charge vectors $Q_i^{(I)}$ 
without $Q_5^{(I_k)}$ do not contain the vector \eqref{xi_cone} inside, 
and the residues at a massless pole relevant to 
one of the charge vectors $Q_5^{(I_k)}$ should be taken.
\end{proof}

Following Proposition \ref{prop:sequence_adjacent} we should take the residue at 
the massless pole relevant to a charge vector $Q_5^{(I)}$ 
in each cyclic sequence, 
and then the massive pole around the $I$-th crossing is moved away. 
Therefore, for discussing cones 
which contain the vector \eqref{xi_cone} inside, 
it is enough to consider the subsequences 
(referred to as \textit{reduced sequences}) extracted 
by removing the charge vectors $Q_5^{(I)}$ from the original sequence.%
\footnote{By definition, the reduced sequences do not contain 
cyclic sequences in Proposition \ref{prop:sequence_adjacent}.} 
For example, for the original sequence \eqref{sq_31} of $\knot{3}{1}$ 
a reduced sequence
\begin{align}
Q_+^{(1)},\ Q_+^{(3)},\ Q_-^{(1)},\ Q_-^{(3)},
\label{red_sq_31}
\end{align}
is found by removing $Q_5^{(2)}$. 
When no reduced sequences are found, it is clear that, 
for the stability parameters \eqref{xi_cone}, 
there are no cones, other than $\mathrm{Cone} (\bQ_5)$, 
which contain the vector \eqref{xi_cone} inside. 
If such cones, other than $\mathrm{Cone} (\bQ_5)$, exist, 
some reduced sequences should be found, and 
the following proposition is proved.

\begin{prop}\label{prop:other_pole}
If there exist cones, other than $\mathrm{Cone} (\bQ_5)$, 
which contain the vector \eqref{xi_cone} inside, 
the spans of the cones should contain 
both some $P_{\pm}^{(I)}$ and some $\overline{P}_{\pm}^{(J)}$.
\end{prop}
\begin{proof}
To show the assertion, consider a cone consisted of
$\bR=\{R_I\}_{I=1,\ldots,N_v}=\{P_{\pm}^{(I_{\ell})}\}_{\ell=1,\ldots,L} \cup \bQ_5\backslash \{Q_5^{(I_{\ell})}\}_{\ell=1,\ldots,L}$ or 
$\overline{\bR}=\{\overline{R}_I\}_{I=1,\ldots,N_v}=\{\overline{P}_{\pm}^{(I_{\ell})}\}_{\ell=1,\ldots,L} \cup \bQ_5\backslash \{Q_5^{(I_{\ell})}\}_{\ell=1,\ldots,L}$.

For the cone spanned by $\bR$, starting from the first incoming arc, 
along the tangle in order, consider a part of 
the ordered sequence of charge vectors in 
$\{Q_5^{(I_{\ell})}\}_{\ell=1,\ldots,L}$,
\begin{align}
\begin{split}
&
\ldots, Q_+^{(J)}, \ Q_-^{(J_1)}, \ Q_-^{(J_2)}, \ \ldots, \ Q_-^{(J_N)}, \ \ldots,
\\
\textrm{or}\quad
&
\ldots, Q_-^{(J)}, \ Q_+^{(J_1)}, \ Q_+^{(J_2)}, \ \ldots, \ Q_+^{(J_N)}, \ \ldots,
\label{charge_sq}
\end{split}
\end{align}
which is referred to as a \textit{forward subsequence}.
Similarly, for the cone spanned by $\overline{\bR}$, consider
\begin{align}
\begin{split}
&
\ldots, Q_-^{(J_1)}, \ Q_-^{(J_2)}, \ \ldots, \ Q_-^{(J_N)}, \ Q_+^{(J)}, \ \ldots,
\\
\textrm{or}\quad
&
\ldots, Q_+^{(J_1)}, \ Q_+^{(J_2)}, \ \ldots, \ Q_+^{(J_N)}, \ Q_-^{(J)}, \ \ldots,
\label{charge_sq_bw}
\end{split}
\end{align}
which is referred to as a \textit{backward subsequence}.
The number of forward and backward subsequences is finite, 
and for all such subsequences as \eqref{charge_sq} and \eqref{charge_sq_bw} 
we consider pairs $\{Q_{+}^{(J)}, Q_{-}^{(J_{\ell})}\}$ 
or $\{Q_{+}^{(J_{\ell})}, Q_{-}^{(J)}\}$, $\ell=1,\ldots,N$. 
The tangle passes through each crossing twice, and 
we then find at least one unbounded sequence of pairs as 
$\{Q_+^{(K_2)}, Q_-^{(K_1)}\}$, $\{Q_+^{(K_3)}, Q_-^{(K_2)}\}$, 
$\{Q_+^{(K_4)}, Q_-^{(K_3)}\}$, $\ldots$ or
$\{Q_+^{(K_1)}, Q_-^{(K_2)}\}$, $\{Q_+^{(K_2)}, Q_-^{(K_3)}\}$, 
$\{Q_+^{(K_3)}, Q_-^{(K_4)}\}$, $\ldots$.%
\footnote{As an example, if the ordered sequence starts as 
$Q_+^{(J_1)}, Q_+^{(J_2)}, Q_-^{(J_3)}, \ldots$, 
the charge vectors except $Q_{\pm}^{(J_1)}$ can have pairs. 
In this case, by taking a first pair 
$\{Q_+^{(J_2)}, Q_-^{(J_3)}\}$, it is possible subsequently to find 
unbounded sequence of pairs as 
$\{Q_+^{(J_2)}, Q_-^{(J_3)}\}$, $\{Q_+^{(J_m)}, Q_-^{(J_2)}\}$, 
$\{Q_+^{(J_n)}, Q_-^{(J_m)}\}, \ldots$, $m, n\ne 1$.}
Furthermore, since this sequence is finite, 
as a subsequence of it, a cyclic sequence 
$\{Q_+^{(I_1)}, Q_-^{(I_2)}\}$, $\{Q_+^{(I_2)}, Q_-^{(I_3)}\}, \ldots$, 
$\{Q_+^{(I_{M-1})}, Q_-^{(I_{M})}\}$, $\{Q_+^{(I_{M})}, Q_-^{(I_1)}\}$ 
should be obtained. 
Because any charge vectors $P_{\pm}^{(I)}$ in the set $\bR$ 
(resp. $\overline{P}_{\pm}^{(I)}$ in the set $\overline{\bR}$),
\begin{align}
P_{\pm}^{(I)}\ (\textrm{resp.}\ \overline{P}_{\pm}^{(I)})=
\sum_{\ell=1}^{M}\left(\alpha_{\pm, \ell}^{(I)}Q_+^{(I_{\ell})}+
\beta_{\pm, \ell}^{(I)}Q_-^{(I_{\ell})}\right) + \cdots
=
\sum_{\ell=1}^{M}\left(\alpha_{\pm, \ell}^{(I)}+
\beta_{\pm, \ell}^{(I)}\right)Q_5^{(I_{\ell})} + \cdots,
\label{p_charge_rel}
\end{align} 
are expressed as \eqref{charge_p} (resp. \eqref{charge_p_bar}) 
in terms of the basis $\bQ_5$, we see that
$\sum_{\ell=1}^{M}\big(\alpha_{\pm, \ell}^{(I)}+\beta_{\pm, \ell}^{(I)}\big)=-1$
if $P_{\pm}^{(I)}$ (resp. $\overline{P}_{\pm}^{(I)}$) 
is in between a pair in the cyclic sequence or 
$\sum_{\ell=1}^{M}\big(\alpha_{\pm, \ell}^{(I)}+\beta_{\pm, \ell}^{(I)}\big)=0$ 
if otherwise, where $\cdots$ does not contain 
$Q_5^{(I_{\ell})}$, $\ell=1,\ldots,M$. 
Therefore, the cones spanned by $\bR$ and $\overline{\bR}$ do not contain 
the vector \eqref{xi_cone} inside, and our assertion follows.
\end{proof}

The forward and backward subsequences in reduced sequences 
can be used to find cones other than $\mathrm{Cone} (\bQ_5)$. 
As an example, 
for the reduced sequence \eqref{red_sq_31} of $\knot{3}{1}$ we find 
a cyclic sequence $\{Q_+^{(3)}, Q_-^{(3)}\}$ from a forward subsequence and 
a cyclic sequence $\{Q_+^{(1)}, Q_-^{(1)}\}$ from a backward subsequence. 
Then, from the proof of Proposition \ref{prop:other_pole}, 
by considering $\overline{P}_-^{(1)}$ for $\{Q_+^{(3)}, Q_-^{(3)}\}$ and 
$P_+^{(3)}$ for $\{Q_+^{(1)}, Q_-^{(1)}\}$, 
we find a cone, which contains the vector \eqref{xi_cone} inside, 
consisted of
\begin{align}
\overline{P}_-^{(1)}=Q_5^{(3)}-Q_5^{(2)},\quad 
P_+^{(3)}=Q_5^{(1)}-Q_5^{(2)},\quad
Q_5^{(2)},
\label{other_sq_31}
\end{align}
for the original sequence \eqref{sq_31} of $\knot{3}{1}$. 
Therefore, we need to exclude this type of possibility
\begin{itemize}
\item[{\bf 1)}]
by showing that the contributions like \eqref{other_sq_31} 
coming from cones other than $\mathrm{Cone} (\bQ_5)$ vanish, or;

\item[{\bf 2)}]
if they do not vanish, 
by refining the choice of the stability parameters \eqref{xi_cone}.
\end{itemize}
In the next subsection, we discuss 
flux conditions for the non-zero contributions to 
the twisted partition function, and then consider 
the first option.
Actually, we will see that the cone \eqref{other_sq_31} for $\knot{3}{1}$ 
does not contribute in the massless limit (see Proposition \ref{prop:bridge}).

\subsubsection{Flux conditions}\label{subsec:charge_c}

For flux conditions, the following proposition is proved.

\begin{prop}\label{prop:flux_positive}
When the magnetic fluxes for $\Phi_5^{(I)}$ at all the crossings are non-negative, 
the flux conditions \eqref{R_charge_c} and \eqref{invR_charge_c} 
are satisfied in the massless limit.
\end{prop}
\begin{proof}
Assume that the magnetic fluxes for $\Phi_5^{(I)}$ 
at all the crossings are non-negative. 
When the magnetic flux at the $I$-th crossing is non-negative, 
i.e. the chiral field $\Phi_5^{(I)}$ gives a pole, at least, 
either $\Phi_1^{(I)}$ or $\Phi_4^{(I)}$ and 
either $\Phi_2^{(I)}$ or $\Phi_3^{(I)}$ give zeros as implied by 
the relations $-Q_5^{(I)}=Q_1^{(I)}+Q_4^{(I)}=Q_2^{(I)}+Q_3^{(I)}$ 
(the charges for $U(1)_{c}$ and $U(1)_{ext}$ also satisfy 
the same relations). 
Therefore, under the non-negative flux assumption, 
the number of zeros $Z$ satisfies $Z \ge 2N_v$ and 
the number of poles $P$ satisfies $P \le 3N_v$, where 
$N_v=(Z+P)/5$ is the number of crossings. 
Because we take the residues at $N_v$ poles, 
the contributions to the twisted partition function vanish, 
in the massless limit, if $Z > 2N_v$ ($P < 3N_v$), i.e. 
there exists a crossing such that both $\Phi_1^{(I)}$ and $\Phi_4^{(I)}$ 
or both $\Phi_2^{(I)}$ and $\Phi_3^{(I)}$ give zeros.

Let $d_{in}$ and $d_{out}$ be the magnetic fluxes 
for an incoming $\sfr=0$ chiral field $\Phi$ and an outgoing $\sfr=2$ chiral 
field $\overline{\Phi}$ for an (inverse) $R$-matrix 
assigned to a crossing, respectively. 
The non-negative flux assumption implies that 
the magnetic fluxes $d_{in}$ and $d_{out}$ 
for the chiral fields along an under (resp. over) crossing arc 
satisfy $d_{in} \le d_{out}$ (resp. $d_{out} \le d_{in}$):
\begin{align}
\incc{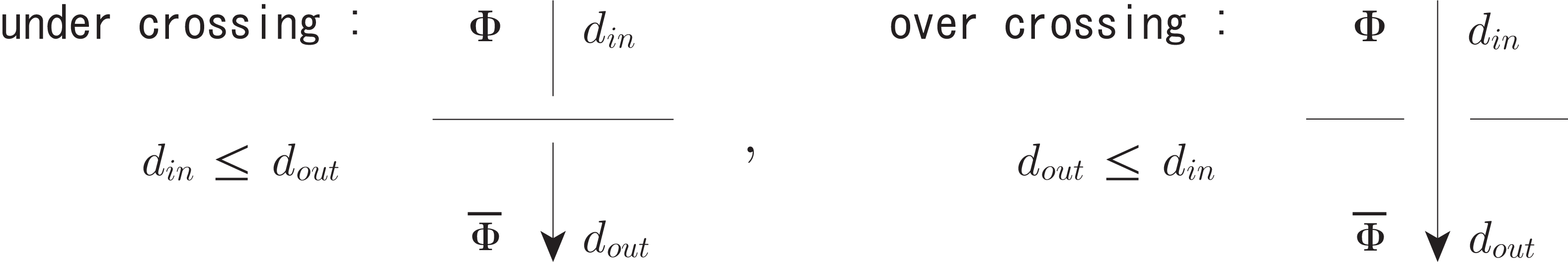}
\quad.
\label{charge_c}
\end{align}
For the under ($I$-th) crossing, if 
$\Phi^{(I)}$ gives a zero and $\overline{\Phi}^{(I)}$ gives a pole, 
i.e. $d_{in}^{(I)} \le d_{out}^{(I)} < 0$, then 
$d_{in}^{(I+1)}(=d_{out}^{(I)}) < 0$ at the next ($(I+1)$-th) crossing 
which results in $d_{out}^{(I+1)} < 0$ for the non-zero contributions 
to the twisted partition function in the massless limit.
Since the last outgoing magnetic flux at the last crossing 
is given by a background magnetic flux 
$f \in \{0,1,\ldots,n\}$ for the global symmetry $U(1)_{ext}$
(e.g. see Figure \ref{fig:tangle_31}), 
the gluing procedure excludes the case 
$d_{in}^{(I)} \le d_{out}^{(I)} < 0$ for the under crossing arc. 
Similarly, the case $n< d_{out} \le d_{in}$ 
for an over crossing arc is also excluded. 
As a result, the magnetic fluxes are constrained by the conditions
\begin{align}
0 \le d_{in} \le d_{out}\ \
\textrm{(under crossing)},
\qquad
d_{out} \le d_{in} \le n\ \
\textrm{(over crossing)},
\label{gauge_R_flux}
\end{align}
which mean all $\Phi^{(I)}$ give poles and all $\overline{\Phi}^{(I)}$ give zeros.

Starting from a crossing with the magnetic fluxes $d_{in}$ and $d_{out}$ 
with the conditions \eqref{gauge_R_flux}, 
in the sequel of gluing procedure along the tangle in order, 
if the outgoing under (resp. over) crossing arc is 
glued with an over (resp. under) crossing incoming arc first at 
the $K$-th crossing 
we find $0 \le d_{in} \le d_{out} \le \ldots \le d_K \le n$ 
(resp. $0 \le d_K \le \ldots \le d_{out} \le d_{in} \le n$), 
where $d_K$ is the magnetic flux for the incoming chiral field 
assigned to the $K$-th crossing. 
Even if the outgoing under (resp. over) crossing arcs are only  
glued with under (resp. over) crossing incoming arcs, 
because the gluing procedure ends up with
the background magnetic flux $f \in \{0,1,\ldots,n\}$, 
we find $0 \le d_{in},  d_{out} \le n$ anyway. 
This gives the desired conditions in \eqref{R_charge_c} and \eqref{invR_charge_c}.
\end{proof}

Proposition \ref{prop:other_pole} and \ref{prop:flux_positive} imply 
the following corollary.

\begin{cor}\label{cor:flux_negative_0}
For the stability parameters \eqref{xi_cone}, 
the cones other than $\mathrm{Cone} (\bQ_5)$ do not contribute to 
the twisted partition function in the massless limit 
when the magnetic fluxes for $\Phi_5^{(I)}$ at all the crossings are non-negative. 
\end{cor}

Next, we consider the cases with negative fluxes for some $\Phi_5^{(I)}$. 
Let $E_v$ be the number of negative fluxes for $\Phi_5^{(I)}$, and then 
the number of non-negative fluxes for $\Phi_5^{(I)}$ is $N_v-E_v$, 
where $N_v$ is the number of crossings. 
As the proof of Proposition \ref{prop:flux_positive}, we see that 
the number of zeros $Z$ satisfies $Z \ge 2(N_v-E_v)+E_v=2N_v-E_v$. 
Therefore, for the condition $Z > 2N_v$ of the vanishing contributions to 
the twisted partition function in the massless limit, 
it needs, at least, $E_v+1$ extra zeros in addition to 
the minimal number of zeros.
From a reduced sequence, let us extract the subsequence composed of 
all of the $E_v$ charge vectors with negative fluxes for $\Phi_5^{(I)}$, 
and, in what follows, we refer to it as \textit{negative sequence}:
\begin{align}
\textrm{original sequence}\ \supset\
\textrm{reduced sequence}\ \supset\
\textrm{negative sequence}.
\end{align}
For negative sequences, the following proposition is proved.

\begin{prop}\label{prop:flux_negative}
For a negative sequence, in addition to the minimal number of zeros 
$2N_v-E_v$, there is at least one extra zero
\begin{itemize}
\item[(1)] 
after the last charge vector $Q_+^{(I)}$ or $Q_-^{(I)}$ 
in the negative sequence;

\item[(2)] 
between adjacent charge vectors $Q_+^{(I)}$ and $Q_+^{(J)}$ or 
$Q_-^{(I)}$ and $Q_-^{(J)}$ in the negative sequence;

\item[(3)] 
for each cyclic sequence 
$\{Q_+^{(I_1)}, Q_-^{(I_2)}\}$, $\{Q_+^{(I_2)}, Q_-^{(I_3)}\}$, $\ldots$, 
$\{Q_+^{(I_{M-1})}, Q_-^{(I_{M})}\}$, $\{Q_+^{(I_{M})}, Q_-^{(I_1)}\}$ 
consisted of adjacent pairs $\{Q_+^{(I)}, Q_-^{(J)}\}$ in the negative sequence, 
when the stability parameters \eqref{xi_cone} are assumed.
\end{itemize}
\end{prop}
\begin{proof}
(1) Consider the last charge vector $Q_+^{(I)}$ (resp. $Q_-^{(I)}$) in 
a negative sequence. 
If the $\sfr=2$ chiral field with the charge vector $\overline{P}_+^{(I)}$ 
(resp. $\overline{P}_-^{(I)}$) gives a pole, the relevant flux 
$d_{out}$ satisfies $d_{out}>n$ (resp. $d_{out}<0$). 
Since the last outgoing $\sfr=2$ chiral field has a background magnetic flux 
$f \in \{0,1,\ldots,n\}$, by \eqref{charge_c}, 
there should be a charge vector $Q_+^{(K)}$ (resp. $Q_-^{(K)}$) 
with non-negative flux in the original sequence 
after $Q_+^{(I)}$ (resp. $Q_-^{(I)}$), and 
both $\sfr=0$ chiral field with the charge vector $P_+^{(K)}$ 
(resp. $P_-^{(K)}$) and 
$\sfr=2$ chiral field with the charge vector $\overline{P}_+^{(K)}$ 
(resp. $\overline{P}_-^{(K)}$) give zeros. This shows the assertion (1).

(2) Assume that there are adjacent charge vectors 
$Q_+^{(I)}$ and $Q_+^{(J)}$ in a negative sequence. 
If both $\sfr=2$ chiral field with the charge vector $\overline{P}_+^{(I)}$ 
and $\sfr=0$ chiral field with the charge vector $P_+^{(J)}$
give pole, the former relevant flux $d_{out}$ satisfies $d_{out}>n$ whereas 
the latter relevant flux $d_{in}$ satisfies $d_{in}<n$. 
Therefore, similarly to the proof of the assertion (1), 
there should be at least one extra zero between $Q_+^{(I)}$ and $Q_+^{(J)}$. 
Similarly, the assertion for adjacent charge vectors $Q_-^{(I)}$ and $Q_-^{(J)}$ is proved.

(3) Assume that there is the cyclic sequence in the assertion. 
To have a cone which contains the vector \eqref{xi_cone} inside, 
by considering $\sum_{k=1}^{M} Q_5^{(I_k)}$ 
(cf. the proof of Proposition \ref{prop:sequence_adjacent}), 
we see that there should be a charge vector $Q_+^{(K)}$ (resp. $Q_-^{(K)}$) 
with non-negative flux in the original sequence between 
\begin{itemize}
\item[(i)]
a pair of adjacent charge vectors in the negative sequence 
ordered as $Q_+^{(I_k)}$, $Q_-^{(I_{k+1})}$, 
where the $\sfr=0$ (resp. $\sfr=2$) chiral field with the charge vector $P_+^{(K)}$ (resp. $\overline{P}_-^{(K)}$) gives a pole, or; 

\item[(ii)]
a pair of adjacent charge vectors in the negative sequence 
ordered as $Q_-^{(I_{k+1})}$, $Q_+^{(I_k)}$, 
where the $\sfr=2$ (resp. $\sfr=0$) chiral field with the charge vector $\overline{P}_+^{(K)}$ (resp. $P_-^{(K)}$) gives a pole.
\end{itemize}
In the case (i), if the $\sfr=2$ chiral field with 
the charge vector $\overline{P}_+^{(I_k)}$ gives a pole, 
the relevant flux $d_{out}$ satisfies $d_{out}>n$ 
whereas the flux $d_{in}$ relevant to $P_+^{(K)}$ (resp. $P_-^{(K)}$) 
satisfies $d_{in}<n$ (resp. $d_{in}<0$). Therefore, at least 
one extra zero between $Q_+^{(I_k)}$ and $Q_+^{(K)}$ (resp. $Q_-^{(K)}$) 
should be found. 
Similarly, in the case (ii), at least 
one extra zero between $Q_-^{(I_{k+1})}$ and $Q_+^{(K)}$ (resp. $Q_-^{(K)}$) 
should be found.
\end{proof}

When $E_v=1$, the ordered negative sequences are $(Q_{\pm}^{(1)}, Q_{\mp}^{(1)})$, and 
Proposition \ref{prop:flux_negative} implies, at least, 
two extra zeros in addition to $2N_v-1$ zeros. 
When $E_v=2$, the ordered negative sequences are 
\begin{align}
(Q_{\pm}^{(1)}, Q_{\pm}^{(2)}, Q_{\mp}^{(1)}, Q_{\mp}^{(2)}),\ 
(Q_{\pm}^{(1)}, Q_{\mp}^{(1)}, Q_{\pm}^{(2)}, Q_{\mp}^{(2)}),\ 
(Q_{\pm}^{(1)}, Q_{\mp}^{(2)}, Q_{\pm}^{(2)}, Q_{\mp}^{(1)}),\ 
(Q_{\pm}^{(1)}, Q_{\mp}^{(2)}, Q_{\mp}^{(1)}, Q_{\pm}^{(2)}),
\label{ev2a}
\end{align}
which imply, at least, three extra zeros in addition to $2N_v-2$ zeros, and 
\begin{align}
(Q_{\pm}^{(1)}, Q_{\pm}^{(2)}, Q_{\mp}^{(2)}, Q_{\mp}^{(1)}),\
(Q_{\pm}^{(1)}, Q_{\mp}^{(1)}, Q_{\mp}^{(2)}, Q_{\pm}^{(2)}),
\label{ev2b}
\end{align} 
which imply, at least, four extra zeros in addition to $2N_v-2$ zeros. 
Therefore, we have the following corollary.

\begin{cor}\label{cor:flux_negative_12}
For the stability parameters \eqref{xi_cone}, 
the cones other than $\mathrm{Cone} (\bQ_5)$ do not contribute to 
the twisted partition function in the massless limit 
when the number of negative fluxes for $\Phi_5^{(I)}$ is up to $E_v=2$. 
\end{cor}

From the negative sequences \eqref{ev2a} and \eqref{ev2b} with $E_v=2$, 
we make negative sequences with $E_v=E_1+E_2$ by replacements
\begin{align}
Q_{\pm}^{(1)}\ \to \
Q_{\pm}^{(1,J_1^{\pm})},\ Q_{\pm}^{(1,J_2^{\pm})},\ \ldots,\ Q_{\pm}^{(1,J_{E_1}^{\pm})},\quad
Q_{\pm}^{(2)}\ \to \
Q_{\pm}^{(2,K_1^{\pm})},\ Q_{\pm}^{(2,K_2^{\pm})},\ \ldots,\ Q_{\pm}^{(2,K_{E_2}^{\pm})},
\label{replace_ev2}
\end{align}
where $\{J_1^{\pm}, J_2^{\pm}, \ldots, J_{E_1}^{\pm}\}=\{1,2,\ldots,E_1\}$ and 
$\{K_1^{\pm}, K_2^{\pm}, \ldots, K_{E_2}^{\pm}\}=\{1,2,\ldots,E_2\}$.
We see that Proposition \ref{prop:flux_negative} also implies the following corollary.

\begin{cor}\label{cor:flux_negative_gen}
For the stability parameters \eqref{xi_cone}, 
the cones other than $\mathrm{Cone} (\bQ_5)$ do not contribute to 
the twisted partition function in the massless limit 
when the negative sequences take the above forms 
by the replacements \eqref{replace_ev2}.
\end{cor}

\begin{figure}[t]
\centering
\includegraphics[width=125mm]{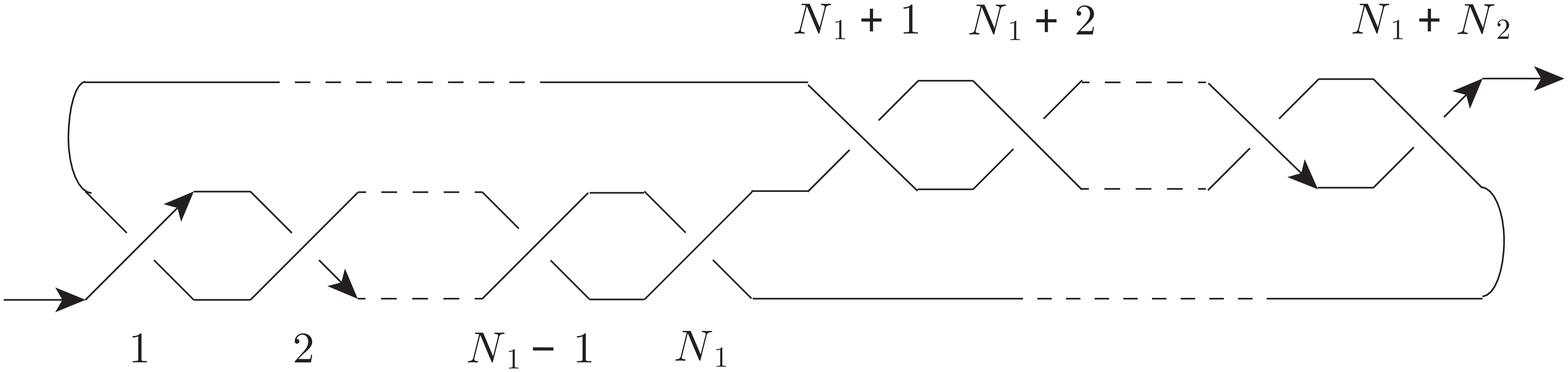}
\caption{A $(1,1)$-tangle diagram of a 2-bridge knot (rational knot) with 
$(N_1, N_2)$ twists, where $N_1 \ge N_2 \ge 0$.}
\label{fig:2_bridge}
\end{figure}

Let us consider a $(1,1)$-tangle diagram of a 2-bridge knot (rational knot) with 
$(N_1, N_2)$ twists ($N_1\ge N_2 \ge 0$) in Figure \ref{fig:2_bridge} whose 
original sequence is one of the following forms:
\begin{align}
\begin{split}
&
\big(Q_{+}^{(1)},\ Q_{-}^{(2)},\ \ldots,\ Q_{-}^{(N_1)};\ 
Q_{+}^{(N_1+N_2)},\ Q_{-}^{(N_1+N_2-1)},\ \ldots,\ Q_{-}^{(N_1+1)};
\\
& \hspace{3em}
Q_{+}^{(N_1)},\ Q_{-}^{(N_1-1)},\ \ldots,\ Q_{-}^{(1)};\ 
Q_{+}^{(N_1+1)},\ Q_{-}^{(N_1+2)},\ \ldots,\ Q_{-}^{(N_1+N_2)}\big),
\\
&
\big(Q_{+}^{(1)},\ Q_{-}^{(2)},\ \ldots,\ Q_{+}^{(N_1)};\ 
Q_{-}^{(N_1+1)},\ Q_{+}^{(N_1+2)},\ \ldots,\ Q_{+}^{(N_1+N_2)};
\\
& \hspace{3em}
Q_{-}^{(N_1)},\ Q_{+}^{(N_1-1)},\ \ldots,\ Q_{-}^{(1)};\ 
Q_{+}^{(N_1+1)},\ Q_{-}^{(N_1+2)},\ \ldots,\ Q_{-}^{(N_1+N_2)}\big),
\\
&
\big(Q_{+}^{(1)},\ Q_{-}^{(2)},\ \ldots,\ Q_{-}^{(N_1)};\ 
Q_{+}^{(N_1+N_2)},\ Q_{-}^{(N_1+N_2-1)},\ \ldots,\ Q_{+}^{(N_1+1)};
\\
& \hspace{3em}
Q_{-}^{(1)},\ Q_{+}^{(2)},\ \ldots,\ Q_{+}^{(N_1)};\ 
Q_{-}^{(N_1+1)},\ Q_{+}^{(N_1+2)},\ \ldots,\ Q_{-}^{(N_1+N_2)}\big),
\label{bridge_sq}
\end{split}
\end{align}
where the first, second, and third cases correspond to the cases with 
$(N_1, N_2)=(\textrm{even}, \textrm{even})$, 
$(N_1, N_2)=(\textrm{odd}, \textrm{even})$, and 
$(N_1, N_2)=(\textrm{even}, \textrm{odd})$, respectively. 
In particular, they include twist knots e.g. as
\begin{align}
\begin{split}
&
(N_1,N_2)=(2,1):\ \ \knot{3}{1},\quad
(N_1,N_2)=(2,2):\ \ \knot{4}{1},\quad
(N_1,N_2)=(3,2):\ \ \knot{5}{2},
\\
&
(N_1,N_2)=(4,2):\ \ \knot{6}{1},\quad
(N_1,N_2)=(5,2):\ \ \knot{7}{2},\quad
(N_1,N_2)=(6,2):\ \ \knot{8}{1},\ \ etc.,
\label{tw_list}
\end{split}
\end{align}
in the Rolfsen table. 
For the original sequences in \eqref{bridge_sq}, because the reduced sequences take the forms constructed by the replacements \eqref{replace_ev2}, 
any negative sequences for them also take the forms constructed by the replacements \eqref{replace_ev2}. 
Then, as a result of Corollary \ref{cor:flux_negative_gen}, 
we find the following proposition.

\begin{prop}\label{prop:bridge}
For the stability parameters \eqref{xi_cone}, 
the factorization of the twisted partition function on $\mathbb{S}^2 \times_{q} \mathbb{S}^1$ for 
the $(1,1)$-tangle diagram of the 2-bridge knot in Figure \ref{fig:2_bridge} 
gives the K-theoretic vortex partition function which agrees, 
in the massless limit 
and the exponentiated FI parameters $z_I \to +1$ or $-1$ limit 
depending on the prefactors in \eqref{R_limit} and \eqref{invR_limit}, 
with the normalized colored Jones polynomial of the 2-bridge knot.
\end{prop}

In the next subsection, we describe an explicit computation for the trefoil knot $\knot{3}{1}$ as 
the simplest example of Proposition \ref{prop:bridge}.

\begin{remark}\label{rem:general_knot}
When we consider more general tangle diagrams, 
for a choice of the stability parameters in \eqref{xi_cone}, 
some cones other than $\mathrm{Cone} (\bQ_5)$ may contribute to 
the twisted partition function in the massless limit. 
For such cases, the second option described at the end of Section \ref{subsec:jk_res} (i.e. more appropriate choices of 
the stability parameters) should be considered. 
We leave them for future research.
\end{remark}

\subsection{Examples}\label{subsec:exs}

\subsubsection{Trefoil knot $\knot{3}{1}$}\label{subsec:gen_3_1}

\begin{figure}[t]
\begin{minipage}{0.30\textwidth}
\begin{center}
\includegraphics[width=35mm]{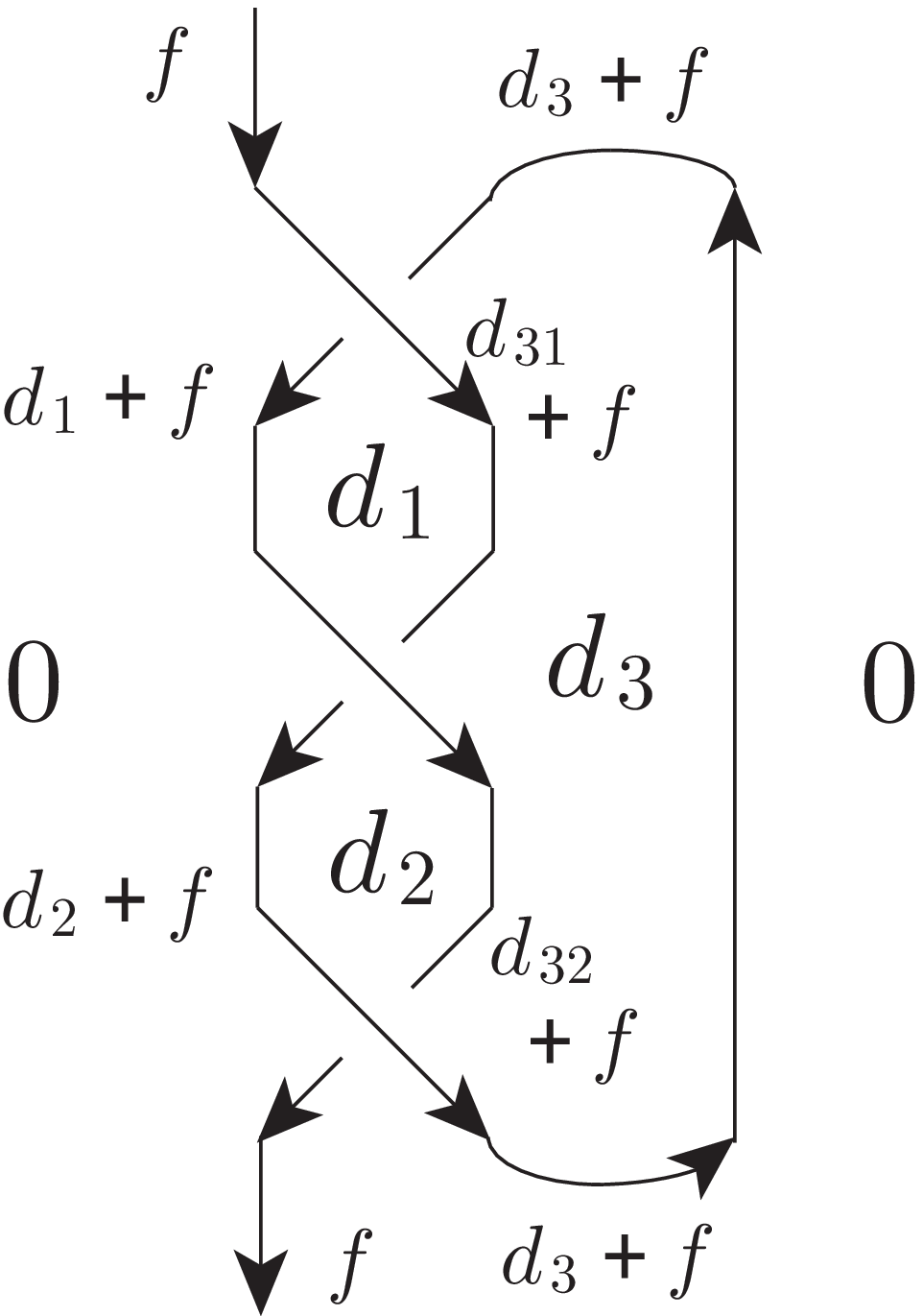}
\caption{$(1,1)$-tangle diagram of the trefoil knot $\knot{3}{1}$, 
where $f \in \{0,1,\ldots, n\}$ is fixed.}
\label{fig:3_1_add}
\end{center}
\end{minipage}
\hspace{0.5em}
\begin{minipage}{0.66\textwidth}
\begin{center}
\makeatletter
\def\@captype{table}
\makeatother
\begin{small}
\begin{tabular}{|c|c|c|c||c|c|c||c|}
\hline
Field & $U(1)_{1}$ & $U(1)_{2}$ & $U(1)_{3}$ & $U(1)_{c}$ & $U(1)_{ext}$ & mass & $U(1)_R$ 
\\ 
\hline
$\Phi_{1}^{(1)}$ & $0$ & $0$ & $0$ & $1$ & $-1$ & $\gamma_{1,1}$ & $0$ \\
$\Phi_{2}^{(1)}$ & $0$ & $0$ & $1$ & $0$ & $1$ & $\gamma_{1,2}$ & $0$ \\
$\overline{\Phi}_{3}^{(1)}$ & $-1$ & $0$ & $0$ & $0$ & $-1$ & $\gamma_{1,2}^{-1}$ & $2$ \\
$\overline{\Phi}_{4}^{(1)}$ & $-1$ & $0$ & $1$ & $-1$ & $1$ & $\gamma_{1,1}^{-1}$ & $2$ \\
$\Phi_{5}^{(1)}$ & $1$ & $0$ & $-1$ & $0$ & $0$ & $1$ & $0$ \\
\hline
\hline
$\Phi_{1}^{(2)}$ & $-1$ & $0$ & $0$ & $1$ & $-1$ & $\gamma_{2,1}$ & $0$ \\
$\Phi_{2}^{(2)}$ & $-1$ & $0$ & $1$ & $0$ & $1$ & $\gamma_{2,2}$ & $0$ \\
$\overline{\Phi}_{3}^{(2)}$ & $0$ & $-1$ & $0$ & $0$ & $-1$ & $\gamma_{2,2}^{-1}$ & $2$ \\
$\overline{\Phi}_{4}^{(2)}$ & $0$ & $-1$ & $1$ & $-1$ & $1$ & $\gamma_{2,1}^{-1}$ & $2$ \\
$\Phi_{5}^{(2)}$ & $1$ & $1$ & $-1$ & $0$ & $0$ & $1$ & $0$ \\
\hline
\hline
$\Phi_{1}^{(3)}$ & $0$ & $-1$ & $0$ & $1$ & $-1$ & $\gamma_{3,1}$ & $0$ \\
$\Phi_{2}^{(3)}$ & $0$ & $-1$ & $1$ & $0$ & $1$ & $\gamma_{3,2}$ & $0$ \\
$\overline{\Phi}_{3}^{(3)}$ & $0$ & $0$ & $0$ & $0$ & $-1$ & $\gamma_{3,2}^{-1}$ & $2$ \\
$\overline{\Phi}_{4}^{(3)}$ & $0$ & $0$ & $1$ & $-1$ & $1$ & $\gamma_{3,1}^{-1}$ & $2$ \\
$\Phi_{5}^{(3)}$ & $0$ & $1$ & $-1$ & $0$ & $0$ & $1$ & $0$ \\
\hline
\end{tabular}
\caption{
Matter content for the $(1,1)$-tangle diagram of 
the trefoil knot $\knot{3}{1}$ in Figure \ref{fig:3_1_add} 
corresponding to 
the inverse $R$-matrices for \eqref{3_1_block}, 
where $U(1)_{ext}$ and $U(1)_c$ are global symmetries.
}
\label{31_matter}
\end{small}
\end{center}
\end{minipage}
\end{figure}

As the simplest non-trivial example, 
consider a $(1,1)$-tangle diagram of 
the trefoil knot $\knot{3}{1}$ in Figure \ref{fig:3_1_add} 
(Figure \ref{fig:tangle_31}). 
The associated $U(1)^3$ knot-gauge theory $T[\knot{3}{1}]$ has 
the chiral fields in Table \ref{31_matter} and 
Chern-Simons couplings
\begin{align}
\begin{split}
k_{11}=k_{22}=k_1^{\text{g-R}}=k_2^{\text{g-R}}=1,\quad
k_{33}=-\frac32,\quad
k_{12}=k_{21}=-\frac12,\quad 
k_3^{\text{g-R}}=\frac12,
\\
k_{3c}^{\text{g-f}}=\frac32,\quad
k_{3f}^{\text{g-f}}=-3,\quad
k_c^{\text{f-R}}=2,\quad
k_{cf}^{\text{f-f}}=k_{fc}^{\text{f-f}}=3,\quad
k_{ff}^{\text{f-f}}=-6,\quad
k_{f}^{\text{f-R}}=2,
\end{split}
\end{align}
and is considered to be a coupled system of the theories labeled by
\begin{align}
\begin{split}
&
T_1=
T[\overline{R}^{\ 0\ d_{3}}_{d_{1} d_{31}}(\gamma_{1,1}, \gamma_{1,2})],\quad
T_2=
T[\overline{R}^{d_{1} d_{31}}_{d_{2} d_{32}}(\gamma_{2,1}, \gamma_{2,2})],
\\
&
T_3=
T[\overline{R}^{d_{2} d_{32}}_{\ 0\ d_{3}}(\gamma_{3,1}, \gamma_{3,2})],\quad
T_4=T[\mu_{d_3}],
\label{3_1_block}
\end{split}
\end{align}
where $d_{ij}=d_i-d_j$, and 
$\bgamma=(\gamma_{1,1},\gamma_{1,2},\gamma_{2,1},\gamma_{2,2},
\gamma_{3,1},\gamma_{3,2})$ are mass parameters. 
The K-theoretic vortex partition function is given by
\begin{align}
I^{T[\knot{3}{1}]}_{\textrm{vortex}}(\bsig;\bz,\bgamma,q)=
\sum_{d_1, d_2, d_3} z_1^{d_1}z_2^{d_2}z_3^{d_3}
\prod_{i=1}^{4} I_{n,f}^{T_i}(\bsig;q),
\label{k_vpf_3_1_gen}
\end{align}
where 
$\bsig=(\sigma_1, \sigma_2, \sigma_3)$, and 
$\bz=(z_1, z_2, z_3)$ are the exponentiated FI parameters associated with 
the $U(1)_1 \times U(1)_2 \times U(1)_3$ gauge symmetry. 
For the FI parameters $\xi_I=-\textrm{Re}(\log z_I)$, $I=1,2,3$, 
inside $\mathrm{Cone} (Q_1, Q_2, Q_3)$, only the JK residue at 
$\bsig=\bsig^*=(1,1,1)$ finally contributes, 
where $Q_1=(1,0,-1)$, $Q_2=(1,1,-1)$, and $Q_3=(0,1,-1)$ are, respectively, 
the charge vectors of $\Phi_5^{(1)}$, $\Phi_5^{(2)}$, and $\Phi_5^{(3)}$ 
for the $U(1)_1 \times U(1)_2 \times U(1)_3$ gauge symmetry. 
Then, under $\bz \to \bz^*=(1,1,-1)$ following \eqref{invR_limit} and 
$\bgamma \to \bgamma^*$ with $\gamma_{k,\ell}^*=1$, 
the K-theoretic vortex partition function \eqref{k_vpf_3_1_gen} yields 
the normalized colored Jones polynomial of $\knot{3}{1}$:
\begin{align}
& 
I^{T[\knot{3}{1}]}_{\textrm{vortex}}(\sigma^*;\bz^*,\bgamma^*,q)
\nonumber
\\
&=
q^{\frac34 n(n+2)}
\mathop{\sum_{0 \le d_{13},d_{23} \le f}}
\limits_{-d_{13}-d_{23} \le d_3 \le n-f}
\left(R^{-1}\right)^{f\ \ \ \ \ \ d_{3}+f}_{d_{1}+f\ d_{31}+f} 
\left(R^{-1}\right)^{d_{1}+f\ d_{31}+f}_{d_{2}+f\ d_{32}+f}
\left(R^{-1}\right)^{d_{2}+f\ d_{32}+f}_{f\ \ \ \ \ d_{3}+f} \, \mu_{d_3+f}
=
J_n^{\knot{3}{1}}(q)
\nonumber
\\
&=
\mathop{\sum_{0 \le d_{13},d_{23} \le f}}
\limits_{-d_{13}-d_{23} \le d_3 \le n-f}
q^{d_{13}^2+d_{23}^2+(d_{13}+d_{23}-3f+1)d_3+(d_3-d_{13}-d_{23}+3f+1)n
-3f^2+f}
\nonumber
\\
&\quad \times
\frac{\qq{q}{d_{13}+d_3+f} \qq{q}{n-f+d_{13}} 
\qq{q}{d_{23}+d_3+f} \qq{q}{n-f+d_{23}} 
\qq{q}{f} \qq{q}{n-f-d_3}}
{\qq{q}{f-d_{13}} \qq{q}{d_{13}} \qq{q}{n-f-d_{13}-d_3}
\qq{q}{f-d_{23}} \qq{q}{d_{23}} \qq{q}{n-f-d_{23}-d_3} 
\qq{q}{d_{13}+d_{23}+d_3} \qq{q}{d_3+f} \qq{q}{n-f}},
\label{jones_31f}
\end{align}
where note that $\qq{q}{d}^{-1}=0$ for $d \in {\IZ}_{<0}$, and 
we can check that the result does not depend on the constant $f \in \{0,1,\ldots,n\}$.

\subsubsection{Unknot and Reidemeister move I}\label{subsec:unknot}

\begin{figure}[t]
\begin{minipage}{0.35\textwidth}
\begin{center}
\includegraphics[width=35mm]{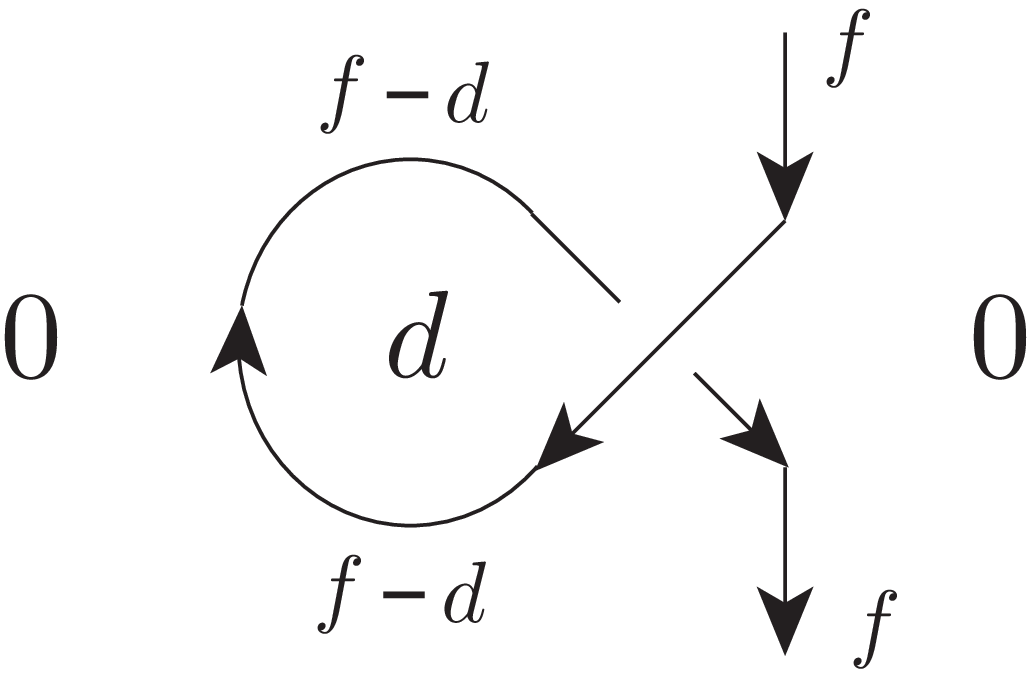}
\caption{$(1,1)$-tangle diagram of the unknot.}
\label{fig:unknot_ex}
\end{center}
\end{minipage}
\hspace{1em}
\begin{minipage}{0.6\textwidth}
\begin{center}
\makeatletter
\def\@captype{table}
\makeatother
\begin{tabular}{|c|c||c|c|c||c|}
\hline
Field & $U(1)$ & $U(1)_{c}$ & $U(1)_{ext}$ & mass & $U(1)_R$ 
\\ 
\hline
$\Phi_{1}$ & $-1$ & $0$ & $1$ & $\gamma_1$ & $0$ \\
$\Phi_{2}$ & $0$ & $1$ & $-1$ & $\gamma_2$ & $0$ \\
$\overline{\Phi}_{3}$ & $-1$ & $-1$ & $1$ & $\gamma_2^{-1}$ & $2$ \\
$\overline{\Phi}_{4}$ & $0$ & $0$ & $-1$ & $\gamma_1^{-1}$ & $2$ \\
$\Phi_{5}$ & $1$ & $0$ & $0$ & $1$ & $0$ \\
\hline
\end{tabular}
\caption{
Matter content for the tangle diagram in Figure \ref{fig:unknot_ex} 
corresponding to 
the $R$-matrix $R^{f-d\, f}_{f-d\, f}$, where 
$U(1)_{ext}$ and $U(1)_c$ are global symmetries.
}
\label{unknot_ex_R_matter}
\end{center}
\end{minipage}
\end{figure}

Consider a $(1,1)$-tangle diagram of the unknot in 
Figure \ref{fig:unknot_ex}. 
Then, the $U(1)$ knot-gauge theory $T[\mathbf{0}]$ with chiral fields in 
Table \ref{unknot_ex_R_matter} 
and Chern-Simons couplings
\begin{align}
k=\frac12,\quad
k_{c}^{\text{g-f}}=\frac12,\quad
k_{f}^{\text{g-f}}=-1,\quad
k^{\text{g-R}}=\frac32,\quad
k_{cf}^{\text{f-f}}=k_{fc}^{\text{f-f}}=-1,\quad
k_{ff}^{\text{f-f}}=2,\quad
k_{f}^{\text{f-R}}=-2,
\end{align}
which is considered to be a coupled system of two theories labeled by
\begin{align}
T_1=T[R^{-d\, 0}_{-d\, 0}(\bgamma)],\quad
T_2=T[\overline{\mu}_{-d}],
\end{align}
is associated, 
where $\bgamma=(\gamma_1, \gamma_2)$ are mass parameters. 
The K-theoretic vortex partition function is given by
\begin{align}
I_{\textrm{vortex}}^{T[\mathbf{0}]}(\sigma;z,\bgamma,q)=
\sum_{d} z^{d}\,
I_{n,f}^{T_1}(\sigma;q)\, I_{n,f}^{T_2}(q).
\end{align}
Here the exponentiated FI parameter $z$ associated with 
the $U(1)$ gauge symmetry is taken as $-\textrm{Re}(\log z)>0$, and then 
the residue at $\sigma=\sigma^*=1$ relevant to $\Phi_5$ is taken.  
As a result, the K-theoretic vortex partition function yields, 
by $z \to z^*=1$ following \eqref{R_limit} and 
$\bgamma \to \bgamma^*=(1,1)$, 
the normalized colored Jones polynomial of unknot:
\begin{align}
I_{\textrm{vortex}}^{T[\mathbf{0}]}(\sigma^*;z^*, \bgamma^*,q)
&=
q^{-\frac14n(n+2)}
\sum_{0 \le d \le f} R^{f-d\, f}_{f-d\, f}\, \mu_{f-d}^{-1}
=
J_n^{\mathbf{0}}(q)
\nonumber
\\
&=
\sum_{0 \le d \le f}
(-1)^{d}\, q^{\frac12 d(d+1) -fd - f(n-f+1)}
\frac{\qq{q}{n-f+d} \qq{q}{f}}
{\qq{q}{f-d} \qq{q}{n-f} \qq{q}{d}}
\nonumber
\\
&
=1.
\label{v_unknot}
\end{align}
Here the last equality is obvious for $f=0$ and, in general, 
follows from the fact that it equals to the one by $n \to n-1$ and $f \to f-1$, 
i.e.,
$$
\sum_{0 \le d \le f-1}
(-1)^{d}\, q^{\frac12 d(d+1) - (f-1)d - (f-1)(n-f+1)}
\frac{\qq{q}{n-f+d} \qq{q}{f-1}}
{\qq{q}{f-d-1} \qq{q}{n-f} \qq{q}{d}},
$$
which can be shown by the following $q$-Pascal relation
\begin{align}
\frac{\qq{q}{m+n}}{\qq{q}{m} \qq{q}{n}}=
\frac{\qq{q}{m+n-1}}{\qq{q}{m-1} \qq{q}{n}} +
q^{m} \frac{\qq{q}{m+n-1}}{\qq{q}{m} \qq{q}{n-1}},
\label{q_pascal}
\end{align}
with $m=f-d$ and $n=d$. 
The result \eqref{v_unknot} is understood by the Reidemeister move I.

\section{Reduced knot-gauge theory}\label{sec:ded_gauge}

In this section, we consider $(1,1)$-tangle diagrams with the external 
incoming (outgoing) constant $f=0$ 
such that the first incoming arc is over crossing and 
the last outgoing arc is under crossing. 
On this basis, the (inverse) $R$-matrices assigned to 
the first and last crossings are extremely degenerated, and
the corresponding gauge theories are simplified. 
As the examples, we describe the trefoil knot $\knot{3}{1}$, 
the figure-eight knot $\knot{4}{1}$ and the 3-twist knot $\knot{5}{2}$.

\subsection{Setup}\label{subsec:setup}

We assume the following setup for a $(1,1)$-tangle diagram of a knot $\cK$
\cite{Yokota:2011, ChoMurakami} 
(see Figure \ref{fig:3_1}, \ref{fig:4_1} and \ref{fig:5_2} for examples):
\begin{itemize}
\item[{\bf 1.}]
The tangle starts from an over crossing arc and ends with an under crossing arc. 

\item[{\bf 2.}]
Trivial external incoming (outgoing) constant $f=0$ is assigned 
at the end points of the tangle diagram.
\end{itemize}
By the last conditions in \eqref{R_charge_c} and \eqref{invR_charge_c}, 
this setup implies that 
the $R$-matrices assigned to the first and last crossings 
have degenerated forms \eqref{deg_R_matrix_3} or \eqref{deg_R_matrix_m3}. 
Instead of introducing chiral fields in 
Section \ref{subsec:gauge_knot_build}, we construct them 
with normalization factors in \eqref{c_norm_jones_knot} just 
as Chern-Simons factors in \eqref{cs_factor_vortex}, 
with non-zero Chern-Simons couplings 
$k_{1c}^{\text{g-f}}=\mp 1/2$, $k_{2c}^{\text{g-f}}=\pm 1/2$ 
and $k_c^{\text{f-R}}=\mp 1$:
\begin{align}
\begin{split}
&
I_{d_{12},n}^{\textrm{CS}_+}(\bsig;q)=
\left(\frac{\sigma_1}{\sigma_2}\right)^{-\frac12 n}
q^{-\frac12 n d_{12} - \frac12 n},
\qquad
I_{d_{12},n}^{\textrm{CS}_-}(\bsig;q)=
\left(\frac{\sigma_1}{\sigma_2}\right)^{\frac12 n}
q^{\frac12 n d_{12} + \frac12 n},
\label{cs_deg_r}
\end{split}
\end{align}
as
\begin{align}
I_{n}^{T[R^{d_{12} \ 0\ }_{\ 0\ d_{12}}]}(q)=
I_{n}^{T[R^{\ 0\ d_{12}}_{d_{12} \ 0\ }]}(q)&=
I_{d_{12},n}^{\textrm{CS}_+}(\bsig;q)
\nonumber\\
&=
q^{-\frac14 n(n+2)}\, R^{d_{12} \ 0\ }_{\ 0\ d_{12}}=
q^{-\frac14 n(n+2)}\, R^{\ 0\ d_{12}}_{d_{12} \ 0\ }
\,,
\label{vortex_pf_R_deg_a}
\\
I_{n}^{T[\overline{R}^{d_{12} \ 0\ }_{\ 0\ d_{12}}]}(q)=
I_{n}^{T[\overline{R}^{\ 0\ d_{12}}_{d_{12} \ 0\ }]}(q)&=
I_{d_{12},n}^{\textrm{CS}_-}(\bsig;q)
\nonumber\\
&=
q^{\frac14 n(n+2)} \left(R^{-1}\right)^{d_{12} \ 0\ }_{\ 0\ d_{12}}=
q^{\frac14 n(n+2)} \left(R^{-1}\right)^{\ 0\ d_{12}}_{d_{12} \ 0\ }
\,,
\label{vortex_pf_R_deg_b}
\end{align}
where $\bsig=(\sigma_1, \sigma_2)$ are the complex scalars for a gauge symmetry 
$U(1)_1 \times U(1)_2$. Here, as in footnote \ref{footnote:p_anomaly}, 
the $U(1)_1$-$U(1)_c$ and $U(1)_2$-$U(1)_c$ parity anomalies 
are shown to be absent for each loop. 


In the following, we will represent the arcs with the trivial constant $0$, like 
the first and last ones, by dashed lines.
If the second crossing after the first over crossing is also 
over crossing, the second arc is also represented by the dashed line, and 
the same applies to the subsequent and last crossings. 
Let $n_v$ be the number of loops in the tangle diagram after removing the dashed lines. 
Then a $U(1)^{n_v}$ gauge theory $T^{\text{red}}[\cK]$ 
is constructed similarly as Section \ref{sec:knot_building}, and 
Proposition \ref{prop:bridge} is also established because 
the charge vectors like $Q_{+}^{(1)}$ and $Q_{-}^{(N_1+N_2)}$ are 
just removed from the original sequences in \eqref{bridge_sq}.
The gauge theory $T^{\text{red}}[\cK]$ is simpler than the knot-gauge theories 
$T[\cK]$ in Section \ref{sec:knot_building}, and we refer to it as 
the reduced knot-gauge theory.

\subsection{Examples}\label{subsec:red_exs}

\subsubsection{Trefoil knot $\knot{3}{1}$}\label{subsec:3_1}

\begin{figure}[t]
\begin{minipage}{0.40\textwidth}
\begin{center}
\includegraphics[width=34mm]{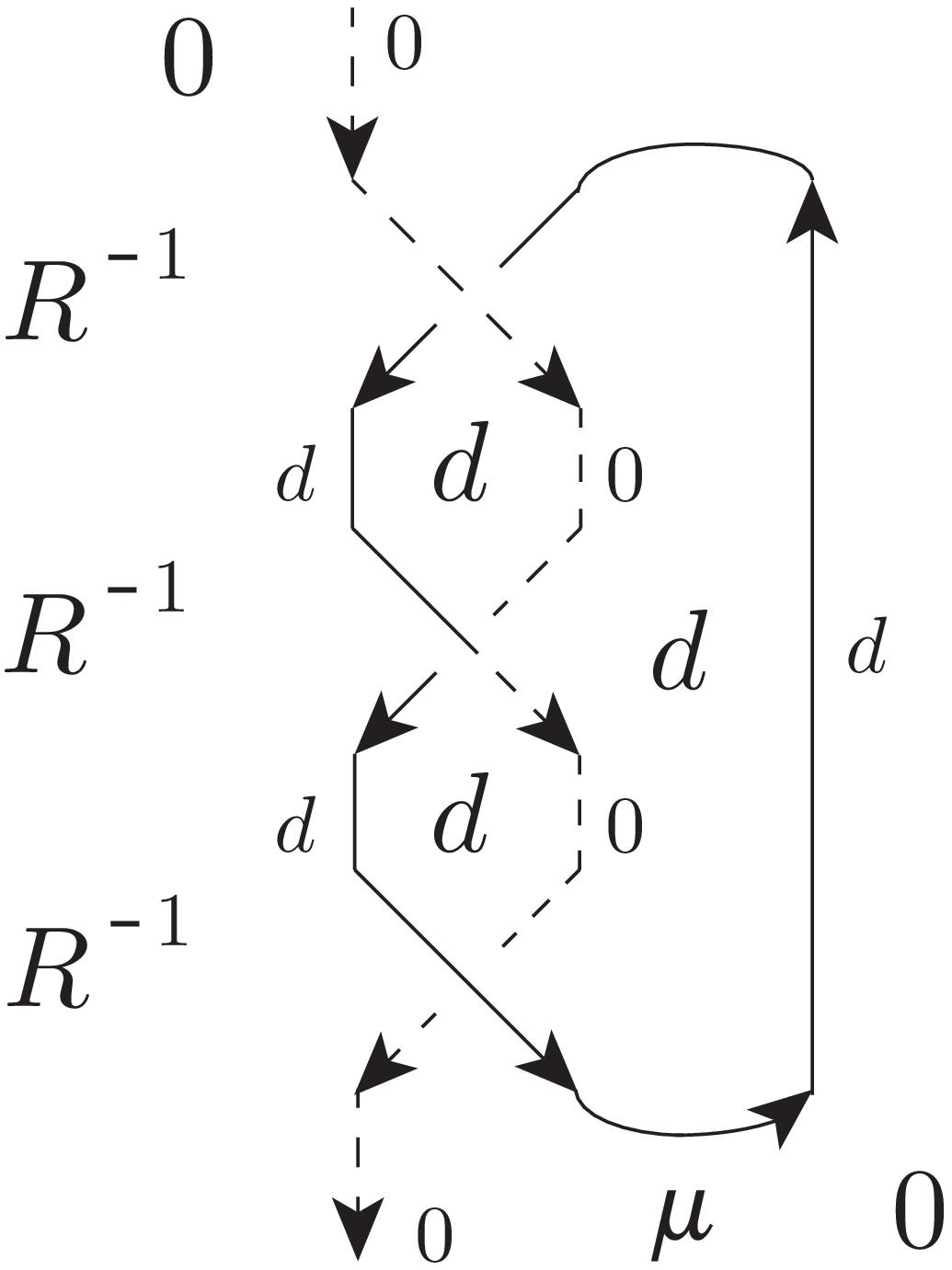}
\caption{$(1,1)$-tangle diagram of the trefoil knot $\knot{3}{1}$, 
where $0 \le d \le n$.}
\label{fig:3_1}
\end{center}
\end{minipage}
\hspace{1.5em}
\begin{minipage}{0.55\textwidth}
\begin{center}
\makeatletter
\def\@captype{table}
\makeatother
\begin{tabular}{|c|c||c|c||c|}
\hline
Field & $U(1)$ & $U(1)_{c}$ & mass & $U(1)_R$ 
\\
\hline
$\Phi_{1}^{(1)}$ & $-1$ & $1$ & $\gamma_1$ & $0$ \\
$\Phi_{2}^{(1)}$ & $0$ & $0$ & $\gamma_2$ & $0$ \\
$\overline{\Phi}_{3}^{(1)}$ & $-1$ & $0$ & $\gamma_2^{-1}$ & $2$ \\
$\overline{\Phi}_{4}^{(1)}$ & $0$ & $-1$ & $\gamma_1^{-1}$ & $2$ \\
$\Phi_{5}^{(1)}$ & $1$ & $0$ & $1$ & $0$ \\
\hline
\end{tabular}
\caption{
Matter content for the tangle diagram in Figure \ref{fig:3_1}, 
which corresponds to the $R$-matrix $\left(R^{-1}\right)^{d\, 0}_{d\, 0}$.
}
\label{31_R_matter}
\end{center}
\end{minipage}
\end{figure}

Consider the reduced $U(1)$ knot-gauge theory $T^{\text{red}}[\knot{3}{1}]$ 
for a $(1,1)$-tangle diagram of the trefoil knot $\knot{3}{1}$ 
in Figure \ref{fig:3_1} (see Section \ref{subsec:gen_3_1} for 
the $U(1)^3$ knot-gauge theory $T[\knot{3}{1}]$). 
The associated chiral fields are in 
Table \ref{31_R_matter} and the associated Chern-Simons couplings are 
\begin{align}
k=-\frac12,\quad
k_{c}^{\text{g-f}}=\frac32,\quad
k^{\text{g-R}}=\frac52,\quad
k_{c}^{\text{f-R}}=2.
\end{align}
The gauge theory is considered to be a coupled system of 
the theories labeled by
\begin{align}
T_1=T[\overline{R}^{0\, d}_{d\, 0}],\quad
T_2=T[\overline{R}^{d\, 0}_{d\, 0}(\bgamma)],\quad
T_3=T[\overline{R}^{d\, 0}_{0\, d}],\quad
T_4=T[\mu_{d}],
\end{align}
where the building block theory $T_2$ has 
mass parameters $\bgamma=(\gamma_1, \gamma_2)$ 
which are taken to be $\gamma_1, \gamma_2 \to 1$ at the end. 
The K-theoretic vortex partition function is given by%
\footnote{In this section, for simplicity we use a notation 
$I_{n}^{T_i}(\bsig;q)=I_{n,f=0}^{T_i}(\bsig;q)$.}
\begin{align}
I^{T^{\text{red}}[\knot{3}{1}]}_{\textrm{vortex}}(\sigma;z,\bgamma,q)=
\sum_{d} z^{d}\,
I_{n}^{T_1}(q)\,
I_{n}^{T_2}(\sigma;q)\,
I_{n}^{T_3}(q)\,
I_{n}^{T_4}(q).
\label{k_vpf_3_1}
\end{align}
Here $z$ is the exponentiated FI parameter for the $U(1)$ gauge symmetry 
that we take as $\xi =-\textrm{Re}(\log z)>0$, and 
the JK residue at $\sigma=\sigma^*=1$ relevant to $\Phi_5^{(1)}$ is taken. 
Then, the K-theoretic vortex partition function \eqref{k_vpf_3_1} yields, 
by $z \to z^*=-1$ following \eqref{invR_limit} and 
$\bgamma \to \bgamma^*=(1,1)$, 
the normalized colored Jones polynomial of $\knot{3}{1}$:
\begin{align}
I^{T^{\text{red}}[\knot{3}{1}]}_{\textrm{vortex}}(\sigma^*;z^*,\bgamma^*,q)
&=
q^{\frac34 n(n+2)}
\sum_{0 \le d \le n}
\left(R^{-1}\right)^{0\, d}_{d\, 0} \left(R^{-1}\right)^{d\, 0}_{d\, 0}
\left(R^{-1}\right)^{d\, 0}_{0\, d} \, \mu_{d}
= J_n^{\knot{3}{1}}(q)
\nonumber
\\
&=
\sum_{0 \le d \le n}
q^{(n+1)d+n}
\frac{\qq{q}{n}}{\qq{q}{n-d}}.
\label{jones_3_1}
\end{align}
It is easy to see that this result agrees with \eqref{jones_31f} with $f=0$.

\subsubsection{Figure-eight knot $\knot{4}{1}$}\label{subsec:4_1}

\begin{figure}[t]
\begin{minipage}{0.41\textwidth}
\begin{center}
\includegraphics[width=40mm]{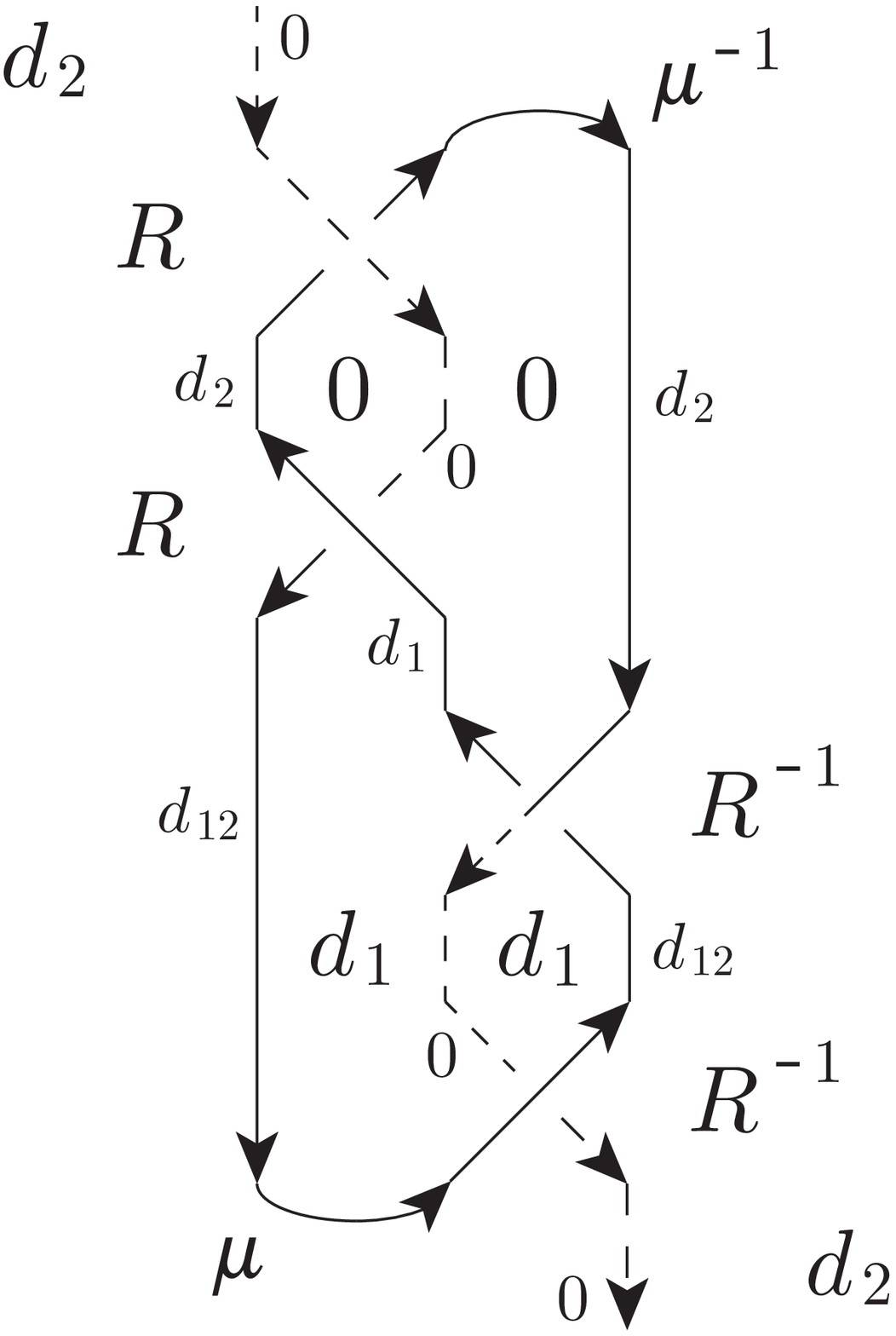}
\caption{$(1,1)$-tangle diagram of the figure-eight knot $\knot{4}{1}$, 
where $0 \le d_2 \le d_1 \le n$.}
\label{fig:4_1}
\end{center}
\end{minipage}
\hspace{1em}
\begin{minipage}{0.56\textwidth}
\begin{center}
\makeatletter
\def\@captype{table}
\makeatother
\begin{tabular}{|c|c|c||c|c||c|}
\hline
Field & $U(1)_{1}$ & $U(1)_{2}$ & $U(1)_{c}$ & mass & $U(1)_R$ 
\\ 
\hline
$\Phi_{1}^{(1)}$ & $0$ & $0$ & $0$ & $\gamma_{1,1}$ & $0$ \\
$\Phi_{2}^{(1)}$ & $-1$ & $0$ & $1$ & $\gamma_{1,2}$ & $0$ \\
$\overline{\Phi}_{3}^{(1)}$ & $0$ & $1$ & $-1$ & $\gamma_{1,2}^{-1}$ & $2$ \\
$\overline{\Phi}_{4}^{(1)}$ & $-1$ & $1$ & $0$ & $\gamma_{1,1}^{-1}$ & $2$ \\
$\Phi_{5}^{(1)}$ & $1$ & $-1$ & $0$ & $1$ & $0$ \\
\hline
%
\hline
$\Phi_{1}^{(2)}$ & $0$ & $-1$ & $1$ & $\gamma_{2,1}$ & $0$ \\
$\Phi_{2}^{(2)}$ & $1$ & $-1$ & $0$ & $\gamma_{2,2}$ & $0$ \\
$\overline{\Phi}_{3}^{(2)}$ & $-1$ & $0$ & $0$ & $\gamma_{2,2}^{-1}$ & $2$ \\
$\overline{\Phi}_{4}^{(2)}$ & $0$ & $0$ & $-1$ & $\gamma_{2,1}^{-1}$ & $2$ \\
$\Phi_{5}^{(2)}$ & $0$ & $1$ & $0$ & 1 & $0$ \\
\hline
\end{tabular}
\caption{
Matter content for the tangle diagram in Figure \ref{fig:4_1}, 
where the first (resp. second) table gives 
the $R$-matrix $R^{0\ d_1}_{d_2 d_{12}}$ 
(resp. $\left(R^{-1}\right)^{d_2 d_{12}}_{d_1\ 0}$).
}
\label{fig8_R_matter}
\end{center}
\end{minipage}
\end{figure}

Consider a $(1,1)$-tangle diagram of 
the figure-eight knot $\knot{4}{1}$ in Figure \ref{fig:4_1}. 
The reduced $U(1)^2$ knot-gauge theory $T^{\text{red}}[\knot{4}{1}]$ has 
the chiral fields in Table \ref{fig8_R_matter} and 
Chern-Simons couplings
\begin{align}
k_{11}=\frac12,\quad
k_{12}=k_{21}=-\frac12,\quad
k_{1c}^{\text{g-f}}=\frac12,\quad
k_{2c}^{\text{g-f}}=-1,\quad
k_1^{\text{g-R}}=\frac32,\quad
k_2^{\text{g-R}}=-3,
\end{align}
and described as a coupled system of the theories labeled by
\begin{align}
\begin{split}
&
T_1=T[R^{d_2 0}_{0\ d_2}],\quad
T_2=T[R^{0\ d_1}_{d_2 d_{12}}(\gamma_{1,1},\gamma_{1,2})],\quad
T_3=T[\overline{R}^{d_2 d_{12}}_{d_1\ 0}(\gamma_{2,1},\gamma_{2,2})],
\\
&
T_4=T[\overline{R}^{d_{12} 0}_{0\ d_{12}}],\quad
T_5=T[\mu_{d_{12}}],\quad
T_6=T[\overline{\mu}_{d_2}].
\end{split}
\end{align}
Here $d_{12}=d_1-d_2$, and we note local deformations of tangle as
\begin{align}
\inc{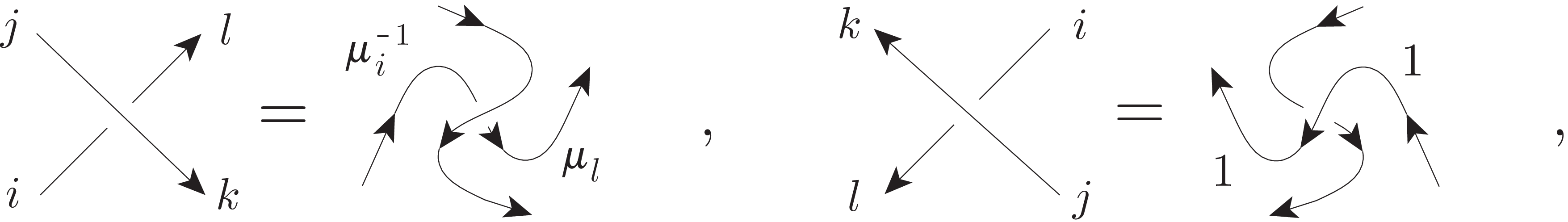}
\label{twist_R}
\end{align}
for the positive crossings and also for the negative crossings as well.
Here mass parameters 
$\bgamma=(\gamma_{1,1},\gamma_{1,2},\gamma_{2,1},\gamma_{2,2})$ 
are introduced for the building block theories $T_2$ and $T_3$. 
The K-theoretic vortex partition function is given by
\begin{align}
I^{T^{\text{red}}[\knot{4}{1}]}_{\textrm{vortex}}(\bsig;\bz,\bgamma,q)=
\sum_{d_1, d_2} z_1^{d_1}z_2^{d_2}
\prod_{i=1}^{6} I_{n}^{T_i}(\bsig;q),
\label{k_vpf_4_1}
\end{align}
where $\bsig=(\sigma_1, \sigma_2)$. 
Here $\bz=(z_1, z_2)$ are the exponentiated FI parameters associated with 
the $U(1)_1\times U(1)_2$ gauge symmetry, and 
for $\xi_i=-\textrm{Re}(\log z_i)$ we take 
$(\xi_1, \xi_2)$ inside $\textrm{Cone}(Q_1, Q_2)$ which 
results in the JK residue at $\bsig=\bsig^*=(1,1)$, where
$Q_1=(1,-1)$ and $Q_2=(0,1)$ are charge vectors of 
$\Phi_5^{(1)}$ and $\Phi_5^{(2)}$, respectively. 
As a result, by $\bz \to \bz^*=(-1,-1)$ 
following \eqref{R_limit} and \eqref{invR_limit}, and 
$\bgamma \to \bgamma^*$ with 
$\gamma_{1,1}^*=\gamma_{1,2}^*=\gamma_{2,1}^*=\gamma_{2,2}^*=1$, 
the K-theoretic vortex partition function \eqref{k_vpf_4_1} yields 
the normalized colored Jones polynomial of $\knot{4}{1}$:
\begin{align}
I^{T^{\text{red}}[\knot{4}{1}]}_{\textrm{vortex}}(\bsig^*;\bz^*,\bgamma^*,q)&=
\sum_{0 \le d_2 \le d_1 \le n}
R^{d_2 0}_{0\ d_2} R^{0\ d_1}_{d_2 d_{12}}
\left(R^{-1}\right)^{d_2 d_{12}}_{d_1\ 0}
\left(R^{-1}\right)^{d_{12} 0}_{0\ d_{12}}
\mu_{d_{12}} \mu^{-1}_{d_2} 
=J_n^{\knot{4}{1}}(q)
\nonumber\\
&=
\sum_{0 \le d_2 \le d_1 \le n}
(-1)^{d_1+d_2} q^{\frac12 d_2(d_2-2d_1-2n-3)+\frac12 d_1(d_1+1)}
\frac{\qq{q}{d_1} \qq{q}{n}}
{\qq{q}{d_2} \qq{q}{d_{12}} \qq{q}{n-d_1}}.
\label{jones_4_1}
\end{align}

\begin{remark}
By a formula \cite[Lemma A.1]{MurakamiMurakami}
\begin{align}
\sum_{0 \le d \le k} (-1)^{d} q^{\frac12 d(d+\ell)}
\frac{\qq{q}{k}}{\qq{q}{d} \qq{q}{k-d}}
=
\qpoch{q^{\frac12 (\ell+1)}}{q}{k}
=
(-1)^{k} q^{\frac12 k(k+\ell)} \qpoch{q^{\frac12 (1-\ell-2k)}}{q}{k},
\label{lemma_murakami2}
\end{align}
which follows from the $q$-Pascal relation \eqref{q_pascal}, 
the colored Jones polynomial \eqref{jones_4_1} is written by 
a single sum as \cite{Murakami:00, Le:03, Habiro:02}
\begin{align}
J_n^{\knot{4}{1}}(q)=
\sum_{0 \le d_1 \le n}
q^{-(n+1)d_1}
\qpoch{q^{n+2}}{q}{d_1} \qpoch{q^{n-d_1+1}}{q}{d_1}.
\end{align}
\end{remark}

\subsubsection{3-twist knot $\knot{5}{2}$}\label{subsec:5_2}

\begin{figure}[t]
\begin{minipage}{0.35\textwidth}
\begin{center}
\includegraphics[width=48mm]{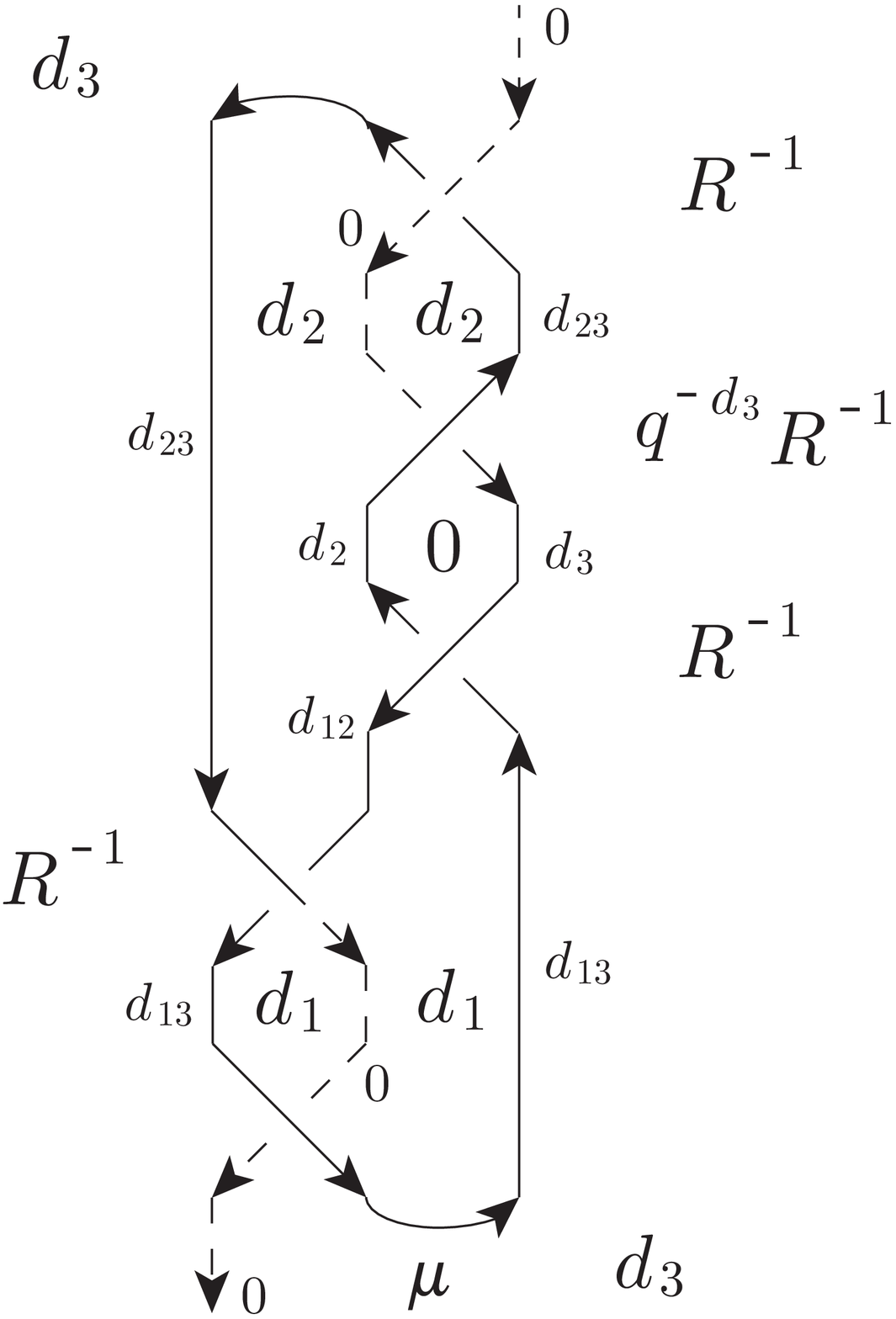}
\caption{$(1,1)$-tangle diagram of the 3-twist knot $\knot{5}{2}$, 
where $0 \le d_{12} \le d_3 \le d_2 \le n$.}
\label{fig:5_2}
\end{center}
\end{minipage}
\hspace{1em}
\begin{minipage}{0.6\textwidth}
\begin{center}
\makeatletter
\def\@captype{table}
\makeatother
\begin{tabular}{|c|c|c|c||c|c||c|}
\hline
Field & $U(1)_{1}$ & $U(1)_{2}$ & $U(1)_{3}$ & $U(1)_{c}$ & mass & $U(1)_R$ 
\\ 
\hline
$\Phi_{1}^{(1)}$ & $0$ & $-1$ & $0$ & $1$ & $\gamma_{1,1}$ & $0$ \\
$\Phi_{2}^{(1)}$ & $0$ & $0$ & $0$ & $0$ & $\gamma_{1,2}$ & $0$ \\
$\overline{\Phi}_{3}^{(1)}$ & $0$ & $0$ & $-1$ & $0$ & $\gamma_{1,2}^{-1}$ & $2$ \\
$\overline{\Phi}_{4}^{(1)}$ & $0$ & $1$ & $-1$ & $-1$ & $\gamma_{1,1}^{-1}$ & $2$ \\
$\Phi_{5}^{(1)}$ & $0$ & $0$ & $1$ & $0$ & $1$ & $0$ \\
\hline
%
\hline
$\Phi_{1}^{(2)}$ & $0$ & $0$ & $-1$ & $1$ & $\gamma_{2,1}$ & $0$ \\
$\Phi_{2}^{(2)}$ & $1$ & $0$ & $-1$ & $0$ & $\gamma_{2,2}$ & $0$ \\
$\overline{\Phi}_{3}^{(2)}$ & $0$ & $-1$ & $0$ & $0$ & $\gamma_{2,2}^{-1}$ & $2$ \\
$\overline{\Phi}_{4}^{(2)}$ & $1$ & $-1$ & $0$ & $-1$ & $\gamma_{2,1}^{-1}$ & $2$ \\
$\Phi_{5}^{(2)}$ & $-1$ & $1$ & $1$ & $0$ & $1$ & $0$ \\
\hline
%
\hline
$\Phi_{1}^{(3)}$ & $0$ & $-1$ & $1$ & $1$ & $\gamma_{3,1}$ & $0$ \\
$\Phi_{2}^{(3)}$ & $1$ & $-1$ & $0$ & $0$ & $\gamma_{3,2}$ & $0$ \\
$\overline{\Phi}_{3}^{(3)}$ & $-1$ & $0$ & $1$ & $0$ & $\gamma_{3,2}^{-1}$ & $2$ \\
$\overline{\Phi}_{4}^{(3)}$ & $0$ & $0$ & $0$ & $-1$ & $\gamma_{3,1}^{-1}$ & $2$ \\
$\Phi_{5}^{(3)}$ & $0$ & $1$ & $-1$ & $0$ & $1$ & $0$ \\
\hline
\end{tabular}
\caption{
Matter content for the tangle diagram in Figure \ref{fig:5_2}, 
where the first, second and third tables correspond to 
the inverse $R$-matrices $\left(R^{-1}\right)^{d_2 0}_{d_3 d_{23}}$, 
$\left(R^{-1}\right)^{d_3 d_{13}}_{d_2 d_{12}}$ and 
$\left(R^{-1}\right)^{d_{23} d_{12}}_{d_{13} 0}$, respectively.
}
\label{52_R_matter}
\end{center}
\end{minipage}
\end{figure}

Consider the reduced $U(1)^3$ 
knot-gauge theory $T^{\text{red}}[\knot{5}{2}]$ 
associated with a $(1,1)$-tangle diagram of 
the 3-twist knot $\knot{5}{2}$ in Figure \ref{fig:5_2}. 
The chiral fields are in Table \ref{52_R_matter}, and 
by noting the local deformations \eqref{twist_R} 
the non-zero Chern-Simons couplings are given by
\begin{align}
\begin{split}
k_{11}=k_{12}=k_{21}=-\frac12,\quad
k_{22}=1,\quad
k_{33}=k_{13}=k_{31}=\frac12,\quad
k_{23}=k_{32}=-1,
\\
k_{1c}^{\text{g-f}}=\frac32,\quad
k_{2c}^{\text{g-f}}=1,\quad
k_{3c}^{\text{g-f}}=-\frac32,\quad
k_1^{\text{g-R}}=\frac32,\quad
k_2^{\text{g-R}}=1,\quad
k_3^{\text{g-R}}=-\frac72,\quad
k_c^{\text{f-R}}=4.
\end{split}
\end{align}
The gauge theory is a coupled system of the theories labeled by
\begin{align}
\begin{split}
&
T_1=T[\overline{R}^{0\ d_{23}}_{d_{23} 0}],\quad
T_2=T[\overline{R}^{d_2 0}_{d_3 d_{23}}(\gamma_{1,1},\gamma_{1,2})],\quad
T_3=T[\overline{R}^{d_3 d_{13}}_{d_2 d_{12}}(\gamma_{2,1},\gamma_{2,2})],
\\
&
T_4=T[\overline{R}^{d_{23} d_{12}}_{d_{13} 0}(\gamma_{3,1},\gamma_{3,2})],\quad
T_5=T[\overline{R}^{d_{13} 0}_{0\ d_{13}}],
\\
&
T_6=T[\mu_{d_{13}}],\quad
T_7=T[\mu_{d_{23}}],\quad
T_8=T[\overline{\mu}_{d_2}],
\end{split}
\end{align}
where $d_{ij}=d_i-d_j$, and 
mass parameters $\bgamma=(\gamma_{1,1},\gamma_{1,2},\gamma_{2,1},\gamma_{2,2},\gamma_{3,1},\gamma_{3,2})$ 
are introduced 
in the building block theories $T_2$, $T_3$ and $T_4$. 
The K-theoretic vortex partition function is given by
\begin{align}
I^{T^{\text{red}}[\knot{5}{2}]}_{\textrm{vortex}}(\bsig;\bz,\bgamma,q)=
\sum_{d_1, d_2, d_3} z_1^{d_1}z_2^{d_2}z_3^{d_3}
\prod_{i=1}^{8} I_{n}^{T_i}(\bsig;q),
\label{k_vpf_5_2}
\end{align}
where $\bsig=(\sigma_1, \sigma_2, \sigma_3)$. 
Here $\bz=(z_1, z_2, z_3)$ are the exponentiated FI parameters for 
the $U(1)_1 \times U(1)_2 \times U(1)_3$ gauge symmetry, 
and for $\xi_i=-\textrm{Re}(\log z_i)$ we take 
$(\xi_1, \xi_2, \xi_3)$ inside $\textrm{Cone}(Q_1, Q_2, Q_3)$ 
so that only the JK residue at $\bsig=\bsig^*=(1,1,1)$ finally contributes, 
where $Q_1=(0,0,1)$, $Q_2=(-1,1,1)$, and $Q_3=(0,1,-1)$ are 
the charge vectors of $\Phi_5^{(1)}$, $\Phi_5^{(2)}$, and $\Phi_5^{(3)}$, 
respectively. 
The K-theoretic vortex partition function \eqref{k_vpf_5_2} then yields, 
by $\bz \to \bz^*=(1,-1,-1)$ following \eqref{invR_limit} 
and $\bgamma \to \bgamma^*$ with 
$\gamma_{k,\ell}^*=1$, 
the normalized colored Jones polynomial of $\knot{5}{2}$:
\begin{align}
&
I^{T^{\text{red}}[\knot{5}{2}]}_{\textrm{vortex}}(\bsig^*;\bz^*,\bgamma^*,q)
\nonumber
\\
&
=
q^{\frac54 n(n+2)}
\sum_{0 \le d_{12} \le d_3 \le d_2 \le n}
q^{-d_3}
\left(R^{-1}\right)^{0\ d_{23}}_{d_{23} 0} 
\left(R^{-1}\right)^{d_2 0}_{d_3 d_{23}}
\left(R^{-1}\right)^{d_3 d_{13}}_{d_2 d_{12}}
\left(R^{-1}\right)^{d_{23} d_{12}}_{d_{13} 0}
\left(R^{-1}\right)^{d_{13} 0}_{0\ d_{13}} \,
\mu_{d_{13}} 
\nonumber
\\
&
=
J_n^{\knot{5}{2}}(q)
\nonumber
\\
&=
\sum_{0 \le d_{12} \le d_3 \le d_2 \le n}
q^{-d_{12}d_2-d_3d_{23} + 2n(d_{12}+d_{23}) + d_{12} + d_2 - 2d_3 + 2 n}
\frac{\qq{q}{d_2} \qq{q}{n-d_{12}} \qq{q}{n}}
{\qq{q}{n-d_2} \qq{q}{n-d_3} \qq{q}{d_3-d_{12}} 
\qq{q}{d_{23}} \qq{q}{d_{12}}}.
\label{jones_5_2}
\end{align}

\subsection*{Acknowledgements}
The authors thank 
Hiroyuki Fuji
and
Masahito Yamazaki 
for valuable discussions.
This work is partially supported by Andrew Sisson Fund and by JST CREST Grant Number JPMJCR14D6 and by JSPS KAKENHI Grant Numbers 
JP17K05243, 
JP17K05414, 
JP21K03240.  

\appendix

\section{$q$-Pochhammer symbol}\label{app:pochhammer}

The $q$-Pochhammer symbol is defined by
\begin{align}
\begin{split}
\qpoch{x}{q}{d}=\frac{\qpoch{x}{q}{\infty}}
{\qpoch{x q^d}{q}{\infty}}=
\begin{cases}
\prod_{k=0}^{d-1} \left(1-x q^k \right)\ \ 
&\textrm{if}\ d \in {\IZ}_{>0},
\\
1
&\textrm{if}\ d =0,
\\
\prod_{k=d}^{-1} \left(1-x q^k \right)^{-1}\ \ 
&\textrm{if}\ d \in {\IZ}_{<0}.
\end{cases}
\label{qPoch}
\end{split}
\end{align}
The $q$-Pochhammer symbol has the following properties:
\begin{align}
&
\qpoch{x}{q}{d}
=\qpoch{q^{d-1}x}{q^{-1}}{d}
=(-x)^d q^{\frac{1}{2}d(d-1)}\qpoch{x^{-1}}{q^{-1}}{d}
=(-x)^d q^{\frac{1}{2}d(d-1)}\qpoch{q^{-d+1}x^{-1}}{q}{d}
\nonumber
\\
&\hspace{2.9em}
=\qpoch{q^d x}{q}{-d}^{-1}
=\qpoch{q^{-1} x}{q^{-1}}{-d}^{-1}
=(-x)^d q^{\frac{1}{2}d(d-1)}\qpoch{q x^{-1}}{q}{-d}^{-1},
\label{qPoch_prop1}
\\
&
\qpoch{x}{q}{d_1+d_2}
=\qpoch{x}{q}{d_1}\qpoch{q^{d_1}x}{q}{d_2}
=\frac{\qpoch{x}{q}{d_1}}{\qpoch{q^{d_1+d_2}x}{q}{-d_2}}
=\frac{\qpoch{q^{d_2}x}{q}{d_1}}{\qpoch{q^{d_2}x}{q}{-d_2}},
\label{qPoch_prop2}
\end{align}
for any $d, d_1, d_2 \in {\IZ}$ and a generic $x \in {\IC}$. 
We also use a notation $\qq{q}{d}=\qpoch{q}{q}{d}$ when $x=q$.



\end{document}